\newif\ifarxiv

\arxivtrue

\ifarxiv 

\documentclass[journal, 10pt]{resources/IEEEtran}

\else 
\documentclass[journal,twoside,web]{ieeecolor}
\usepackage{generic}
\fi

\ifarxiv 
\IEEEoverridecommandlockouts                 
\fi 

\usepackage[utf8]{inputenc}

\usepackage{amsthm, amsmath, amsfonts, amssymb}
\usepackage{resources/custom_commands}
\usepackage{resources/coloredboxes}

\usepackage{enumerate}
\usepackage{graphicx}
\usepackage{mathtools}
\usepackage{thmtools}
\usepackage{thm-restate}
\usepackage{float}
\usepackage{enumerate}
\usepackage{marginnote}
\usepackage[hidelinks]{hyperref}
\usepackage{cleveref}
\usepackage{subfigure}

\usepackage{autonum}

\ifarxiv 
\else 
\usepackage{textcomp}
\def\BibTeX{{\rm B\kern-.05em{\sc i\kern-.025em b}\kern-.08em
    T\kern-.1667em\lower.7ex\hbox{E}\kern-.125emX}}
\fi 

\def\eqdef{\stackrel{\triangle}{=}}

\newcommand{\eqlinebreakshort}{\ensuremath{\nonumber \\ &  \quad}}

\newcommand{\eqbreakshort}{\ensuremath{ \\}}

\newcommand{\degree}{\ensuremath{\mathsf{deg}}}
\newcommand{\diag}{\ensuremath{\mathsf{diag}}}

\newcommand{\sign}{\mathrm{sign}}

\renewcommand\footnotemark{}

\newcommand{\dermat}[1]{\ensuremath{\bJ_{#1}}}

\newcommand{\dist}{\text{dist}}

\newcommand{\logbias}{\ensuremath{\chi}}

\newcommand{\gammaone}{\gamma^*}
\newcommand{\gammatwo}{\gamma^\dagger}

\newcommand{\startindex}{\ensuremath{2}}
\newcommand{\startindexminus}{\ensuremath{1}}
\newcommand{\countrandom}{\ensuremath{t-1}}

\newcommand{\constone}{\ensuremath{c}}

\newcommand{\revise}[1]{\textcolor{black}{#1}}
\newcommand{\revisetwo}[1]{\textcolor{black}{#1}}

\newif\ifsixteen

\sixteenfalse

\newif\ifacc

\accfalse

\begin{document}

\title{
Estimating True Beliefs in Opinion Dynamics with Social Pressure
}

\ifarxiv 
\author{Jennifer Tang, Aviv Adler, Amir Ajorlou, and Ali Jadbabaie
}
\else 
\author{Jennifer Tang, \IEEEmembership{Member, IEEE}, Aviv Adler, \IEEEmembership{Member, IEEE}, Amir Ajorlou, \IEEEmembership{Member, IEEE}, and Ali Jadbabaie, \IEEEmembership{Fellow, IEEE}
}
\fi 

\thanks{The work was supported by ARO MURI W911 NF-19-1-0217 and a Vannevar Bush Fellowship from the Office of the Under Secretary of Defense.}


\thanks{Jennifer Tang, Amir Ajorlou, and Ali Jadbabaie are with the Laboratory for Information and Decision Systems and the Institute for Data, Systems, and Society, Massachusetts Institute of Technology, Cambridge, MA 02139 USA (e-mail:
jstang@mit.edu; ajorlou@mit.edu; jadbabai@mit.edu}

\thanks{Aviv Adler is with Analog Devices Inc, Boston, MA 02110 USA (e-mail:aadler1561@gmail.com)}




\maketitle

\ifarxiv 
\thispagestyle{plain}
\pagestyle{plain}
\fi 


\begin{abstract}
    Social networks often exert social pressure, causing individuals to adapt their expressed opinions to conform to their peers. An agent in such systems can be modeled as having a (true and unchanging) \emph{inherent belief} while broadcasting a \emph{declared opinion} at each time step based on her inherent belief and the past declared opinions of her neighbors. An important question in this setting is \emph{parameter estimation}: how to disentangle the effects of social pressure to estimate inherent beliefs from declared opinions. \revisetwo{This is useful for forecasting when agents' declared opinions are influenced by social pressure while real-world behavior only depends on their inherent beliefs.} To address this, Jadbabaie et al. \cite{socialPressure2021} formulated the \emph{Interacting P\'olya Urn model} of opinion dynamics under social pressure and studied it on complete-graph social networks using an aggregate estimator, and found that their estimator converges to the inherent beliefs unless majority pressure pushes the network to consensus.
    In this work, we study
    this model on arbitrary networks, 
    providing an estimator which converges to the inherent beliefs even in consensus situations. 
    Finally, we bound the convergence rate of our estimator in both consensus and non-consensus scenarios; to get the bound for consensus scenarios (which converge slower than non-consensus) we additionally found how quickly the system converges to consensus.
\end{abstract}

\ifarxiv 
\else 
\begin{IEEEkeywords}
consensus, inverse problem, maximum likelihood estimation, multi-agent system, opinion dynamics
\end{IEEEkeywords} 
\fi 

\ifacc 
\section{Introduction}

Opinion dynamics is the study of how people's opinions evolve over time as they interact with others on social networks. This can provide insights and predictions into how public opinion develops on a variety of political, social, commercial and cultural topics.
For instance, Ancona et al. \cite{Ancona2022} used opinion dynamics models to model the spread of vaccine hesitancy and to develop marketing strategies to help combat it.
%
%
In many opinion dynamics models, there is an assumption that people are truthful in the opinions they share. However, in reality this is not always the case, as people often alter their expressed views to better fit in with their social environment, 
which in turn feeds back into the social environment. The social pressure feedback loop can cause publicly-expressed opinions to become arbitrarily uniform over time \cite{centola_2005}, which can make parameter estimation difficult.  

In this work, we study an \emph{Interacting P\'{o}lya Urn model} for opinion dynamics under social pressure, originating from \cite{socialPressure2021} and developed further in \cite{opiniondynamicsCDC}. This model captures a system of agents with stochastic behaviors who additionally might be untruthful due to a desire to conform to their neighbors. 
This model consists of $n$ agents on a fixed network communicating on an issue with two basic sides, $0$ and $1$. {\color{blue}Each agent has an \emph{inherent} (true and unchanging) belief, which is either $0$ or $1$, and also a \emph{bias parameter} $\gamma$ which indicates to what degree they are willing to share their inherent belief as opposed to conforming to their neighbors. Both the agents' inherent opinions and bias parameters are hidden from their neighbors and outside observers.}
Then the agents communicate their \emph{declared} opinions to their neighbors at discrete time steps: at each step $t = 1, 2, \dots$, all the agents simultaneously declare one of the two opinions (i.e. either $0$ or $1$), which is then observed by their neighbors; the declarations of all the agents at any given step are made at random and independently of each other but with probabilities determined by their inherent belief, bias parameter, and the ratio of the two opinions previously observed by the agent up to the current time. This can represent both scenarios where agents alter their statements (contrary to their actual beliefs) to better fit in with the opinions they have observed from others in the past and scenarios where the agents update their beliefs according to the declared opinions of others, but retain a bias towards their original beliefs.
%

\subsection{Background Literature}

We refer the reader to \cite{socialPressure2021} and \cite{opiniondynamicsCDC} for in-depth discussion of prior work. Here, we discuss some relevant highlights. 

A highly influential opinion dynamics model is the DeGroot model \cite{DeGroot1974}, where agents in a network average their neighbors' opinions in an iterative manner. With this procedure, {\color{blue}on a connected aperiodic graph}, the entire group asymptotically approaches a state where they all share a single opinion, a phenomenon known as consensus.
However in real social networks, consensus is not always reached. To deal with this, other opinion dynamics models were created. Among these is the Friedkin-Johnsen model \cite{friedkin1990}. Each agent in the Friedkin-Johnsen model updates her opinion at each step by averaging her neighbors' opinions (as in the DeGroot model) and then averaging the result with her initial opinion. 

\ifacc
\else
Besides the Interacting P\'{o}lya Urn model in \cite{socialPressure2021}, several other models also include agents that internally retain their initial opinions in some form \cite{Acemoglu2013, gaitonde2020, Ye2019}.
\fi

Ye et al. \cite{Ye2019} study a model in which each agent has both a private and expressed opinion, which evolve differently. Agents' private opinions evolve using the same update as in the Friedkin-Johnsen model, while their public opinions are updated as the average of their own private opinion and the average public opinion of their neighbors. 
Both \cite{centola_2005} and \cite{Ye2019} are very similar to \cite{socialPressure2021}, since agents' expressed opinions may not match their internal beliefs. However, unlike \cite{socialPressure2021}, \cite{Ye2019} assumes opinions are precisely expressed on a continuous interval, which is unrealistic for certain applications. On the other hand \cite{centola_2005} works with binary opinions like \cite{socialPressure2021}, though with a significantly more complex model that includes additional terms and parameters. 

The analysis in \cite{socialPressure2021} is primarily focused on studying whether inherent beliefs are recoverable using an aggregate estimator. This is carried out by establishing the convergence of the dynamics in the network and analyzing the equilibrium state, though the analysis is limited to the complete graph and all agents having the same amount of resistance to social pressure. 
In \cite{opiniondynamicsCDC}, the authors {\color{blue} study the convergence properties of the Interacting P\'{o}lya Urn model introduced in \cite{socialPressure2021} on arbitrary undirected networks, finding that} the proportion of declared opinions of each agent converges almost surely to an equilibrium point in any network configuration. They also determined necessary and sufficient conditions for a network to approach consensus. We note that the definition of consensus used for the Interacting P\'{o}lya Urn model is that all agents will declare a single opinion (all `0' or all `1') with probability tending to $1$. 


\subsection{Contributions}

\ifacc
\else
\fi

\ifacc
\else

\paragraph{Inferring Inherent Beliefs and Bias Parameters}
\fi

In \cite{socialPressure2021}, the authors consider when it is possible to asymptotically determine the inherent beliefs 
of the agents based on their history of declared opinions and those of their neighbors. They study a simplified case in which the social network is an (unweighted) complete graph and all agents have the same, known, degree of bias towards their true beliefs, and consider a specific aggregate estimator which tries to first estimate the proportion of agents with true belief $1$ and then determine which agents those are.
In this setting, they show that the aggregate estimator estimates the proportion of agents with true belief $1$ if and only if the agents do not asymptotically approach consensus (where a large majority causes all agents declare the same opinion with probability approaching $1$). 

{\color{blue} In this work, we consider the problem of estimating the agents' parameters in the general setting presented in \cite{opiniondynamicsCDC} (which analyzed the convergence properties of the Interacting P\'{o}lya Urn model); we also remove the restriction that the social network's graph needs to be undirected}:
\begin{enumerate}
    \item the social network is an arbitrary weighted (and connected) undirected graph, possibly with self-loops;
    \item the agents can have heterogeneous bias parameters, indicating different levels of resistance to social pressure or certainty in their inherent beliefs, which are not known.
\end{enumerate}
{\color{blue} Both the agents' inherent beliefs and bias parameters are unknown and must be inferred from observing the behavior of the network.}
This greatly increases the applicability of the model, as real-life social networks have a variety of different structures and people {\color{blue} have varied reactions to social pressure.}

In this setting, we study the maximum likelihood estimator (MLE), which estimates bias parameters from the history of declared opinions, rather than the aggregate estimator from \cite{socialPressure2021}. We also derive a simplified estimator for inherent beliefs from the MLE, which takes a clean form with a low-dimensional sufficient statistic, consisting of two values which are simple to update at each step. We show that if the history of the agents' declared opinions is known, the MLE almost surely asymptotically converges to the correct inherent beliefs and bias parameters of all the agents in all such networks (even when the network approaches consensus). This resolves the fundamental question posed in \cite{socialPressure2021} of whether such estimation is always possible.

\ifacc
\else
We also examine, asymptotically, how fast the inherent belief estimator converges in the case of consensus (in \Cref{sec::rate_converge_estimator}). 
\fi
\else 
\section{Introduction}

\ifarxiv 
Opinion
\else 
\IEEEPARstart{O}{pinion} 
\fi 
dynamics studies how people's opinions evolve over time as they interact with others on social networks. This can provide insights and predictions about how public opinion develops on a variety of political, social, commercial and cultural topics, as well as guide marketing and political campaign strategies.
For instance, Ancona et al. \cite{Ancona2022} used opinion dynamics models to study the spread of vaccine hesitancy and to develop marketing strategies to help combat it.
%
%
Many common opinion dynamics models assume that people are truthful in the opinions they share. However, in reality this is not always the case, as people often alter their expressed views to better fit in with their social environment, 
which in turn feeds back into the social environment. This social pressure feedback loop can cause publicly-expressed opinions to become arbitrarily uniform over time \cite{centola_2005}, which poses difficulties in estimating and studying the underlying true public opinion.  

In this work, we study an \emph{Interacting P\'{o}lya Urn model} for opinion dynamics under social pressure, originating from \cite{socialPressure2021} and developed further in \cite{opiniondynamicsCDC}, which captures a system of agents with stochastic behaviors who alter their publicly-expressed opinions to conform to their neighbors. 
This model consists of $n$ agents on a fixed network communicating on an issue with two basic sides, denoted $0$ and $1$. Each agent $i$ has an \emph{inherent} (true and unchanging) 
belief $\phi_i$, which is either $0$ or $1$, and a \emph{bias parameter} $\gamma_i$ indicating the ratio of the strength of their attachment to opinions $1$ and $0$, with $\gamma_i = 1$ indicating a neutral position (equal preference for both, though we will assume that no agents are neutral) and $\phi_i = 1 \iff \gamma_i > 1$ (higher preference for $1$ than $0$). 
The agents communicate their \emph{declared opinions} to their neighbors at discrete time steps: at each integer step $t$, each agent $i$ (simultaneously) declares an opinion $\psi_{i,t} \in \{0,1\}$; 
the declarations of the agents at any $t$ are random and independent, and each agent's probability of declaring $1$ is determined by her bias parameter, inherent belief, and the opinions declared by her neighbors in the past. These terms are fully defined in \Cref{sec::parameter-definitions}. 

\revise{This can represent situations where agents alter their statements (contrary to their actual beliefs) to better fit in to their social environment -- for instance, falsely signaling support for a political candidate they actually oppose. This is the primary motivation for the Interacting P\'{o}lya Urn model developed in \cite{socialPressure2021}. \revisetwo{In such a situation, forecasting the behavior of the agents (such as their votes) from social interactions requires separating their inherent beliefs from the pressure their social environment exerts on them.} Furthermore, even though the agents' bias parameters offer a more complete and nuanced picture of their behavior, the true beliefs are often the key factor governing their behavior -- for instance, voting or purchasing patterns.}

\revise{The Interacting P\'{o}lya Urn model can also model situations where the agents honestly update their beliefs according to what they hear from others, but retain a bias towards their original beliefs.}

\subsection{Background Literature}

\revise{Opinion dynamics has been a well-studied topic for many years with a number of mathematical models commonly used to capture specific social phenomena. Two especially important models are the DeGroot model \cite{DeGroot1974}, which seeks to model consensus formation, and the Friedkin-Johnsen model \cite{friedkin1990}, which seeks to model social networks with persistent disagreements. Both models are fundamentally based on iterative averaging: each agent's opinion is represented as a real number, and at each step the agents all simultaneously update their opinion to a (weighted) average of their neighbors' opinions (including their own if the social network has self-loops); Friedkin-Johnsen models persistent disagreement by having each agent also include their own original opinion (which is fixed) in the averaging, thus avoiding consensus.}


\revisetwo{
Many other models for opinions dynamics further build upon the Friedkin-Johnsen model, such as in \cite{semonsen2018opinion, gaitonde2020, Ye2019}.
}
Ye et al. \cite{Ye2019} study a model in which each agent has both a private and expressed opinion, which evolve differently. 
\revise{
Other models that look at opinion dynamics under social pressure include \cite{cheng2019opinion,  liu2021modeling} which consider dynamics similar to that in Hegselmann-Krause \cite{HegselmannKrause2002} and include parameters which measures an agent's resistance to change from their own belief.
\revise{There are also stochastic models which focus on binary opinions including the voter model and some variations \cite{holleyLiggett1975,  Yildiz2013stubborn}. The model in \cite{centola_2005}} features agents who pressure their neighbors into believing one of two opinions, \revisetwo{and studies when the network cascades into total agreement.} 
Other types of P\'{o}lya urn models have also been used to study contagion networks \cite{hayhoe2019,singh2022}. }


%
\revise{The Interacting P\'{o}lya Urn model in \cite{socialPressure2021} was developed to capture the case of agents who lie about their true beliefs in order to fit in with their social environment, i.e. at each step they each have an \emph{inherent (true) belief} (known only to them) and a \emph{declared opinion}, which may not agree. The mechanics of the dynamics was motivated by the well-known Bradley-Terry-Luce discrete choice model, in which an agent's declaration probability is proportional to the number of times she observes an opinion.} 
The analysis in \cite{socialPressure2021} is primarily focused on studying whether inherent beliefs are recoverable from unreliable declared opinions using an aggregate estimator, \revise{in order to estimate outcomes in cases where agents' declared opinions are influenced by social pressure while their actions depend only on their true beliefs (for instance, which candidate they vote for)}. 
This is carried out by establishing the convergence of the dynamics in the network and analyzing the equilibrium state, though the analysis is limited to the complete graph and all agents having the same amount of resistance to social pressure. 
In \cite{opiniondynamicsCDC}, the authors study the convergence properties of the Interacting P\'{o}lya Urn model introduced in \cite{socialPressure2021} on arbitrary undirected networks, finding that the proportion of declared opinions of each agent converges almost surely to an equilibrium point in any network configuration. They also determined necessary and sufficient conditions for a network to approach consensus. We note that the definition of consensus used for the Interacting P\'{o}lya Urn model is that all agents declare a single opinion (all `0' or all `1') with probability tending to $1$ \revisetwo{as the process progresses}. 


\subsection{Contributions}

\ifacc
\else
\fi

\ifacc
\else

\fi

In \cite{socialPressure2021}, the authors consider when it is possible to asymptotically determine the inherent beliefs 
of the agents based on their history of declared opinions and those of their neighbors. They study a simplified case in which the social network is an (unweighted) complete graph and all agents have the same, known, degree of bias towards their inherent beliefs, and consider a specific aggregate estimator which tries to first estimate the proportion of agents with inherent belief $1$ and then determine which agents those are.
In this setting, they show that the aggregate estimator estimates the proportion of agents with inherent belief $1$ if and only if the agents do not asymptotically approach consensus (where a large majority causes all agents declare the same opinion with probability approaching $1$). 

\revise{ In this work, we consider the problem of estimating the agents' hidden parameters (inherent belief and resistance to social pressure) in the general setting presented in \cite{opiniondynamicsCDC} where}
\begin{enumerate}
    \item the social network is an arbitrary weighted (and connected) undirected graph, possibly with self-loops;
    \item the agents can have heterogeneous bias parameters, indicating different levels of resistance to social pressure or certainty in their inherent beliefs, which are not known. \revise{This addresses a limitation listed in \cite{socialPressure2021}.}
\end{enumerate}
\revise{(The focus of \cite{opiniondynamicsCDC} was analyzing the asymptotic dynamics of the Interacting P\'{o}lya Urn model and thus \cite{opiniondynamicsCDC} does not cover estimating inherent beliefs or bias parameters.)
}
 
 Both the agents' inherent beliefs and bias parameters are unknown and must be inferred from observing the behavior of the network.
This greatly increases the applicability of the model, as real-life social networks have a variety of different structures and people have varied reactions to social pressure.

In this setting, we study the maximum likelihood estimator (MLE), which estimates bias parameters from the history of declared opinions, rather than the aggregate estimator from \cite{socialPressure2021}. We also derive a simplified estimator for inherent beliefs from the MLE, which takes a clean form with a low-dimensional sufficient statistic, consisting of two values which are simple to update at each step. We show that if the history of the agents' declared opinions is known, the MLE and the inherent belief estimator almost surely asymptotically converges to the correct inherent beliefs and bias parameters of all the agents in all such networks (even when the network approaches consensus). This resolves the fundamental question posed in \cite{socialPressure2021} of whether such estimation is always possible.

We also show bounds on the convergence rate of the inherent beliefs estimator. These bounds are slower when the system approaches consensus, reflecting the loss of information in the declared opinions. \revise{We show that the convergence rate of the estimator in the case of consensus depends on the network structure through the largest eigenvalue of the normalized adjacency matrix multiplied by the bias parameters.} 
\ifarxiv
\else
\revisetwo{Due to space constraints, some of the detailed proofs for the convergence rate results are not included and can be found in \cite{opinionDynamicsPart2}. Proof sketches are included instead.} 
\fi

\fi 

\section{Model Description}

We use the model from \cite{opiniondynamicsCDC} which is a generalization of the model from \cite{socialPressure2021}. 
\ifacc
\else
Some changes from \cite{socialPressure2021} include the addition that each edge in the network has a (nonnegative) weight denoting how much the two agents' declared opinions influence each other and the use of bias parameters instead of honesty parameters (which are different, but mathematically equivalent, representations of the same process). 
\fi

\subsection{Graph Notation}

Let graph $G = (V,E)$ (\revise{
undirected
} and including self-loops) be a network of $n$ agents (corresponding to the vertices) 
labeled $i = 1, 2, \dots, n$. 
For each edge $(i,j) \in E$, there is a weight $a_{i,j} \geq 0$, where by convention $a_{i,j} = 0$ if $(i,j) \not \in E$. 
The matrix of these weights is $\bA \in \bbR^{n \times n}$. 
The weighted degree of vertex $i$ is $\degree(i) = \sum_j a_{i,j}$. 
\ifacc
\else
The vector of degrees of all agents is
\ifarxiv 
\begin{align}
    \bd \eqdef [\degree(1), \degree(2), \dots, \degree(n)]
\end{align} 
\else 
$\bd \eqdef [\degree(1), \degree(2), \dots, \degree(n)]$
\fi 
and its diagonalization is $\bD = \diag(\bd)$, i.e. the diagonal matrix of the degrees. Let the \emph{normalized adjacency matrix} be 
\ifacc
$    \bW = \bD^{-1} \bA \,.\label{eq::graph_matrix}
$
\else
$    \bW = \bD^{-1} \bA  \label{eq::graph_matrix}
$
whose entries are $w_{i,j}$.
\fi
We assume that $G$ is connected so $\bW$ is irreducible. We denote the largest eigenvalue of a matrix by $\lambda_{\max}(\cdot)$ (the matrices we use this with have real eigenvalues). Let $\bbI\{ \cdot\}$ be the indicator function.
\fi
%
Finally, we denote an all-$0$ vector as $\bzero$ and an all-$1$ vector as $\bone$.

\subsection{Inherent Beliefs and Declared Opinions} \label{sec::parameter-definitions}

We define the Interacting P\'{o}lya Urn model of opinion dynamics under social pressure by defining the key parameters governing the behavior of the agents and their relationship to each other. The basic concept of the model is: each agent $i$ declares at each (integer) step $t$ an opinion $\psi_{i,t} \in \{0,1\}$; in expectation, agent $i$ imitates the (weighted) average opinion they have observed declared by their neighbors (including themselves via self-loops), but biased by an internal \emph{bias parameter $\gamma_i > 0$}. The value of $\gamma_i$ denotes how an observation of a neighbor declaring $1$ is weighted compared to the same neighbor declaring $0$, e.g. $\gamma_i = 2$ denotes that each observation of neighbor $j$ declaring $1$ counts twice as much as when they declare $0$, while $\gamma_i = 1/2$ denotes the converse.\footnote{The honesty parameter in \cite{socialPressure2021} is equivalent to the bias towards the agent's true belief, i.e. a honesty parameter of $\gamma$ with a inherent belief of $0$ corresponds to a bias parameter of $1/\gamma$.}

Then, the inherent belief of agent $i$ is the opinion they are biased toward:
\begin{align}
    \phi_i = \begin{cases} 1 &\text{if } \gamma_i > 1 \\ 0 &\text{if } \gamma_i < 1 \end{cases}
\end{align}
If $\gamma_i = 1$ then the agent is \emph{unbiased} and is considered to not have an inherent belief; since the goal is to estimate the inherent beliefs of the agents, for the remainder of this work we assume that the agent under consideration is not unbiased.
%
%
\revisetwo{The inherent belief and bias parameter of agent $i$ are assumed to be fixed and hidden from observers and other agents.}

To formally state the model, let $b_i^0, b_i^1 >0$ be the \emph{initialization} of agent $i$'s declared opinions, where $b_i^0 + b_i^1 = 1$. Then we define the \emph{declared proportion} of $0$'s (or $1$'s) declared by agent $i$ up to time $t \in \mathbb{Z}_{+}$ as:
\begin{align}
    \beta_{i}^0(t) &= \frac{b_{i}^0}{t} + \frac{1}{t}\sum_{\tau = \startindex}^t (1-\psi_{i,\tau}) \label{eq:def_of_beta_0}
\\
    \beta_{i}^1(t) &=\frac{b_{i}^1}{t} +\frac{1}{t}\sum_{\tau = \startindex}^t \psi_{i,\tau}\,. \label{eq:def_of_beta_1}
\end{align}
Then \revise{$\beta_{i}^0(t), \beta_{i}^1(t) \in (0,1)$ and} $\beta_{i}^0(t) + \beta_{i}^1(t) = 1$; thus to specify these values it is sufficient to specify just $\beta_i(t) \eqdef \beta_i^1(t)$ (i.e. the time-averaged declarations of agent $i$'s declared opinions with initial conditions). 

For any agent $i$ we also denote the total (weighted) proportion of opinions $0$ and $1$ she has observed by time $t$ from her neighbors (including herself via self loop) as
\begin{align}
    \mu^0_i(t) &= \frac{1}{\degree(i)}\sum_{j = 1}^n a_{i,j} \beta_{j}^0(t) = \sum_{j = 1}^n w_{i,j} \beta_{j}^0(t)\\
    \text{ and } \mu^1_i(t) &= \frac{1}{\degree(i)} \sum_{j=1}^n  
 a_{i,j}  \beta_{j}^1(t) = \sum_{j=1}^n  
 w_{i,j}  \beta_{j}^1(t)\,;
\end{align}
as before, $\mu^0_i(t) + \mu^1_i(t) = 1$ by definition so it suffices to specify $\mu_i(t) \eqdef \mu_i^1(t)$. This corresponds to the social environment that agent $i$ finds herself in at time $t$.

Then, at time $t+1$, each agent $i$ will (independently) declare an opinion $\psi_{i,t+1}$ where 
\begin{align}\label{eq::declared_opinion_prob_2}
    \psi_{i,t+1} \eqdef \left\{\begin{array}{lll}
    0 & \text{ with probability } p_i(t) = f(\mu_i(t), \gamma_i)  \\
    1 & \text{ with probability } 1- f(\mu_i(t), \gamma_i) 
    \end{array}\right.\,.
\end{align}
and
\begin{align}
    f(\mu_i(t), \gamma_i) \eqdef \frac{\gamma_i \mu_i(t)}{\gamma_i \mu_i(t) + (1-\mu_i(t))} \,.
\end{align}
The values of $\beta_i(t+1)$ and $\mu_i(t+1)$ for all $i$ are updated according to the declared opinions at time $t$ and the values of $\beta_i(t)$ and $\mu_i(t)$.
Since $\mu_i(t) = \mu_i^1(t)$ and $1-\mu_i(t) = \mu_i^0(t)$, this corresponds to weighting each observation of opinion $1$ as $\gamma_i$ times an equivalent observation of opinion $0$.

We denote as $\cH_t$ the \emph{history} of the network up to time $t$ (consisting of all declared opinions, including initializations, and thus can be used to compute $\beta_i(\tau), \mu_i(\tau)$ for $\tau \leq t$). 
The sequence $\cH_0 \subseteq \cH_1 \subseteq \dots$ \revisetwo{is also the notation we use for the filtration on which we can base the stochastic process. The random variables which generate the $\sigma$-algebras in this filtration are the declared opinions of all agents.}
We also denote the vectors of $\beta_i(t), \mu_i(t)$ over agents $i$ as 
\begin{align}
    \bmu(t) &\eqdef \left[\mu_1(t),...\mu_n(t)\right]^{\ltop}
\text{ and }
    \bbeta(t) \eqdef \left[\beta_1(t),...\beta_n(t)\right]^{\ltop}\,.
\end{align}

In \cite{opiniondynamicsCDC}, it was shown that these dynamics must approach some equilibrium point satisfying
\begin{align}\label{eq::equilibrium_beta}
    \beta_i = f(\mu_i, \gamma_i) ~\text{ for all $i$}
\end{align}
as $t \to \infty$ (with probability $1$). In this work, we consider the following estimation problem (which was considered in \cite{socialPressure2021} for a more restricted model on complete graphs): given the history $\cH_t$ up to time $t$, can we estimate $\gamma_i, \phi_i$ for all agents $i$ in the limit as $t \to \infty$?\footnote{While we assume for simplicity that $b_i^0, b_i^1$ are known to the estimator, \revisetwo{this is not necessary for estimation in the long term as these} terms become negligible in the limit as $t \to \infty$.}


\subsection{Consensus}

An important term for this work is \emph{consensus}, which needs to be defined appropriately for our stochastic system.

\begin{definition}\label{def::consensus}
\emph{Consensus} is approached if 
\begin{align}
    \bbeta(t) \to \bone \text{ or } \bbeta(t) \to \bzero\quad \text{as}\quad t\to\infty\,.
\end{align}
\end{definition}

Since $\beta_i(t)$ represents the fraction \revise{(time-average)} of agent $i$'s declared opinions which are $1$, consensus is approached when this ratio goes to $0$ or $1$. Let the diagonal matrix with $\bgamma$ along the diagonal be
$\bGamma = \diag(\bgamma)\label{eq::diag_gamma_matrix}$
and
let $\dermat{\bone} = \bGamma^{-1} \bW$ and $\dermat{\bzero} = \bGamma \bW$. 

In \cite[Therorem 3]{opiniondynamicsCDC}, it is shown that 
consensus $\bbeta(t) \to \bone$ occurs when $\lambda_{\max}(\dermat{\bone}) \leq 1$ and $\bbeta(t) \to \bzero$ occurs when $\lambda_{\max}(\dermat{\bzero}) \leq 1$.
 Approaching consensus is important for the parameter estimation problem we consider in this work because it represents a major obstacle to solving the estimation problem, as it is an uninformative equilibrium. 

\iftrue 
\subsection{Intuition for the Interacting P\'{o}lya Urn Model}
\revise{
%
\revisetwo{In this section, we give another (intuitive) view of the Interacting P\'{o}lya Urn model.}
(For clearer intuition, we consider unweighted graphs, i.e. where $a_{i,j} = 0$ or $1$.) Typically, urn models start with some composition of balls of different colors in an urn. At each step, a ball is drawn at random from the urn and additional balls are added based on the drawn ball according to some urn functions \revisetwo{(thus affecting future draws)}. In the Interacting P\'{o}lya Urn model, 
\revisetwo{at each step, agent $i$ puts balls corresponding to her \emph{neighbors'} declared opinions into her urn}
(the proportion of balls labeled $1$ in agent $i$'s urn is given by $\mu_i(t)$). 
Then, when agent $i$ declares an opinion, it is modeled by the following: she draws a ball from her urn and declares the corresponding opinion; each ball corresponding with opinion $1$ is $\gamma_i$ times as likely to be drawn as one with opinion $0$ (\revisetwo{as given by function $f$}). 
Note that if $\gamma_i = 1$ then agent $i$ is simply (stochastically) mimicking the opinions her neighbors have declared in the past (plus her initial state, which becomes asymptotically negligible). We remark that the bias parameter is similar to the initial opinions in the Friedkin-Johnsen model \cite{friedkin1990} since they both are fixed parameters that influence all steps; 
however, note that there is a significant difference as the bias parameter can be overwhelmed over time by social pressure, thus leading to consensus.
}

\fi 

\subsection{Organization of results}

\revise{
In \Cref{sec::inferring_inherent} we will introduce and define the inherent opinion and bias parameter estimators we will analyze; in \Cref{sec::consistency-of-estimators} we show that the estimators based on maximum likelihood are \emph{consistent}, i.e. almost surely they converge to the correct result; 
\revisetwo{
and in \Cref{sec::rate_converge_estimator} we will study the convergence rates of these estimators. In \Cref{sec::rate_converge_consensus}, in order to compute the convergence rate of the estimator in the case where the network approaches consensus, we determine a convergence rate result for the declare opinions.}
}


\revise{
\begin{remark}
Our estimator consistency results from \Cref{sec::consistency-of-estimators} also apply to the Interacting P\'{o}lya Urn model on \emph{directed} graphs, since the results are based only on the local environment of an individual agent. However, the convergence rate results in \Cref{sec::rate_converge_estimator,sec::rate_converge_consensus} depend on the convergence rate of the network to consensus, which has only been shown in the case of social networks on \emph{undirected} graphs \revisetwo{(see \cite{opiniondynamicsCDC})}; hence, those results do not generalize to directed graphs.
\end{remark}
}


\section{Estimators for Inferring Inherent Beliefs and Bias Parameters} 

\label{sec::inferring_inherent}

One of the key questions in \cite{socialPressure2021} is whether it is possible to infer the inherent beliefs of agents from the history of declared opinions. The authors of \cite{socialPressure2021} studied the Interacting P\'{o}lya Urn model on the complete graph using an aggregate estimator which keeps track of the fraction of declared opinions of all agents throughout time, and showed that this estimator may not converge to the inherent beliefs of all agents if they approach consensus. Consensus presents difficulties for estimators since asymptotically all agents approach the same behavior regardless of their inherent beliefs.

However, we show that estimators based on maximum likelihood estimation (MLE) almost surely infer the inherent belief of any agent $i$ in the limit, even when consensus is approached. 
This fact is connected to \cite[Lemma 2]{opinionDynamicsPart1} -- each agent declares both opinions infinitely often, yielding sufficient information to determine inherent beliefs over time.

Additionally, unlike \cite{socialPressure2021}, our formulation also allows agents to have different bias parameters. Thus, it is natural to ask how to estimate the bias parameter 
of any agent. Intuitively, after enough time has passed, the values of $\mu_i(t)$ and $\beta_i(t)$ will converge to values close to the equilibrium point. In such a case, we can use \eqref{eq::equilibrium_beta} to estimate the bias parameter $\gamma_i$ and inherent belief $\phi_i$ with
\begin{align}\label{eq::estimator_eq_point}
\widehat{\gamma}^{eq}_i(t) &= \frac{\beta_i(t)}{1 - \beta_i(t)} \frac{1 - \mu_i(t)}{\mu_i(t)}
\\
\widehat{\phi}^{eq}_i(t) &= \bbI\{\beta_i(t) < \mu_i(t) \}\label{eq::estimator_eq_point_inherent}
\end{align}
These estimators are asymptotically consistent, i.e.
\begin{align}
\lim_{t \to \infty} \widehat{\gamma}^{eq}_i(t) &= \gamma_i
\text{ and }
\lim_{t \to \infty} \widehat{\phi}^{eq}_i(t) = \phi_i
\end{align}
when the dynamics converge to an interior equilibrium point \revisetwo{(which is when consensus does not occur)}. 

However, plugging the equilibrium values into \eqref{eq::estimator_eq_point} is not well-defined if $\beta_i(t) $ and $\mu_i(t) $ both converge to either $0$ or $1$ for all $i$, i.e. 
when consensus is approached.
This shows that more careful analysis needs to be done in order to estimate the bias parameters and inherent beliefs in all circumstances. \revise{Simulations indicate that on the networks tested, the estimator in \eqref{eq::estimator_eq_point_inherent} tends to converge to the correct inherent belief, however, the number of time steps needed for the estimator to be correct with high probability can be very large in some situations.
}


\subsection{Definition of Estimators}

We assume at time $t$ the estimator has at its disposal the history of agent $i$ and agent $i$'s neighbors' declarations up to and including time $t$ (recall we denote this as $\cH_{t}$). Given $\cH_{t-1}$, we can compute exactly the values of 
\begin{align}
p_i(t) = \bbP\left[\psi_{i, t} = 1|\cH_{t-1} \right] = f(\mu_i(t-1), \gamma_i)\,.
\end{align}
Note that in general $\bbP\left[\psi_{i, t} = 1 \right]$ is a random variable dependent on $\cH_{t-1}$, while $\bbP\left[\psi_{i, t} = 1|\cH_{t-1} \right]$ is constant.

Our estimator to predict $\gamma_i$ is based on the maximum log-likelihood estimator:

\begin{definition}
The \emph{single-step negative log-likelihood} for a given agent $i$ at time $t > 1$ and parameter $\gamma$ is
 \begin{align}
\ell_i(\gamma, t) &\eqdef -\bigg(\bbI\{\psi_{i,t} = 1\} \log (f(\mu_i(t-1), \gamma)) 
\eqlinebreakshort
+ \bbI\{\psi_{i,t} = 0\} \log (1-f(\mu_i(t-1), \gamma))\bigg)
\end{align}
The negative log-likelihood for a given agent $i$ at time $t$ and parameter $\gamma \in (0, \infty)$ is
\begin{align}
    L_i(\gamma, t) \eqdef \sum_{\tau = \startindex}^t \ell_i(\gamma, \tau)\,.
\end{align}
\end{definition}
\ifarxiv
Here, $\gamma_i$ is the actual bias parameter of agent $i$, whereas $\gamma$ represents a proposed value whose loss we are measuring.
\fi
The MLE for bias parameter $\gamma_i$ gives the value of $\gamma$ that maximizes the likelihood of agent $i$'s declarations, which also minimizes the negative log-likelihood.

\begin{definition}[Estimator for Bias Parameter] \label{def::estimator_bias}
The \emph{maximum likelihood estimator (MLE)} for the bias parameter $\gamma$ at time $t$ is given by
\begin{align}
    \widehat{\gamma}_i(t) \eqdef \arg \min_{\gamma} L_i(\gamma, t)\,.\label{eq::hat_gamma_mle}
\end{align}
\end{definition}
 
Since the inherent belief of an agent is defined as whether the bias parameter is greater than or less than $1$,
given the MLE estimator, we can always predict the inherent belief of agent $i$ by taking 
$
\sign(\log(\widehat{\gamma}_i(t)))
$.

However, if we assume that $\gamma_i \neq 1$, and are only interested estimating the inherent beliefs, this reduces to a simpler form. 
Let $\bar \beta_i(t) = \frac{1}{t-1}\sum_{\tau = \startindex}^t \bbI[\psi_{i,\tau} = 1]$, which a similar quantity to $\beta_i(t)$ except that the arbitrary initial conditions are not included. 
\ifarxiv
(If $t$ is large, then difference between $\beta_i(t)$ and $\bar \beta_i(t)$ is negligible.)
\fi
\begin{definition}[Inherent Belief Estimator]\label{def::estimator_inherent}
    Let 
    \begin{align}
    \widehat{\phi}_{i}(t) &= \frac{1}{2}\sign\left((\countrandom) \bar \beta_i(t) - \left(\sum_{\tau = \startindexminus}^{t-1} \mu_{i}(\tau)\right) \right) + \frac{1}{2}\,.
    \label{eq::steve}
    \end{align}
\end{definition}
Multiplying by $1/2$ and adding $1/2$ maps the output of $\sign(\cdot)$ to $0$ and $1$.  Fundamentally, this estimator requires only comparing 
\iffalse 
$
\bar \beta_i(t) > \frac{1}{\countrandom}\sum_{\tau = \startindexminus}^{t-1} \mu_i(\tau)\,.
$
\else 
\begin{align}
\bar \beta_i(t) \gtrless \frac{1}{\countrandom}\sum_{\tau = \startindexminus}^{t-1} \mu_i(\tau)\,.
\end{align}
\fi 
Note that  $\widehat{\phi}_{i}(t)$ does not depend on knowing the bias parameter, as it only assumes that $\gamma \neq 1$, and the estimator is simple to compute as it only requires the aggregate count of an agent's declarations and her neighborhood's declarations. 

Intuitively, this compares agent $i$'s actual declarations against its expected declarations if $\gamma_i = 1$ (i.e. if the agent were unbiased); however, the consistency of this estimator is derived from that of the MLE for the bias parameter given in \Cref{def::estimator_bias}. 
We show this derivation in \Cref{sec::derive_inherent_estimator}.


%



\label{sec::empirical_estimators}


\revise{
In numerical simulations where the network does not approach consensus, the estimators $\widehat \phi_i(t)$ and $\widehat \phi_i^{eq}(t)$ perform similarly. Depending on the agent, either estimator may converge to the correct prediction faster. In networks where consensus is approached, the simulations suggest that both estimators will also eventually converge to the correct inherent belief with probability one, however, in certain networks, estimator $\widehat \phi_i^{eq}(t)$ may take significantly longer time to converge to all accurate predictions, as shown in \Cref{fig::estimator_convergence}.
}

\begin{figure}
\centering
\includegraphics[scale = .5]{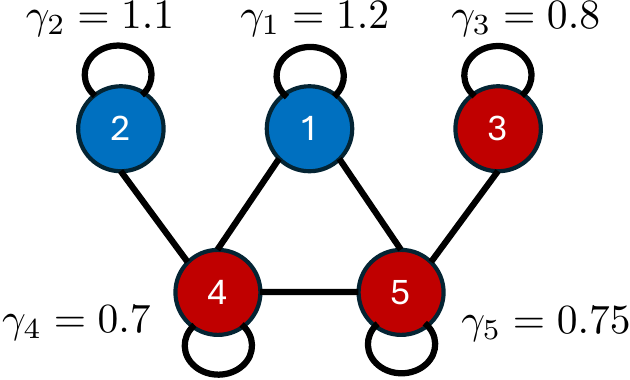}
\caption{\label{fig::5-node-graph} \revise{Example network used in numerical simulations. Blue nodes have inherent belief $\phi_i = 1$ (i.e. $\gamma_i > 1$) and red have inherent belief $\phi_i = 0$ (i.e. $\gamma_i < 1$). \revisetwo{All edges have weight $1$.} As per \cite[Theorem 3]{opiniondynamicsCDC}, this network approaches a consensus of $\bzero$.}}
\end{figure}

\begin{figure}
\centering
\includegraphics[scale = .6]{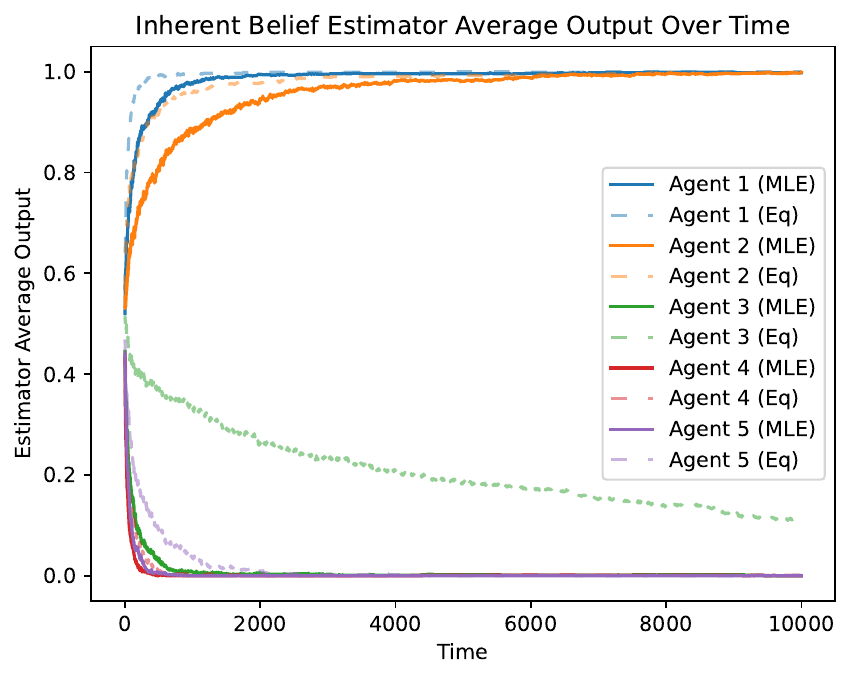}
\caption{\label{fig::estimator_convergence}\revise{Simulation for the network from \Cref{fig::5-node-graph}, comparing estimators for inherent beliefs. `MLE' is the estimator \eqref{eq::steve} (solid) and `Eq' is the estimator \eqref{eq::estimator_eq_point_inherent} (dashed). $1000$ instances of the network were run; each line corresponds to the average prediction of the estimator at the given time over these instances.
Note that for certain agents, the average of `Eq' converges faster than `MLE'; however, when given an agent in general agreement with its neighbors, such as Agent 3, `Eq' can converge extremely slowly (dashed green line).
}
}
\end{figure}

\section{Consistency of Estimators} \label{sec::consistency-of-estimators}

\revise{In this section we show that the MLE based estimators for the bias parameter and inherent belief almost surely converge to the correct value. The key to showing this result is to first show that when two bias parameters, $\gammaone$ and $\gammatwo$, are compared using log-likelihood ratios, in the long run, the true bias parameter (WLOG let this be $\gammaone$) will have the higher likelihood. 
First, in \Cref{sec::consistency_preliminaries} we will examine preliminaries regarding log-likelihoods. Then, in \Cref{sec::consistency_martingale} we show that the difference of the log-likelihood ratios for $\gammaone$ and $\gammatwo$ summed over time is a submartingale whose predictable quadratic variation is bounded by (a multiple of) its predictable expected value. This property allows us, in \Cref{sec::freedman}, to apply Freedman's inequality to get a concentration bound on the process which can be used to show that after some finite time $\gammaone$ will always have a higher likelihood than $\gammatwo$. Then in \Cref{sec::consistency_final}, we show that how comparing just two bias parameters can be extended to show consistency of the bias estimator, and by extension the inherent belief estimator.}

\subsection{Preliminaries}

\label{sec::consistency_preliminaries}

\subsubsection{Bounds on $\mu_i(t)$}

\label{sec::bounds_on_mu}

\newcommand{\mindeg}{\ensuremath{\kappa}}

\begin{lemma} \label{lem::bounds_on_mu}
    Letting $\mindeg \eqdef \min_i ( \min(b_i^0, b_i^1)) > 0$, for any agent $i$ and time $t$, 
    \begin{align}
        \mu_i(t) \in \Big[\frac{\mindeg}{t}, 1-\frac{\mindeg}{t}\Big] \,.
    \end{align}
\end{lemma}

\begin{proof}
    This follows since by definition $b_i^0, b_i^1 \geq \mindeg$ for any $i$; thus by equations \eqref{eq:def_of_beta_0}, \eqref{eq:def_of_beta_1} we know that $\beta_i^0(t), \beta_i^1(t) \geq \mindeg/t$ so $\beta_i(t) = \beta_i^1(t) = 1 - \beta_i^0(t)$ satisfies $\beta_i(t) \in [\frac{\mindeg}{t}, 1 - \frac{\mindeg}{t}]$. But each $\mu_i(t)$ is a weighted average of $\beta_j(t)$, and hence $\mu_i(t) \in [\frac{\mindeg}{t}, 1 - \frac{\mindeg}{t}]$ for all $i,t$.
\end{proof}

Note that this means that any agent $i$ will (almost surely) declare both $0$ and $1$ infinitely many times, even if the network approaches consensus, because either opinion has probability $\geq \Theta(1/t)$ at step $t$ (and $\sum_t 1/t = \infty$).

\subsubsection{Negative Log-Likelihood Properties}
\label{sec::symmetric_parameterization}
We will analyze in depth the MLE which is key to our analysis. We start by introducing an alternative representation for $\ell_i(\gamma, t)$. 
%
%
 Let $\tilde \psi_{i, t} = 2 \psi_{i,t} - 1$,
which takes values  $-1$ and $+1$, instead of $0$ and $1$, which gives a more symmetric representation of the process.
Since $f(\mu_i(t), \gamma)$ is still the probability of $\tilde \psi_{i, t} = 1$,
\ifarxiv 
\begin{align}
    \ell_i(\gamma, t) &= -\log\left( \frac{1}{1 + e^{ -\tilde \psi_{i,t} \log\left(\gamma \frac{\mu_{i}(t-1)}{1-\mu_{i}(t-1)} \right) }}\right)
    \\ & = \log\left( {1 + e^{ -\tilde \psi_{i,t} \log\left(\gamma \frac{\mu_{i}(t-1)}{1-\mu_{i}(t-1)} \right) }}\right)\label{eq::ell_exp_form}
    \,.
\end{align}
\else 
\begin{align}
    \ell_i(\gamma, t) &= \log\left( {1 + e^{ -\tilde \psi_{i,t} \log\left(\gamma \frac{\mu_{i}(t-1)}{1-\mu_{i}(t-1)} \right) }}\right)\label{eq::ell_exp_form}
    \,.
\end{align}
\fi 

We reparameterize $\gamma$ and $\mu_i(t)$ as follows:
\begin{align}
    \logbias \eqdef \log \gamma ~\text{ and }~
    \nu_{i}(t) \eqdef \log \frac{\mu_{i}(t)}{1 - \mu_{i}(t)}
    \label{eq::lambda_nu_def}\,.
\end{align}

Using $\logbias$ symmetrizes the bias parameter across $\bbR$ (so $\logbias = 0$ represents an unbiased agent).
%
%
We thus define some quantities which take $\logbias = \log \gamma$ as the argument instead of $\gamma$ and use them where convenient:
\begin{align}
\tilde \ell_i(\logbias, t) \eqdef \ell_i(\gamma, t) \text{ and } \tilde L_i(\logbias, t) \eqdef L_i(\gamma, t)\,.
\end{align}


For this section to \Cref{sec::freedman} 
\ifacc
\else
\fi
we will fix an agent $i$ and then use $\gammaone$ and $\gammatwo$ to represent any two possible choices for $\gamma_i$.  We then show that if we know that one of these is the true value of $\gamma_i$, in the limit it is almost surely possible to determine which one (\Cref{thm::binary-success}); this result will then be used to show that $\lim_{t \to \infty} \widehat{\gamma}_i = \gamma_i$ almost surely (\Cref{thm::mle-consistency}).
Define
\begin{align}
Z(t) = Z(\gammaone, \gammatwo, t) \eqdef L_i(\gammatwo, t) - L_i(\gammaone, t)\,.
\label{eq::loss_diff_gamma1_gamma2}
\end{align}
If $Z(t)$ is positive, intuitively, $\gammaone$ fits the observed behavior better than $\gammatwo$, so we expect $\gammaone$ to be the true parameter.
Indeed, if $\gammaone$ is the true parameter, then 
\begin{align}
    \bbE&[Z(t) | \cH_{t-1}]  
    = \sum_{\tau = \startindex}^t \bbE\bigg[\bbI\{\psi_{i,\tau} = 1\} \log \frac{f(\mu_i(\tau-1),\gammaone)}{f(\mu_i(\tau-1),\gammatwo)} 
    \eqlinebreakshort
    + \bbI\{\psi_{i,\tau} = 0\} \log \frac{1-f(\mu_i(\tau-1),\gammaone)}{1-f(\mu_i(\tau-1),\gammatwo)} \Big | \cH_{\tau-1}\bigg]
    \\& = \sum_{\tau = \startindex}^t D_{\kl}(f(\mu_i(\tau-1),\gammaone) \| f(\mu_i(\tau-1),\gammatwo))\label{eq::likelihood_as_kl}
\end{align}
which is always a nonnegative quantity.


\begin{proposition}\label{prop::likelihood_properties}
    $L_i(\gamma, t)$ is a stochastic process which satisfies the following properties:
    \begin{enumerate}[(a)]
        \item For fixed $\gamma$, $L_i(\gamma, t)$ (and $\tilde L_i(\logbias, t)$) is an increasing function in $t$
        \item \label{item::convexity_ell} For fixed $t$, $\tilde L_i(\logbias, t)$ is a strictly convex function in $\logbias$
        \item \label{item::bounds_on_ell_inc} $\ell_i(\gamma, t) \in [0, \infty)$, and for a fixed $t$,
        \begin{itemize}
            \item If $\tilde \psi_{i,t} = -1$, then $\ell_i(\gamma, t)$ is a decreasing function in $\gamma$ (and $\tilde \ell_i(\logbias, t)$ is decreasing in $\logbias$)
            \item If $\tilde \psi_{i,t} = 1$, then $\ell_i(\gamma, t)$ is an increasing function in $\gamma$ (and $\tilde \ell_i(\logbias, t)$ is increasing in $\logbias$)
        \end{itemize}
        \item If there exists $t_1, t_2 \leq t$ where $\tilde \psi_{i, t_1} = 1$ and $\tilde \psi_{i, t_2} = -1$, then $\tilde L_i(\logbias,t)$ has unique finite minimum as a function in $\logbias$. Also $L_i(\gamma,t)$ has the same minimum at $\gamma = e^{\logbias}$.
        \item \label{item::likelihood_properties_exp_ell} For any $\gamma \neq \gamma_i$,
        \begin{align}
            \bbE[\ell_i(\gamma, t)|\cH_{t-1}] > \bbE[\ell_i(\gamma_i, t)|\cH_{t-1}] 
        \end{align}
    \end{enumerate}
\end{proposition}
\vspace{-1.5pc}

\ifacc
We show the proof of property (\ref{item::convexity_ell}) below; proofs of the other properties are omitted.
\begin{proof}
    \begin{align}
        \frac{d^2}{d\logbias^2} \tilde\ell_i(\logbias,t) &= \frac{d^2}{d\logbias^2} \log\big(1 + e^{-\tilde \psi_{i,t}(\logbias + \nu_i(t-1))}\big)
        \\ &= \frac{d}{d\logbias} \frac{-\tilde \psi_{i,t} e^{-\tilde \psi_{i,t}(\logbias + \nu_i(t-1))}}{1 + e^{-\tilde \psi_{i,t}(\logbias + \nu_i(t-1))}}\label{eq::first_derivative_tilde_ell}
        \\ &= -\tilde \psi_{i,t} \frac{d}{d\logbias} \frac{1}{1 + e^{\tilde \psi_{i,t}(\logbias + \nu_i(t-1))}}
        \\ &= \tilde \psi_{i,t}^2 \frac{e^{\tilde \psi_{i,t}(\logbias + \nu_i(t-1))}}{(1 + e^{\tilde \psi_{i,t}(\logbias + \nu_i(t-1))})^2}
        \\ &= \frac{e^{\tilde \psi_{i,t}(\logbias + \nu_i(t-1))}}{(1 + e^{\tilde \psi_{i,t}(\logbias + \nu_i(t-1))})^2}
        \\ &> 0\,.
    \end{align}
    Thus $\tilde \ell_i(\logbias,t)$ is convex for all $t$, and so $\tilde L_i(\logbias,t) = \sum_{\tau=\startindex}^t \tilde \ell_i(\logbias,\tau)$ is also convex.
\end{proof}
\else

\begin{proof}
\indent
\begin{enumerate}[(a)]
    \item For each $t$, $\ell_i(\gamma, t)$ is nonnegative, so $L_i(\gamma, t)$ must be increasing. (Similar for $\tilde L_i(\logbias, t)$).
    \item 
  \ifarxiv 
    \begin{align}
        \frac{d^2}{d\logbias^2} \tilde\ell_i(\logbias,t) &= \frac{d^2}{d\logbias^2} \log\big(1 + e^{-\tilde \psi_{i,t}(\logbias + \nu_i(t-1))}\big)
        \\ &= \frac{d}{d\logbias} \frac{-\tilde \psi_{i,t} e^{-\tilde \psi_{i,t}(\logbias + \nu_i(t-1))}}{1 + e^{-\tilde \psi_{i,t}(\logbias + \nu_i(t-1))}}\label{eq::first_derivative_tilde_ell}
        \\ &= -\tilde \psi_{i,t} \frac{d}{d\logbias} \frac{1}{1 + e^{\tilde \psi_{i,t}(\logbias + \nu_i(t-1))}}
        \\ &= \tilde \psi_{i,t}^2 \frac{e^{\tilde \psi_{i,t}(\logbias + \nu_i(t-1))}}{(1 + e^{\tilde \psi_{i,t}(\logbias + \nu_i(t-1))})^2}
        \\ &= \frac{e^{\tilde \psi_{i,t}(\logbias + \nu_i(t-1))}}{(1 + e^{\tilde \psi_{i,t}(\logbias + \nu_i(t-1))})^2}
        \\ &> 0\,.
    \end{align}

    Note that $\tilde \psi_{i,t}^2 = 1$. 
    \else
    \begin{align}
        \frac{d^2}{d\logbias^2} \tilde\ell_i(\logbias,t) &= \frac{d^2}{d\logbias^2} \log\big(1 + e^{-\tilde \psi_{i,t}(\logbias + \nu_i(t-1))}\big)
        \\ &= \frac{e^{\tilde \psi_{i,t}(\logbias + \nu_i(t-1))}}{(1 + e^{\tilde \psi_{i,t}(\logbias + \nu_i(t-1))})^2}
        > 0\,.
    \end{align}
    \fi 
    Thus $\tilde \ell_i(\logbias,t)$ is convex for all $t$, and so $\tilde L_i(\logbias,t) = \sum_{\tau=\startindex}^t \tilde \ell_i(\logbias,\tau)$ is also convex.

    \item Using \eqref{eq::ell_exp_form}, changing the sign of $\tilde \psi_{i, t}$ changes the sign on the exponent. If $\tilde \psi_{i, t}$ is positive, then the quantity in the exponent is decreasing as $\gamma$ increases. The range is $[0, \infty)$ since the quantity in the $\log$ is greater than or equal to $1$.
    \item Follows from (\ref{item::convexity_ell}) and (\ref{item::bounds_on_ell_inc}). The function $\tilde L_i(\logbias, t)$ must be convex and go to infinity at both ends. Since $L_i(\gamma,t) = \tilde L_i(\logbias, t)$, it has the same minimum.
    \item Result follows from \eqref{eq::likelihood_as_kl} setting $\gammaone = \gamma_i$ and $\gammatwo = \gamma$. The KL divergence must always be nonnegative and equal to zero iff $\gammaone = \gammatwo$.
\end{enumerate}
\end{proof}
\fi

\subsection{Log-Likelihood Ratios and Martingales}

\label{sec::consistency_martingale}



To properly analyze the quantity \eqref{eq::loss_diff_gamma1_gamma2}, we need the following definitions.
Unless otherwise stated, $\gammaone$ is the true parameter from which the random data is generated. 
\ifacc
\else
The following definitions will be used starting from this section to \Cref{sec::freedman}.
\fi
%
The \emph{loss difference} is
\begin{align}
    Z(t) &\eqdef Z(\gammaone, \gammatwo, t) 
    \\ z(t) &\eqdef z(\gammaone, \gammatwo, t) \eqdef \ell_i(\gammatwo, t) - \ell_i(\gammaone, t)\,.
\end{align}

The \emph{predictable expected value} is
\begin{align}
    X(t) &\eqdef X(\gammaone,\gammatwo, t) \eqdef \sum_{\tau = \startindex}^t \bbE[z(\tau) | \cH_{\tau-1}] 
    \\ x(t) &\eqdef x(\gammaone,\gammatwo, t) \eqdef \bbE[z(t) | \cH_{t-1}]\,.
\end{align}

The \emph{loss martingale} is
\begin{align}
    Y(t) &\eqdef Y(\gammaone,\gammatwo, t) \eqdef X(t) - Z(t)
    \\ y(t) &\eqdef y(\gammaone,\gammatwo, t) \eqdef x(t) - z(t)\,.
\end{align}

The \emph{predictable quadratic variation} is
\begin{align}
    W(t) \eqdef W(\gammaone,\gammatwo, t) &\eqdef  \sum_{\tau = \startindex}^t \var[z(\tau) | \cH_{\tau-1}]
    \\& =  \sum_{\tau = \startindex}^t \var[y(\tau) | \cH_{\tau-1}]
    \\ w(t) \eqdef w(\gammaone,\gammatwo, t) &\eqdef \var[z(t) | \cH_{t-1}] 
    \\& = \var[y(t) | \cH_{t-1}] \,.
\end{align}

%
We give some preliminary results about these processes.

\begin{proposition} 
We have the following properties:
    \begin{enumerate}[(a)]
        \item $Z(t)$ is a submartingale and $X(t)$ is strictly increasing
        \item $Y(t)$ is a martingale
        \item $W(t)$ is strictly increasing 
    \end{enumerate}
\end{proposition}

\ifacc
\else
\begin{proof}
\indent
\begin{enumerate}[(a)]
    \item The two statements are equivalent. From \Cref{prop::likelihood_properties} (\ref{item::likelihood_properties_exp_ell}), we have
    \begin{align}
        x(t) 
        &= \bbE[\ell_i(\gammatwo, t)| \cH_{t-1}] - \bbE[\ell_i(\gammaone, t)| \cH_{t-1}] > 0 \,.
    \end{align}
    \item This follows from the definitions of $Y(t)$ and $y(t)$.
    \begin{align}
        \bbE[y(t) | \cH_{t-1}] &= \bbE[x(t) - z(t) | \cH_{t-1}] 
        \\&= \bbE[\bbE[z(t) | \cH_{t-1}] - z(t) | \cH_{t-1}] = 0
    \end{align} 
    and
    \begin{align}
        \bbE[Y(t) | \cH_{t-1}] = Y(t-1) + \bbE[y(t) | \cH_{t-1}] = Y(t-1)\,.
    \end{align}
    \item Because of the bounds on $\mu_i(t)$ given in \Cref{lem::bounds_on_mu}, each $\ell_i(\gamma, t)$ must be non-constant so long as $\gamma \neq 1$. Then $z(t)$ is non-constant so long as $\gammaone \neq \gammatwo$. The quantity $w(t)$ is the conditional variance of $z(t)$ which therefore must always be positive. The quantity $W(t)$ is a sum of $w(t)$ so it must be increasing. 
\end{enumerate}
\end{proof}
\fi
%
Next we determine bounds on our quantities. 
%
\ifacc
\begin{lemma}
\label{lem::xt_1overct}
If $\gammaone \neq \gammatwo$, then there is some $t_0 = t_0(\gammaone, \gammatwo)$ and $c_0 = c_0(\gammaone, \gammatwo) > 0$ such for all $t > t_0$
\begin{align}
x(t) \geq c_0 (\mindeg/t)\,.
\end{align}

Additionally, there are some constants $k$, $t_1$ (which depend on $t_0, \gammaone, \gammatwo$) such that for all $t > t_1$,
\begin{align}
    X(t) > k c_0 \mindeg \log (t) \,.
\end{align}

\label{lem::var_as_exp}
Similarly, there exists a constant $c_1 = c_1(\gammaone, \gammatwo) > 0$ and $t_2$ such that for all $t > t_2$
\begin{align}
    w(t) \leq c_1 x(t)\,.
\end{align}
This also implies
\begin{align}
    W(t) \leq c_1 X(t)\label{eq::Wt_less_cXt}
    \,.
\end{align}

\end{lemma}
 Combining \Cref{lem::xt_1overct} and \eqref{eq::likelihood_as_kl} gives that for $\gammaone \neq \gammatwo$,
\begin{align}
    \lim_{t\to \infty} \bbE[Z(\gammaone, \gammatwo, t)] 
    &=  \lim_{t\to \infty} X(\gammaone, \gammatwo, t)   
    = \infty
    \label{eq::xt_infty}
    \,.
\end{align}

\begin{remark}
    {\color{blue}The fact that $\bbE[Z(t)] \to \infty$ relies on $\mu_i(t) \in \left[\mindeg/ t , 1 - \mindeg/ t\right]$, 
    as discussed in \Cref{sec::bounds_on_mu} (\Cref{lem::bounds_on_mu}).}
    
    If instead, $\mu_{i}(t)$ scales as $1/t^2$, 
    then the limit of $\bbE[Z(t)]$ would be finite. In such a scenario, randomness might make $Z(t)$ unreliable for distinguishing between $\gammaone$ and $\gammatwo$.

    Changes to the model which may cause the condition $\mu_i(t) \in \left[\mindeg/ t , 1 - \mindeg/ t\right]$ to fail include putting higher weight on previously declared opinions or having the network add more agents at each time step $t$. 
\end{remark}


\else
First we bound the predictable expected value $X(t)$.
Since $\gammaone$ is the true bias, like in \eqref{eq::likelihood_as_kl}, we can write
\begin{align}
    x(t) = D_{\kl}\left(f(\mu_i(t), \gammaone) \| f(\mu_i(t), \gammatwo)\right)
    \label{eq::xt_kl}\,.
\end{align}

\ifacc
\else
\begin{lemma}\label{lem::xt_bound_hellinger}
For each time $t$, we can bound
\begin{align}
x(t) \geq \frac{(\sqrt{\gammaone} - \sqrt{\gammatwo})^2 \mu_i(t)(1-\mu_i(t))}{\max\{ \frac{\gammaone + \gammatwo}{2}, 1 \}^2}\,.
\end{align}
\end{lemma}
\begin{proof}
\revise{
To start, since $x(t)$ can be expressed as a KL divergence,} 
\ifarxiv
we will lower bound this KL divergence by using squared Hellinger distance, specifically,
\begin{align}\label{eq::kl_vs_H}
D_{\kl}(P \| Q) \geq 2 H^2(P, Q)
\end{align}
\else
we can use
$\label{eq::kl_vs_H}
D_{\kl}(P \| Q) \geq 2 H^2(P, Q)
$
\fi
which we can derive from \cite[7.3]{polyanskiy2014lecture}.
\ifarxiv
For discrete distributions $p$ and $q$ over set $[1,\dots, k]$, 
\begin{align}
H^2(p, q) = 1-  \sum_{i = 1}^k \sqrt{ p(i) q(i)}\,.
\end{align}
\fi
This gives that
\ifarxiv 
\begin{align}
H^2&(f(\mu, \gammaone), f(\mu, \gammatwo)) 
\\&=
1 - \sqrt{\frac{\gammaone \mu \gammatwo \mu}{(\gammaone \mu + (1-\mu))(\gammatwo \mu + (1-\mu))}} 
\eqlinebreakshort
- \sqrt{\frac{(1-\mu) (1-\mu)}{(\gammaone \mu + (1-\mu))(\gammatwo \mu + (1-\mu))}}
\\ & = 1 - \frac{\sqrt{\gammaone \gammatwo} \mu + (1-\mu)}{\sqrt{(\gammaone \mu + (1-\mu))(\gammatwo \mu + (1-\mu))}}\,.
\end{align}
\else 
\begin{align}
H^2&(f(\mu, \gammaone), f(\mu, \gammatwo)) 
\\ & = 1 - \frac{\sqrt{\gammaone \gammatwo} \mu + (1-\mu)}{\sqrt{(\gammaone \mu + (1-\mu))(\gammatwo \mu + (1-\mu))}}\,.
\end{align}
\fi 

Let
\begin{align}
A &= \sqrt{\gammaone \gammatwo} \mu + (1-\mu)
\\ B &= \sqrt{(\gammaone \mu + (1-\mu))(\gammatwo \mu + (1-\mu))}\,.
\end{align}

\ifarxiv
Note that $B>A$ since squared Hellinger distance is always between $0$ and $1$.
\fi
\ifarxiv 
Then
\begin{align}
H^2&(f(\mu, \gammaone), f(\mu, \gammatwo)) = \frac{B - A}{B} 
\\&= \frac{B^2 - A^2}{B(B+A)} \geq \frac{B^2 - A^2}{2 B^2}
\,.
\end{align}
We can compute
\begin{align}
& B^2 - A^2 
\\&=(\gammaone \mu + (1-\mu))(\gammatwo \mu + (1-\mu)) - (\sqrt{\gammaone \gammatwo} \mu + (1-\mu))^2
\\& = (\gammaone + \gammatwo - 2\sqrt{\gammaone \gammatwo}) \mu (1-\mu)
\\& = (\sqrt{\gammaone } - \sqrt{\gammatwo})^2 \mu(1-\mu)
\end{align}
and using AM-GM
\begin{align}
B^2 &=  (\gammaone \mu + (1-\mu))(\gammatwo \mu + (1-\mu))
\\&\leq \left(\frac{\gammaone + \gammatwo}{2} \mu + (1-\mu)\right)^2
\leq \left(\max\left\{\frac{\gammaone + \gammatwo}{2}, 1\right\}\right)^2\,.
\end{align}
\else 
We will use 
\begin{align}
 \frac{B - A}{B} &= \frac{B^2 - A^2}{B(B+A)} \geq \frac{B^2 - A^2}{2 B^2}\,,
\end{align}
the fact that
$B^2 - A^2  = (\sqrt{\gammaone } - \sqrt{\gammatwo})^2 \mu(1-\mu)$ and that
\begin{align}
B^2 &\leq \left(\frac{\gammaone + \gammatwo}{2} \mu + (1-\mu)\right)^2
\leq \left(\max\left\{\frac{\gammaone + \gammatwo}{2}, 1\right\}\right)^2
\end{align}
\fi 
This results in 
\begin{align}
H^2&(f(\mu, \gammaone), f(\mu, \gammaone)) \geq \frac{(\sqrt{\gammaone } - \sqrt{\gammatwo})^2 \mu(1-\mu)}{2 \max\left\{\frac{\gammaone + \gammatwo}{2}, 1\right\}^2}\,.
\end{align}
Combining this with 
\ifarxiv
\eqref{eq::kl_vs_H} and 
\fi
\eqref{eq::xt_kl} completes the proof. 
\end{proof}

\fi

\begin{lemma}\label{lem::xt_1overct}
If $\gammaone \neq \gammatwo$, then there is some $t_0 = t_0(\gammaone, \gammatwo)$ and $c_0 = c_0(\gammaone, \gammatwo) > 0$ such that for all $t > t_0$
\begin{align}
x(t) \geq c_0 (\mindeg/t)\,.
\end{align}

Additionally, there are some constants $k$, $t_1$ (which depend on $t_0, \gammaone, \gammatwo$) such that for all $t > t_1$,
\begin{align}
    X(t) > k c_0 \mindeg \log (t) \,.
\end{align}
\end{lemma}
The value of $\mindeg$ is defined in \Cref{{lem::bounds_on_mu}}.

\begin{proof}
From \Cref{lem::xt_bound_hellinger}, it suffices to let 
\begin{align}
c_0 = \frac{1}{2}\frac{(\sqrt{\gammaone} - \sqrt{\gammatwo})^2}{\max\{ \frac{\gammaone + \gammatwo}{2}, 1 \}^2}\,.
\end{align}

Then using \Cref{lem::bounds_on_mu},
\begin{align}
\mu_i(t)(1-\mu_i(t)) \geq \frac{1}{2} \min\{\mu_i(t), 1-\mu_i(t)\} \geq \frac{1}{2}\frac{\mindeg}{t}\,.
\end{align}
This
gives that
$
x(t) \geq c_0 (\mindeg/ t) 
$
which implies
\begin{align}
    X(t) &= \sum_{s = \startindex}^t x(s) 
    \geq \sum_{s = t_0}^t c_0 \frac{\mindeg}{s} 
    \geq k_0 c_0 \mindeg \log (t)\,.
\end{align}
Constant $k_0$ accounts for some loss which occurs since $X(t)$ is a sum of terms $x(t)$, and for small $t$, the results may not be exact.
\end{proof}

The 
stochastic process $Z(\gammaone, \gammatwo, t)$ is a likelihood ratio test for determining whether $\gammaone$ or $\gammatwo$ is the true parameter. Combining \Cref{lem::xt_1overct} and \eqref{eq::likelihood_as_kl} gives that for $\gammaone \neq \gammatwo$,
\begin{align}
    \lim_{t\to \infty} \bbE[Z(\gammaone, \gammatwo, t)] 
    &=  \lim_{t\to \infty} X(\gammaone, \gammatwo, t)   
    = \infty
    \label{eq::xt_infty}
    \,.
\end{align}

\ifarxiv 
The likelihood ratio $Z(\gammaone, \gammatwo, t)$ on average is very large as $t$ gets large. This means that $Z(\gammaone, \gammatwo, t)$ can be used to distinguish which of the two parameters, $\gammaone$ or $\gammatwo$, is the true parameter governing the data. 
\fi
If $Z(\gammaone, \gammatwo, t)$ is very large (positive), then $\gammaone$ is the true parameter.  If $Z(\gammaone, \gammatwo, t)$ is very small (negative), then $\gammatwo$ is the true parameter.

\begin{remark}
    The fact that $\bbE[Z(t)] \to \infty$ relies on $\mu_i(t) \in \left[\mindeg/ t , 1 - \mindeg/ t\right]$, 
    as discussed in \Cref{sec::bounds_on_mu} (\Cref{lem::bounds_on_mu}).
    
    If instead, $\mu_{i}(t)$ scales as $1/t^2$, 
    then the limit of $\bbE[Z(t)]$ would be finite. In such a scenario, randomness might make $Z(t)$ unreliable for distinguishing between $\gammaone$ and $\gammatwo$.

    Changes to the model which may cause the condition $\mu_i(t) \in \left[\mindeg/ t , 1 - \mindeg/ t\right]$ to fail include putting higher weight on previously declared opinions or having the network add more agents at each time step $t$.  
\end{remark}


\begin{lemma} \label{lem::bdd-step-size}
For any $t$, we have
\begin{align}
    |Z(t) - Z(t-1)| &\leq \left|\log \frac{\gammaone}{\gammatwo} \right| 
    \\ |Y(t) - Y(t-1)|&\leq \left|\log \frac{\gammaone}{\gammatwo} \right| 
\end{align}
\end{lemma}

\ifacc
The proof is omitted.
\else
\begin{proof}
Let $\mu = \mu_i(t)$.
Since $\gammatwo \neq \gammaone$, we know that either:
\begin{enumerate}[(i)]
\item $\log\left(\frac{f(\mu,\gammaone)}{f(\mu,\gammatwo)}\right) < 0 < \log\left(\frac{1-f(\mu,\gammaone)}{1-f(\mu,\gammatwo)}\right)$ (if $\gammaone < \gammatwo$), 
\item
$\log\left(\frac{1-f(\mu,\gammaone)}{1-f(\mu,\gammatwo)}\right) < 0 < \log\left(\frac{f(\mu,\gammaone)}{f(\mu,\gammatwo)}\right)$ (if $\gammaone > \gammatwo$).
\end{enumerate}

Since
\begin{align}
    Z(t) = \begin{cases} 
    Z(t-1) + \log\big(\frac{f(\mu,\gammaone)}{f(\mu,\gammatwo)}\big) &\text{if } \psi_{i,t} = 1 
    \\ Z(t-1) + \log\big(\frac{1-f(\mu,\gammaone)}{1-f(\mu,\gammatwo)}\big) &\text{if } \psi_{i,t} = 0  \end{cases}
\end{align}
and by 
\ifarxiv 
\begin{align}
&\left| \log\left(\frac{f(\mu,\gammaone)}{f(\mu,\gammatwo)}\right) - \log\left(\frac{1-f(\mu,\gammaone)}{1-f(\mu,\gammatwo)}\right)\right| 
\\&=\left| \log \frac{\gammaone}{\gammatwo}\frac{\gammatwo \mu + (1-\mu)}{\gammaone \mu + (1-\mu)}  -\log \frac{\gammatwo \mu -  (1-\mu)}{\gammaone \mu + (1-\mu)} \right|
\\&= \left| \log \frac{\gammaone}{\gammatwo} \right|
\end{align}
\else 
\begin{align}
&\left| \log\left(\frac{f(\mu,\gammaone)}{f(\mu,\gammatwo)}\right) - \log\left(\frac{1-f(\mu,\gammaone)}{1-f(\mu,\gammatwo)}\right)\right| 
&= \left| \log \frac{\gammaone}{\gammatwo} \right|
\end{align}
\fi 
this means that $Z(t-1)$ and $Z(t)$ are both in the same $\left| \log \frac{\gammaone}{\gammatwo} \right|$-sized interval ($Z(t)$ is one of the endpoints and $Z(t-1)$ is somewhere in the middle). Additionally, $Y(t)$ (given history $\cH_{t-1}$) is also a binary random variable whose possible outcomes are $\left| \log \frac{\gammaone}{\gammatwo} \right|$ apart, and since $\bbE[Y(t) \, | \, \cH_{t-1}] = Y(t-1)$ this interval also must contain $Y(t-1)$, and we are done.
\end{proof}
\fi
%
%
Finally, we bound the predictable quadratic variation $W(t)$.
\ifacc 
\else  
\ifarxiv
Since $W(t)$ is defined as a variance, the following standard identity is helpful for finding an upper bound. Suppose that variable $X$ takes two values, $a$ with probability $p$ and $b$ with probability $(1-p)$, then
\begin{align}
    \var[X] &= p(1-p) (a-b)^2\,.
\end{align}
Applied to $w(t)$, we get that
\else
Using standard variance calculation, we get
\fi
\ifarxiv 
\begin{align}
    &w(t) = \var[z(t) | \cH_{t-1}]
    \\& = f(\mu_i(t), \gammaone) \left(1 - f(\mu_i(t), \gammaone) \right)
    \eqlinebreakshort
    \left(\log \frac{f(\mu_i(t), \gammaone)}{f(\mu_i(t), \gammatwo)} - \log \frac{1 - f(\mu_i(t), \gammaone)}{1 - f(\mu_i(t), \gammatwo)} \right)^2
    \\ & = \left(\frac{\gammaone \mu_i(t) }{1 + (\gammaone -1) \mu_i(t)} \right) 
    \left( \frac{1- \mu_i(t) }{1 + (\gammaone -1) \mu_i(t)}\right) 
    \eqlinebreakshort
    \left(\log \frac{\gammaone}{\gammatwo}\frac{1 + (\gammatwo -1) \mu_i(t)}{1 + (\gammaone -1) \mu_i(t)} - \log \frac{1 + (\gammatwo -1) \mu_i(t)}{1 + (\gammaone -1) \mu_i(t)} \right)^2
    \\ & = \frac{\gammaone \mu_i(t) (1- \mu_i(t))}{(1 + (\gammaone -1) \mu_i(t))^2}  \left(\log \frac{\gammaone}{\gammatwo}\right)^2\label{eq::lrt_variance}\,.
\end{align}
\else 
\begin{align}
    &w(t) = \var[z(t) | \cH_{t-1}]
      = \frac{\gammaone \mu_i(t) (1- \mu_i(t))}{(1 + (\gammaone -1) \mu_i(t))^2}  \left(\log \frac{\gammaone}{\gammatwo}\right)^2\label{eq::lrt_variance}\,.
\end{align}
\fi 

\fi  

\begin{lemma}\label{lem::var_as_exp}
When $\gammaone \neq \gammatwo$, there exists a constant $c_1 = c_1(\gammaone, \gammatwo) > 0$ and $t_1$ such that for all $t > t_1$
\begin{align}
    w(t) \leq c_1 x(t)\,.
\end{align}

This also implies
\begin{align}
    W(t) \leq c_1 X(t)\label{eq::Wt_less_cXt}
    \,.
\end{align}

\end{lemma}

\ifacc
The proof is omitted
\else
\begin{proof}
Starting with \eqref{eq::lrt_variance}, we have that
\begin{align}
w(t) &\leq  \frac{\gammaone \left(\log \frac{\gammaone}{\gammatwo}\right)^2 }{\left(\min\{1, \gammaone \}\right)^2}  \mu_i(t) (1- \mu_i(t))
\\ & \leq  \frac{\gammaone \left(\log \frac{\gammaone}{\gammatwo}\right)^2 }{\left(\min\{1, \gammaone \}\right)^2} \frac{\max\{ \frac{\gammaone + \gammatwo}{2}, 1 \}^2}{(\sqrt{\gammaone} - \sqrt{\gammatwo})^2 } x(t)
\end{align}
and thus we can set
$c_1$ to be the coefficient in front of $x(t)$.
Since $W(t)$ is a sum of $w(t)$ and $X(t)$ is a sum of $x(t)$, we naturally have
$
    W(t) \leq c_1 X(t)
$
for all $t$.
\end{proof}
\fi

\fi

\ifacc
\else
\subsection{Concentration by Freedman's Inequality}
\fi

\label{sec::freedman}

We want to show that the test $Z(t) > 0$ works to distinguish whether $\gammaone$ or $\gammatwo$ is the true parameter. We do this by showing that if $\gammaone$ is the true parameter, then almost surely $Z(t) \leq 0$ (i.e. the test fails) for only finitely many $t$. 
\ifacc
We show this by applying Freedman's inequality (\cite{tropp2011} and \cite{freedman1975} (Thm 1.6)) and \Cref{lem::xt_1overct} (proof omitted):
\else
We show this by applying Freedman's inquality (this formulation taken from \cite{tropp2011} (Thm 1.1), but originally from \cite{freedman1975} (Thm 1.6)):
\begin{theorem}[Freedman's Martingale Inequality \cite{tropp2011}]
\label{thm::freedman}
If $Y(t)$ is a martingale with steps $y(t) = Y(t) - Y(t-1)$ such that $|y(t)| \leq \alpha$ almost surely (and $Y(0) = 0$), and the \emph{predictable quadratic variation} of $Y(t)$ is
\begin{align}
    W(t) = \sum_{\tau=\startindex}^t \bbE[y(\tau)^2 | Y(1), \dots, Y(\tau-1)]
\end{align}
then for any $s, \sigma^2 > 0$,
\begin{align}
    \bbP[\exists \, t : Y(t) \geq s, W(t) \leq \sigma^2] \leq \exp \Big(\frac{-s^2/2}{\sigma^2 + \alpha s/3} \Big)\,.
\end{align}
\end{theorem}
\ifarxiv
This inequality is an extension of Bernstein's inequality to martingales, where the variance of each step is not fixed but is itself a random variable dependent on the history.
\fi
Using \Cref{thm::freedman}, we get our result for our test:
\fi

\begin{theorem}\label{thm::binary-success}
    If $\gamma_i \in \{\gammaone, \gammatwo\}$, then the likelihood ratio test $Z(t) = L_i(\gammatwo, t) - L_i(\gammaone, t)$ is such that
    \begin{align}
        Z(t) \begin{cases} 
        > 0 \text{ if } \gamma_i = \gammaone \\
        < 0 \text{ if } \gamma_i = \gammatwo
        \end{cases}
    \end{align}
    for all but finitely many $t$.
\end{theorem}

\ifacc
\else
\begin{proof}
For the proof, suppose that $\gamma_i = \gammaone$. If $\gamma_i = \gammatwo$, the same proof holds except with $-Z(t)$.

For the purposes of finding a contradiction, assume that $Z(t) \leq 0$ infinity often. Since $Z(t) = X(t) - Y(t)$, we have
\begin{align}
    Z(t) \leq 0 \iff X(t) \leq Y(t)\,.
\end{align}
Using \eqref{eq::Wt_less_cXt} from \Cref{lem::var_as_exp}, we have
\begin{align}
     X(t) \leq Y(t) \iff W(t) \leq c_1 X(t) \leq c_1 Y(t)\,.
\end{align}

Now suppose there are infinite values of $t$ where $Z(t) \leq 0$. This means there are infinite values where $W(t) \leq c_1 Y(t)$. 
From \eqref{eq::xt_infty}, we know that $X(t) \to \infty$ as $t\to \infty$.
This means that there is a $t$ where $X(t) > 2$.
In that case, we have
\begin{align}
    Y(t) - \frac{1}{2 c_1} W(t) \geq X(t) -\frac{1}{2} X(t) \geq 1\,.
\end{align}
This means that so long as $t$ is large enough so that $X(t) > 2$, we have that
\begin{equation}
    \begin{aligned}
    &Y(t) - \frac{1}{2 c_1} W(t) \geq 1 
    \\
    \implies &\exists s \in \bbZ_{>0} \text{ such that }
\frac{1}{2c_1} W(t) \leq s \leq Y(t)\,.
    \end{aligned}
\label{eq::integer_s_exists}
\end{equation}

Let us define a set of \emph{bad times} where the estimator fails:
\begin{align}
    \cT \eqdef \{t>1: Z(t) \leq 0, X(t) > 2 \} = \{ \tilde t_1, \tilde t_2,\dots\}
\end{align}
where they are ordered $\tilde t_1< \tilde t_2 < \dots$.
We assume that $\cT$ is infinite and then derive a contradiction.
We define
\begin{align}
    \tilde s_k \eqdef \max_{s}\left\{s \in \bbZ_{>0}: \frac{1}{2c_1} W(\tilde t_k) \leq s \leq Y(\tilde t_k) \right\}\,.
\end{align}
From \eqref{eq::integer_s_exists}, we know such an $\tilde s_k$ exists for each $k$. In fact,
$
    \tilde s_k = \lfloor Y(\tilde t_k)\rfloor \,.
$
If $\cT$ is infinite, then this produces an infinite sequence of integers $\tilde s_1, \tilde s_2, \dots$. Because $X(t) \to \infty$, which implies that $Y(\tilde t_k) \to \infty$, we have also that
\begin{align}
    \lim_{k \to \infty} \tilde s_k = \infty\,.
\end{align}

While the sequence $\tilde s_1, \tilde s_2, \dots$ could have many copies of the same integer, the set $\{s : s = \tilde s_k \text{ for some } k\}$ must be infinite since $\lim_{k \to \infty} \tilde s_k = \infty$. In other words, if $\cT$ is infinite there must be infinitely many positive integers $s$ such that
\begin{align}
    Y(t) \geq s \text{ and } W(t) \leq 2 c_1  s\,.
\end{align}

\ifacc
Applying Freedman's inequality,
\else
Applying \Cref{thm::freedman}, 
\fi
we get that there are infinitely many $ s$ such that
\begin{align}
      \bbP&[\exists \, t : Y(t) \geq s, W(t) \leq 2 c_1  s] \leq \exp \Big(\frac{- s^2/2}{ 2 c_1  s + \alpha  s/3} \Big) 
      \eqlinebreakshort
      = \exp\left(\frac{ s}{4 c_1 + 2 \alpha/3} \right) = \exp(-\xi  s)\label{eq::applying_freedman}
\end{align}
where $\xi = 1/(4 c_1 + 2 \alpha/3)$. \Cref{lem::bdd-step-size} then gives the bound
\begin{align}
    y(t) \leq \left|\log \frac{\gammaone}{\gammatwo} \right| \eqdef \alpha\,.
\end{align}

For each $s \in \bbZ_{>0}$, let event $A_s$ be 
\begin{align}
    A_s = \left\{ \exists t: Y(t)>s, W(t) \leq 2 c_1 s \right \}\,.
\end{align}
Then,
\begin{align}
    \sum_{s = 1}^\infty \bbP[A_s] \leq \sum_{s = 1}^\infty \exp(-\xi s) < \infty
\end{align}
and therefore by Borel-Cantelli, almost surely only finitely many events $A_s$ can occur. This is a contradiction and proves our result.
\end{proof}
\fi

\subsection{Combining Results to Prove Consistency of Estimators}

\label{sec::consistency_final}

The above concentration result shows that the MLE $\widehat{\gamma}_i(t)$ \eqref{eq::hat_gamma_mle} converges asymptotically to the true $\gamma_i$.
\begin{theorem}
   \label{thm::mle-consistency}
For any agent $i$, almost surely, 
\begin{align}
    \lim_{t \to \infty} \widehat{\gamma}_i(t) = \gamma_i\,.
\end{align}

\end{theorem}

\begin{proof}

We take advantage of the alternative parameterization of $\logbias = \log \gamma$ (from \Cref{sec::symmetric_parameterization}). Let the MLE for $\logbias$ be 
\begin{align}
\widehat{\logbias}_i (t) = \arg \min_{\logbias} \tilde{L_i}(\logbias, t)\,.
\end{align}
We will show that $\widehat{\logbias}_i(t)$ converges to the true parameter, which we call $\logbias_i$, so $\widehat{\gamma}_i(t)$ converges to the true $\gamma_i$. 

For any fixed $\epsilon > 0$, let $a = \logbias_i - \epsilon$ and $b =  \logbias_i + \epsilon$. 
From \Cref{thm::binary-success}, there exists some time $t_a$ so that $\tilde{L}_i(a, t) > \tilde{L}_i(\logbias_i, t)$ for all $t > t_a$, and there exists some time $t_b$ so that $\tilde{L}_i(b, t) > \tilde{L}_i(\logbias_i, t)$ for all $t>t_b$.

At all times $t > \max\{t_a, t_b\} \eqdef t(\epsilon)$, the value of $\tilde{L}_i(\logbias_i, t)$ is less than both $\tilde{L}_i(a, t)$ and $\tilde{L}_i(b, t)$. By \Cref{prop::likelihood_properties}(\ref{item::convexity_ell}), the function $\tilde{L}_i(\logbias, t)$ is convex in $\logbias$, and thus the minimum of $\tilde{L}_i(\logbias, t)$ at any $t > t(\epsilon)$ must be in $[a,b] = [\logbias_i - \epsilon, \logbias_i + \epsilon]$. 

Thus, for every $\epsilon > 0$, we can always find a $t(\epsilon)$ where for all $t > t(\epsilon)$ we have that $\widehat{\logbias}_i(t)$ is within $\epsilon$ of $\logbias_i$, and thus
$
    \lim_{t \to \infty} \widehat{\logbias}_i(t) = \logbias_i
$
completing the proof.
\end{proof}

\label{sec::derive_inherent_estimator}

This also shows that the inherent belief estimator from \Cref{def::estimator_inherent} almost surely converges to the correct result.

\begin{theorem} 
Almost surely, if $\gamma_i \neq 1$, then
\begin{align}
    \lim_{t \to \infty} \widehat{\phi}_i(t) = \phi_i
\end{align}
\end{theorem}

\begin{proof}

This is equivalent to
\begin{align}
\lim_{t \to \infty}\left(\sum_{\tau = \startindexminus}^{t-1} \mu_{i}(\tau)\right) -  (\countrandom) \bar \beta_i(t)  \begin{cases} 
< 0 & \text{ if } \phi_i = 1
\\> 0 &\text{ if } \phi_i = 0
\end{cases}
\end{align}

The result follows from three facts:
\begin{enumerate}[(i)]
    \item letting $\logbias_i = \log(\gamma_i)$ and $\widehat{\logbias}_i(t) = \argmin_\logbias \tilde{L}_i(\logbias,t) = \log(\widehat{\gamma}_i(t))$ be the maximum likelihood estimator of $\logbias_i$, then $\lim_{t \to \infty} \widehat{\logbias}_i(t) = \logbias_i$;
    \item for any $t$, $\tilde{L}_i(\logbias,t)$ is strictly convex in $\logbias$;
    \item $\left(\sum_{\tau = \startindexminus}^{t-1} \mu_{i}(\tau)\right) -  (\countrandom) \bar \beta_i(t) = \frac{\partial}{\partial \logbias} \tilde{L}_i(\logbias,t) \Big |_{\logbias=0}$.
\end{enumerate} 
Fact (i) follows directly from \Cref{thm::mle-consistency} and (ii) is \Cref{prop::likelihood_properties}(\ref{item::convexity_ell}). Fact (ii) also shows that
\begin{align}
    \widehat{\logbias}_i(t) > 0 \iff \frac{\partial}{\partial \logbias} \tilde{L}_i(\logbias,t) \Big |_{\logbias=0} < 0 \,;
\end{align}
and (assuming $\logbias_i \neq 0$),  $\phi_i = 1 \iff \logbias_i > 0$. Thus facts (i) and (ii) show that (almost surely)
\begin{align}
    &\phi_i = 1 \implies \widehat{\logbias}_i(t) > 0 \text{ for all sufficiently large } t  
    \\ &\implies \frac{\partial}{\partial \logbias} \tilde{L}_i(\logbias,t) \Big |_{\logbias=0} < 0 \text{ for all sufficiently large } t  
\end{align}
Thus, only fact (iii) remains to be shown.

Using $\tilde{L}_i(\logbias,t) = \sum_t \tilde{\ell}_i(\logbias,t)$ and
\begin{align}
 \tilde \ell_i(\logbias, t) = \log \left(1 + e^{-\tilde \psi_{i,t} (\logbias + \nu_i(t-1)) } \right)
\end{align}
to evaluate the derivative
\ifarxiv 
(as in \eqref{eq::first_derivative_tilde_ell})
\else 
\fi 
 at $\logbias = 0$, we get
 \ifarxiv 
\begin{align}
    \frac{\partial}{\partial \logbias} \tilde \ell_i(\logbias, t) \bigg |_{\logbias = 0}
& = \frac{-\tilde \psi_{i,t}}{e^{\tilde \psi_{i,t} \nu_i(t-1) } + 1}
\\& = \begin{cases}
\frac{1}{\frac{1-\mu_{i}(t-1)}{\mu_{i}(t-1) }+1} & \text{ if } \tilde \psi_{i, t} = -1 \\
\frac{-1}{\frac{\mu_{i}(t-1)}{1 -\mu_{i}(t-1) }+1} & \text{ if } \tilde \psi_{i, t} = 1
\end{cases}
\\& = \begin{cases}
\mu_{i}(t-1) & \text{ if } \tilde \psi_{i, t} = -1 \\
\mu_{i}(t-1) - 1 & \text{ if } \tilde \psi_{i, t} = 1
\end{cases}
\\ & =  \mu_{i}(t-1) - \bbI\{\psi_{i,t} = 1\}
\,.
\end{align}
\else 
\begin{align}
    \frac{\partial}{\partial \logbias} \tilde \ell_i(\logbias, t) \bigg |_{\logbias = 0}
& =   \mu_{i}(t-1) - \bbI\{\psi_{i,t} = 1\}
\,.
\end{align}
\fi 

And thus the derivative of 
the entire negative log-likelihood evaluated at $0$ is given by
\begin{align}
\frac{\partial}{\partial \logbias} \tilde L_i(\logbias, t) \bigg |_{\logbias = 0}  
& = \sum_{\tau = \startindex}^t \mu_{i}(\tau-1) - \bbI\{\psi_{i,t} = 1\}
\\ & = \left(\sum_{\tau = \startindexminus}^{t-1} \mu_{i}(\tau)\right) -  (\countrandom) \bar \beta_i(t)
\,.
\end{align}
This shows (iii) and completes the proof.
%
%
\end{proof}

\ifacc

\else

\section{Convergence Rates for Inferring Inherent Beliefs}

\label{sec::rate_converge_estimator}

While we have shown that the estimator for inherent beliefs in \Cref{def::estimator_inherent} will eventually correctly converge to an agent's inherent belief, we are also interested in how fast it converges. 
Since the estimator $\widehat{\phi}_i(t)$ only takes values of $0$ and $1$ (and is therefore always exactly correct or exactly wrong) we say the estimator \emph{converges by time $t^*$} if
\begin{align}
    \widehat{\phi}_i(t) = \phi_i \text{ for all } t \geq t^* \,.
\end{align}
The question is: for any $\delta > 0$, how many steps $t^*$ does it take for the estimator to have a $1-\delta$ probability of converging? Or, in other words, for what $t^*$ do we have
\begin{align}
\bbP\left[\exists t \geq t^* : \widehat{\phi}_i(t) \neq \phi_i\right] \leq \delta \,? \label{eq::estimator_less_delta}
\end{align}

\revise{In this section, the goal is to characterize how the estimator converges based on the convergence rate of declared opinions. Then in the remainder of this section, we determine bounds for the worst-case convergence rate. In the next section, we find bounds when declared opinions approach consensus.}
As in \Cref{sec::freedman} we fix the agent $i$ and omit it from the notation in this section.
Also as in \Cref{sec::freedman}, we assume that agent $i$ has inherent belief $\phi_i = 1$ (so $\gamma_i > 1$). \revise{For agents where $\gamma_i < 1$, the same analysis holds replacing $\gamma_i$ with $1/\gamma_i$.}

 The analysis for the convergence rate of the inherent belief estimator is similar to the analysis above for showing the convergence of the MLE. We use many of the same symbols in this proof as we did for the proof in the previous section (such as $X(t)$, $Y(t)$, etc.) but these will represent different (though analogous) quantities. 

\iffalse
Note that $\widehat{\phi}_i(t) = \phi_i = 1$ if and only if
\begin{align}
    Z(t) \eqdef t \bar \beta_i(t) - \sum_{\tau = 0}^{t-1} \mu_i(\tau) > 0 \,.
\end{align}
As before, let $z(t) = Z(t) - Z(t-1)$ and let $X(t)$ be the cumulative predictable expected value and $W(t)$ be the cumulative predictable quadratic variation of $Z(t)$:
\begin{align}
    x(t) \eqdef \bbE[z(t)|\cH_t] &\text{ and } X(t) \eqdef \sum_{\tau=1}^t x(\tau)
    \\ w(t) \eqdef \var[z(t)|\cH_t] &\text{ and } W(t) \eqdef \sum_{\tau=1}^t w(\tau)\,.
\end{align}

\begin{proposition}
\label{prop::estimator_convergence_function}
If $X(t) > g(t)$, and $X(t_0) > 2$ for some $t_0$, then
\begin{align}
\bbP\left[\exists t \geq t^* : \widehat{\phi}_i(t) \neq \phi_i\right] \leq \delta
\end{align}
holds for
\begin{align}
t^* \geq g^{-1}\left( \frac{1}{\xi_i} \log \frac{1}{\delta (e^{\xi_i} - 1)} \right)
\end{align}
where
\begin{align}
\xi_i \frac{1}{4 c_i + 2/3} \label{eq::def_for_xi_i}
\end{align}
and $c_i = \frac{\gamma_i}{\gamma_i - 1}$.
\end{proposition}

\begin{proof}
First, note that $g$ must be a monotonic function. 

Then we create a martingale $Y(t) = X(t) - Z(t)$ (which has the same cumulative expected value). Explicitly computing $x(t)$ as a function of $\mu_i(t)$ and applying the fact that $\mu_i(t) \in [\mindeg/t, 1-\mindeg/t]$ then yields $X(t) > c_2 \log (t)$ for some constant $c_2$ (which depends on $\gamma_i$). This is because $f(\mu_i(t),\gamma_i)$ and $1 - f(\mu_i(t),\gamma_i)$ is at minimum of order $1/t$.

Explicitly computing $w(t)$ and $x(t)$ also yields $w(t) \leq c_1 x(t)$ for the constant $c_1 \eqdef \frac{\gamma_i}{\gamma_i-1}$. As before, we then set $t_0 = e^{2/c_2}$, i.e. a value such that $X(t_0) \geq 2$ is guaranteed. Hence, for $t > t_0$, we have $X(t)-1 \geq X(t)/2 > W(t)/(2c_1)$, and so if $Y(t) > X(t)$ (in which case $Z(t) < 0$ and the estimator fails) we get that there must be an integer $s$ such that $W(t)/(2c_1) < s < Y(t)$. We then let event $A_s$ for integer $s$ be
\begin{align}
A_s \eqdef \{\exists t > t_0 :  Y(t) > s, W(t) < 2c_1 s\} 
\end{align}
We then use Freedman's Inequality (\Cref{thm::freedman}) which bounds the probability of $A_s$ above as exponentially small in $s$. Next, we find an $s_0$ such that the probability that $A_s$ is true for any integer $s > s_0$ is less than $\delta$. We then relate this value of $s_0$ to $t$ using that $X(t) > s_0 + 1$ and $X(t) > c_2 \log (t)$.
\else

Note that $\widehat{\phi}_i(t) = \phi_i = 1$ if and only if
\begin{align}
    Z(t) \eqdef (\countrandom) \bar \beta_i(t) - \sum_{\tau = \startindexminus}^{t-1} \mu_i(\tau) > 0 \,.
\end{align}

Then $Z(t)$ is a stochastic process with differences
\begin{align}
    z(t) &\eqdef Z(t) - Z(t-1) = \bbI\{\psi_{i, t} = 1 \} - \mu_i(t-1) \,.
\end{align}
We then make a martingale $Y(t)$ as in \Cref{sec::freedman}: first, the expected updates and (cumulative) predictable expected value
\begin{align}
x(t) & \eqdef \bbE[z(t)| \cH_{t-1}] \text{ and } X(t) \eqdef \sum_{\tau=\startindex}^t x(\tau)\,.
\end{align}
\ifarxiv 
We then derive $x(t)$ as
\begin{align}
x(t) &= f(\mu_i(t-1), \gamma_i) - \mu_i(t-1)
\\ & = \frac{(\gamma_i - 1) \mu_i(t-1) (1 - \mu_i(t-1))}{1 + (\gamma_i - 1) \mu_i(t-1)}
\label{eq::compute_xt_for_estimator}
\end{align}
\fi 

Note that since $\gamma_i > 1$ and $\mu_i(t-1) \in (0,1)$, we have $x(t) > 0$ for all $t$, and hence $X(t)$ is increasing and $Z(t)$ is a submartingale.

\begin{proposition}
\label{prop::estimator_convergence_function}
If $X(t) > g(t)$ for some increasing function $g$, and $X(t_0) > 2$ for some $t_0$, then
\begin{align}
\bbP\left[\exists t \geq t^* : \widehat{\phi}_i(t) \neq \phi_i\right] \leq \delta
\end{align}
holds for
\begin{align}
t^* &\geq \max \left\{ g^{-1}\left( \frac{1}{\xi_i} \log \frac{1}{\delta (e^{\xi_i} - 1)} \right), t_0 \right\}
\\ \xi_i &=\frac{1}{4 c_i + 2/3} \label{eq::def_for_xi_i}
\end{align}
and $c_i = \frac{\gamma_i}{\gamma_i - 1}$.
\end{proposition}

\ifarxiv 

\begin{proof}
First, note that $g$ must be a monotonic function. 
We define
\begin{align}
    y(t) & \eqdef x(t) - z(t) \text{ and } Y(t) \eqdef \sum_{\tau = \startindex}^t y(\tau) = X(t) - Z(t)\,.
\end{align}

Then $Y(t)$ is a martingale, as $x(t) = \bbE[z(t)|\cH_{t-1}] \implies \bbE[y(t)|\cH_{t-1}] = 0$. Furthermore, the martingale $Y(t)$ has bounded step sizes%
\ifarxiv
:
\begin{align}
|y(t)| &\leq \bbE[\bbI\{\psi_{i, t} = 0\} - \mu_i(t-1)  ] \eqbreakshort &~~ - (\bbI\{\psi_{i, t} = 0 \} - \mu_i(t-1))
\\ & = \bbE[\bbI\{\psi_{i, t} = 0\} ] - \bbI\{\psi_{i, t} = 0\} 
\\ & \leq 1
\end{align}
\else 
.
\fi 
%
We define the predictable quadratic variation as:
\begin{align}
    w(t) &\eqdef \var[y(t) | \cH_{t-1}] \text{ and } W(t) \eqdef \sum_{\tau = \startindex}^t w(\tau)
\end{align}
Noting that (given $\cH_{t-1}$) the value $\mu_i(t)$ is fixed yields
\begin{align}
w(t) &= \var[\psi_{i, t}| \cH_{t-1}]
 = \frac{\gamma_i \mu_i(t-1) (1 - \mu_i(t-1))}{(1 + (\gamma_i - 1) \mu_i(t-1))^2}
\end{align}
This then implies bounds on the ratio between $w(t)$ and $x(t)$:

\ifarxiv
\begin{align}
&\frac{w(t)}{x(t)} =  \frac{\gamma_i}{\gamma_i - 1}\frac{1}{1 + (\gamma_i - 1) \mu_i(t-1)} \in \Big(\frac{1}{\gamma_i-1}, \frac{\gamma_i}{\gamma_i-1} \Big)
\\ &\implies \frac{1}{\gamma_i - 1}x(t) \leq w(t) \leq \frac{\gamma_i}{\gamma_i - 1}x(t)
\end{align}
\else
\begin{align}
 \frac{1}{\gamma_i - 1}x(t) \leq w(t) \leq \frac{\gamma_i}{\gamma_i - 1}x(t)
\end{align}
\fi

Since $c_i = \frac{\gamma_i}{\gamma_i - 1}$, this gives
\ifarxiv
\begin{align}
    w(t) &\leq c_i x(t)
    \\ \implies W(t) &\leq c_i \sum_{\tau = \startindex}^t x(\tau) \leq c_i X(t) 
\end{align}
\else
\begin{align}
 W(t) &\leq c_i \sum_{\tau = \startindex}^t x(\tau) \leq c_i X(t) 
\end{align}
\fi
Rewriting as $W(t)/c_i \leq X(t)$, we note that $X(t) \geq 2$ implies $W(t)/(2c_i) \leq X(t) - 1$; this means that when $X(t) \geq 2$, if $Y(t) > X(t)$ there must be some integer $s$ such that $W(t)/(2c_i) < s < Y(t)$.

Let $t_0$ be when $X(t_0) \geq 2$.
For any $t > t_0$, the estimator being wrong then implies
\begin{align}
Z(t) < 0 &\iff X(t) - Y(t) < 0
\\ &\iff Y(t) > X(t)
\\ &\implies \exists s \in \bbZ_{+}: X(t) - 1 < s < Y(t)
\\ &\implies \exists s\in \bbZ_{+}:  W(t)/(2c_i) < s < Y(t)
\end{align}
Let event $A_s$ for integer $s$ be defined as
\begin{align}
A_s = \{\exists t> t_0: X(t) - 1 < s < Y(t)\}
\end{align}
If $A_s$ occurs, we call $s$ a \emph{separator}; by the above, for any $t > t_0$, if the estimator is wrong at time $t$ there must be a separator corresponding to $t$ (note however that one separator $s$ can work for multiple $t$). We now apply Freedman's Inequality (\Cref{thm::freedman}) to bound the probability that any given $s$ is a separator:
%
%
%
\begin{align}
\bbP[A_s] & = \bbP[\exists t> t_0: X(t) - 1 < s < Y(t)] 
\\ &\leq \bbP[\exists t: Y(t) \geq s, W(t) \leq 2 c_i s  ]  
\\& \leq \exp\left( \frac{-s}{2(2c_i + 1/3)}\right)\,.
\end{align}
However, given some $s_0$, we want to bound the probability that \emph{any} integer $s > s_0$ is a separator. We thus define
\begin{align}
B_{s_0} &= \{\not \exists \text{ } s \in \bbZ > s_0
\text{ such that } A_s \text{ holds}\} \,.
\end{align}
Using \eqref{eq::def_for_xi_i} gives:
\begin{align}
1 - \bbP[B_{s_0}] &= \sum_{s = s_0+1}^{\infty}  \bbP[A_s] 
\leq \sum_{s = s_0+1}^{\infty} \exp\left( -\xi_i s\right)
\\ & = \frac{e^{- \xi_i (s_0+1) }} {e^{\xi_i} - 1}\,.
\end{align}
If $t$ is such that
\begin{align}
s_0+1 &\leq g(t) \leq X(t) 
 \implies
t \geq g^{-1}(s_0 + 1)\label{eq::bound_t_s0_invert_log}
\end{align}
and $B_{s_0}$ holds,
then $Z(t) > 0$ and the estimator is correct.
Thus, if \eqref{eq::bound_t_s0_invert_log} holds, then
\begin{align}
\bbP[Z(t) < 0] &\leq 1 - \bbP[B_{s_0}] 
\leq \frac{e^{-\xi_i (s_0+1)}}{e^{\xi_i} - 1}\,.
\end{align}

If we want
\begin{align}
\bbP[Z(t) < 0] \leq \frac{e^{-\xi_i (s_0+1)}}{e^{\xi_i} - 1} \leq \delta
\end{align}
then
\begin{align}
e^{-\xi_i (s_0+1)} &\leq \delta (e^{\xi_i} - 1)
\\ \implies s_0 + 1 &\geq \frac{1}{\xi_i} \log \frac{1}{\delta (e^{\xi_i} - 1)}\,.
\end{align}
Combining this with 
\eqref{eq::bound_t_s0_invert_log}
gives
\begin{align}
&g(t) \geq \frac{1}{\xi_i} \log \frac{1}{\delta (e^{\xi_i} - 1)}\,.
\end{align}

\end{proof}

\else 
\begin{proof}[Proof Sketch]

\revise{
We define
$y(t) \eqdef x(t) - z(t)$ and $ Y(t) \eqdef \sum_{\tau = \startindex}^t y(\tau) = X(t) - Z(t)$
where $Y(t)$ is a bounded step martingale. We define the predictable quadratic variation as
$ w(t) \eqdef \var[y(t) | \cH_{t-1}]$ and  $W(t) \eqdef \sum_{\tau = \startindex}^t w(\tau)$
and derive that
\begin{align}
 \frac{1}{\gamma_i - 1}x(t) \leq w(t) \leq \frac{\gamma_i}{\gamma_i - 1}x(t)
\end{align}
Since $c_i = \frac{\gamma_i}{\gamma_i - 1}$, this gives
$
 W(t) \leq c_i \sum_{\tau = \startindex}^t x(\tau) \leq c_i X(t) 
$, which for $t > t_0$ means that $W(t)/(2c_i) \leq X(t)/2 < X(t)-1$ and so there is an integer $s$ such that $W(t)/(2c_i) \leq s < X(t)$. Thus, if $Y(t) > X(t)$ for some $t \geq t_0$, there must be some integer $s$ such that $W(t)/(2c_i) < s < Y(t)$.
Let event $A_s$ for integer $s$ be
$A_s = \{\exists t> t_0: W(t)/(2c_i) < s < Y(t)\}\,.
$
Freedman's Inequality gives that
\begin{align}
\bbP[A_s]  \leq \exp\left( \frac{-s}{2(2c_i + 1/3)}\right) = \exp\left( -\xi_i s\right)\,.
\end{align}
Define
$B_{s_0} = \{\text{there is no } s \in \bbZ > s_0
\text{ such that } A_s \text{ holds}\}$. First, note that
\begin{align}
    \sum_{s=0}^\infty \bbP[A_s] \leq \sum_{s=0}^\infty \exp(-\xi_i s) < \infty 
\end{align}
so by Borel-Cantelli, only finitely many events $A_s$ can happen (almost surely). Thus $B_{s_0}$ holds (almost surely) for some $s_0$.
We can then also compute 
\begin{align}
1 - \bbP[B_{s_0}] &\leq \sum_{s = s_0+1}^{\infty} \exp\left( -\xi_i s\right) = \frac{e^{- \xi_i (s_0+1) }} {e^{\xi_i} - 1}\,.
\end{align}
Then, if $B_{s_0}$ holds, for all $t$ such that
\begin{align}
s_0+1 &\leq g(t) \leq X(t) 
 \implies
t \geq g^{-1}(s_0 + 1)\label{eq::bound_t_s0_invert_log}
\end{align}
we know that $Z(t) > 0$ and the estimator is correct.
If we want
$\bbP[Z(t) < 0] \leq 1 - \bbP[B_{s_0}] \leq \delta$
then this 
gives
\begin{align}
g(t) \geq \frac{1}{\xi_i} \log \frac{1}{\delta (e^{\xi_i} - 1)}\,.
\end{align}
}
\end{proof}
\fi 

To use \Cref{prop::estimator_convergence_function}, we need to determine how quickly $\mu_i(t)$ or $\beta_i(t)$ approaches $0$.
Before doing a calculation of this, we show what happens if we use a worst-case bound on $\mu_i(t)$. Since $\mu_i(t-1) \in [\mindeg/t, 1-\mindeg/t]$, 
\ifarxiv 
we have
\begin{align}
    x(t) \geq \frac{\gamma_i-1}{\gamma_i} \frac{1}{2} \frac{\mindeg}{t}
\end{align}
as $\mu_i(t) (1-\mu_i(t)) \geq \frac{1}{2} \min(\mu_i(t),1-\mu_i(t))$. This 
\else 
this
\fi 
implies that for $t \geq 3$,
\begin{align} \label{eq::X-bound-below-inherent}
    X(t) \geq \frac{\gamma_i-1}{\gamma_i} \frac{1}{4} \mindeg \log(t) \,.
\end{align} 

\ifarxiv 
Using \eqref{eq::X-bound-below-inherent}, we know that $X(t) \geq 2$ when 
\begin{align}
    t \geq e^{8 \frac{\gamma_i}{\gamma_i-1}\mindeg^{-1}} \eqdef t_0\,.
    \label{eq::bound_t_from_Xt2}
\end{align}
\fi 
%
Then the estimator has converged with probability $\geq 1-\delta$ for all $t$ such that
\ifarxiv 
\begin{align}
&\frac{\gamma_i - 1}{4\gamma_i}\mindeg \log(t) \geq \frac{1}{\xi_i} \log \frac{1}{\delta (e^{\xi_i} - 1)} \label{eq::rate-estimation-equation}
\\ \implies & 
t \geq \left(\frac{1}{\delta(e^{\xi_i} - 1)} \right)^{\frac{4}{\mindeg} \frac{\gamma_i}{\gamma_i - 1} \left(4 \frac{\gamma_i}{\gamma_i - 1} + \frac{2}{3} \right) } = \Theta\left((1/\delta)^c\right)\label{eq::rate_worst_case_mu}
\end{align}
\else 
\begin{align}
t \geq \left(\frac{1}{\delta(e^{\xi_i} - 1)} \right)^{\frac{4}{\mindeg} \frac{\gamma_i}{\gamma_i - 1} \left(4 \frac{\gamma_i}{\gamma_i - 1} + \frac{2}{3} \right) } = \Theta\left((1/\delta)^c\right)\label{eq::rate_worst_case_mu}
\end{align}
\fi 
%
where $c$ is a constant depending on $\mindeg$ and $\gamma_i$.
Thus, we know that for some $t^*$ on the order of $(1/\delta)^c$, the inherent belief estimator converges by time $t^*$ with probability at least $1-\delta$. 

However, since \eqref{eq::X-bound-below-inherent} is a lower bound, corresponding to using the lower bound of $\Theta(1/t)$ for $\mu_i(t)$, the computed convergence rate \eqref{eq::rate_worst_case_mu} is too low. This raises the question of improving it by using better bounds on $\mu_i(t)$, thus yielding a better bound of $X(t)$ in \eqref{eq::X-bound-below-inherent}. This can be divided into two cases: consensus and non-consensus. 

If the system does not approach consensus, 
then $X(t)$ is linear since 
%
for large enough $t$, $x(t)$ will be very close to a constant, and thus $X(t) \geq K t$ for some constant $K$. Using \Cref{prop::estimator_convergence_function} then gives that we converge for all $t$ such that
\begin{align}
t \geq \max\left\{ \frac{2}{K}, \frac{1}{K} \frac{1}{\xi_i} \log \frac{1}{\delta (e^{\xi_i} - 1)}\right\} = \Theta\left( \log \frac{1}{\delta} \right)\label{eq::order_rate_for_interior}
\end{align}
(where $\xi_i = \frac{1}{4c_i + 2/3}$ as defined in \Cref{prop::estimator_convergence_function}).

When the system approaches consensus $X(t)$ will be sublinear as $\mu_i(t) \to 0$ as $t \to \infty$; however, by analyzing the rate of convergence to consensus, we will obtain a better bound on $\mu_i(t)$ than $\Theta(1/t)$, which will yield a more precise convergence rate for the inherent belief estimator. 

\revise{For the consensus case, \Cref{fig::delta_plot} shows an experiment approximating $\delta$ for given $t^*$.}

\begin{figure}
\centering
\includegraphics[scale = .5]{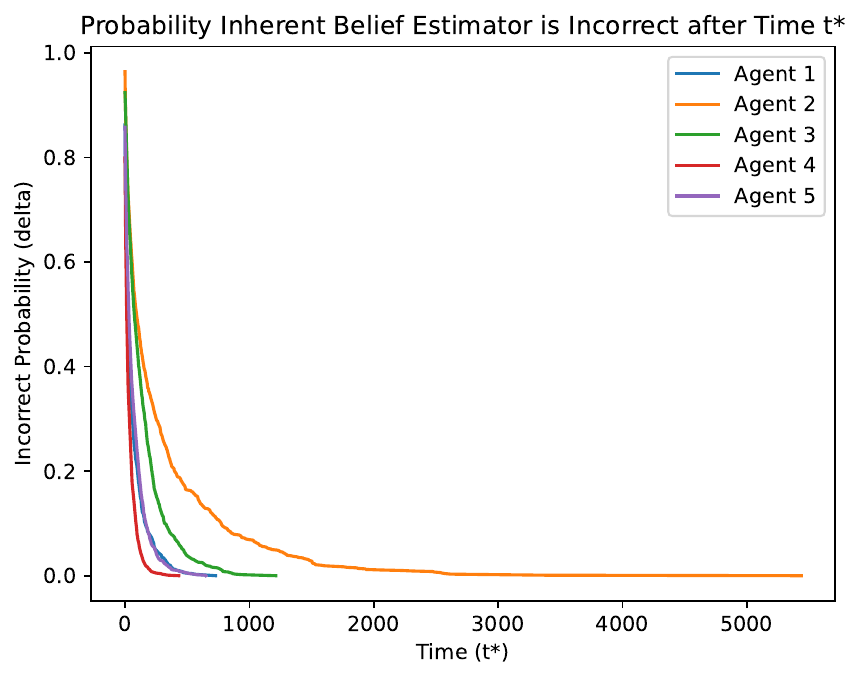}
\caption{\label{fig::delta_plot}
\revise{Empirical plot (over $1000$ experiments), for each agent $i$ of the network in \Cref{fig::5-node-graph}, of the probability $\delta_i(t^*)$ that the MLE inherent belief estimator \eqref{eq::steve} is wrong for some $t \geq t^*$. Each experiment was run for $100000$ steps; the curves are cut off at the first $t^*$ such that \eqref{eq::steve} was always correct for all $t^* \leq t \leq 100000$ (i.e. $\delta = 0$ empirically).}}
\end{figure}


\section{Bounds on Rate of Convergence to Consensus}

\label{sec::rate_converge_consensus}

\revise{In this section, we look at what the rate of convergence to consensus is, i.e. how quickly the time-averaged declared opinions $\beta_i(t)$ for each agent $i$ converges to $0$ or $1$ in the case of consensus.
The ultimate goal is to be able to characterize the probability of error of inherent belief estimator \eqref{eq::steve} using \Cref{prop::estimator_convergence_function}, illustrating how quickly we should expect the estimation of inherent beliefs to be correct.}

Consensus occurs when either $ \lambda_{\max}(\dermat{\bzero})$ or $\lambda_{\max}(\dermat{\bone})$ is less than $1$ 
(see \cite{opiniondynamicsCDC})
. For this section, WLOG suppose that 
$\lambda_{\max}(\dermat{\bzero}) < 1$ meaning that consensus to $\bzero$ occurs. \revise{We will show that in the case of consensus,  the convergence rate of $\beta_i(t)$ and the estimator $\eqref{eq::steve}$ is a function of $ \lambda_{\max}(\dermat{\bzero})$. This eigenvalue is how the network structure influences these quantities.}

We define that $\beta_i(t)$ for agent $i$ converges to $0$ at a rate of $t^r$ if for every constant $\epsilon > 0$, we have that
  \begin{align}
        \lim_{t \to \infty} \frac{\beta_i(t)}{t^{r+\epsilon}} &= 0 
        ~~\text{ and }~~ \lim_{t \to \infty} \frac{\beta_i(t)}{t^{r- \epsilon}} = \infty\,.
    \end{align}
(We note that this definition does exclude subpolynomial terms,  for instance, $t^r \log t$ and $t^r / \log t$ will both satisfy the conditions.) The goal is to show that the rate $r = \lambda - 1$. 

\begin{figure}
    \centering
    \includegraphics[scale = .5]{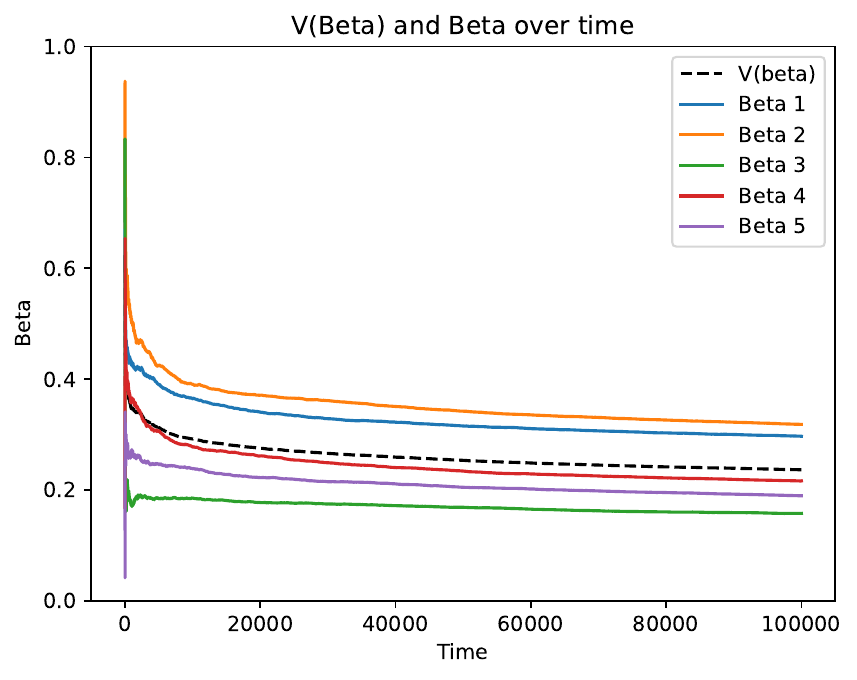}
    \caption{\label{fig::beta_decay}
    \revise{Plot illustrating example evolution of $V(\bbeta(t))$ and $\beta_i(t)$ for each agent $i$ on the network in \Cref{fig::5-node-graph}.}}
\end{figure}

\revise{In this section, first, we define a function $V$ which is a scalar function of $\bbeta$. We define linear processes $h$ which can be used to upper and lower bound $V(\bbeta(t))$. We then normalize process $h$ to show the rate of convergence of $V(\bbeta(t))$. Afterwards, we show that convergence rate of each $\beta_i(t)$ is a constant multiplied by the convergence rate $V(\bbeta(t))$. Finally, using the convergence rate of $\beta_i(t)$ with \Cref{prop::estimator_convergence_function} gives a bound on the convergence rate of inherent belief estimate \eqref{eq::steve}. An illustration of the decay rate of $\beta_i(t)$ and of $V(\bbeta(t))$ for a given experiment is shown in \Cref{fig::beta_decay}.}

\revise{\subsection{Definition of $V$ and Linearized Processes}}

Let $\lambda = \lambda_{\max}(\dermat{\bzero})$.
 Let $\bv$ be the associated (left) eigenvector of eigenvalue $\lambda$. Since $\dermat{\bzero}$ is irreducible (since $\bW$ is irreducible) and a nonnegative matrix, the 
Perron-Frobenius theorem \cite{Berman1994} implies that  
$\bv$ is a positive eigenvector. Scale $\bv$ so that $\bv^\ltop \bone = 1$. 
Let
\begin{align}
\label{eq::V_linear_def}
V(\bbeta) = \bv ^{\ltop} \bbeta\,.
\end{align}

For convenience, we represent the declared opinions at a given step $t$ as a (random) vector $\bpsi(\bbeta(t))$ depending on the state $\bbeta(t)$. Thus the dynamics follow the update
\begin{align}
\label{eq::beta_formula_psi}
    \bbeta(t+1) &= \frac{t}{t+1}\bbeta(t) + \frac{1}{t+1} \bpsi(\bbeta(t))
\end{align}
and $\bpsi(\bbeta(t))$ is a vector with $i$th component given by
\begin{align}
    \psi_i(\bbeta(t)) = \begin{cases}
    1 & \text{w.p. }  f(\mu_i(t), \gamma_i) \\
    0 & \text{w.p. }  1 - f(\mu_i(t), \gamma_i)
    \end{cases}\,.
\end{align}

Since consensus to $\bzero$ occurs, this implies that $\lim_{t\to \infty}V(\bbeta(t)) = 0$.
\ifarxiv
We will show that the rate at which $V(\bbeta(t))$ approaches $0$ is also the rate at which each $\beta_i(t) \to 0$ up to a constant factor.
\fi

To analyze the convergence of $V(\bbeta(t))$, we will define linear functions which we can use to upper and lower bound the value of $V(\bbeta(t)$ when $\bbeta(t)$ is in a certain region.



\newcommand{\consttwo}{\ensuremath{\zeta}}

For $\consttwo > 0$, let random vector $\bar{\bpsi}(\bbeta(t), \consttwo)$ be such that the $i$th component is given by 
\begin{align}
    \bar{\psi}_i(\bbeta(t), \consttwo) = \begin{cases} \begin{cases} 1 &\text{w.p. } \consttwo \gamma_i \mu_i(t) \\ 0 &\text{w.p. } 1 - \consttwo \gamma_i \mu_i(t) \end{cases} &\text{if } \consttwo \gamma_i \mu_i(t) \leq 1 \\ \consttwo \gamma_i \mu_i(t)  &\text{if } \consttwo \gamma_i \mu_i(t) > 1 \end{cases} \label{eq::linearized-declaration}
\end{align}
(The case where $\consttwo \gamma_i \mu_i(t) > 1$ is a technicality we need to consider)
where $\bar{\bpsi}(\bbeta(t), \consttwo)$ is also maximally coupled with ${\bpsi}(\bbeta(t))$. 
 This means that the joint distribution between $\bar{\bpsi}(\bbeta(t), \consttwo)$ and ${\bpsi}(\bbeta(t))$ is such that
\begin{align}
\bbP[\bar{\psi}_i(\bbeta(t), \consttwo) = {\psi}_i(\bbeta(t))]
\end{align}
is maximized. 
\ifarxiv 
Note that by \eqref{eq::linearized-declaration},
\begin{align}
    \bbE[\bar{\psi}_i(\bbeta(t),\consttwo) \,|\, \mu_i(t)] &= \consttwo \gamma_i \mu_i(t) ~\text{ and}
    \\ \var[\bar{\psi}_i(\bbeta(t),\consttwo) \,|\, \mu_i(t)] &\leq \consttwo \gamma_i \mu_i(t) 
    \label{eq::variance_bound_psi}
    \,.
\end{align}
where the variance term follows from
\begin{align}
    \var[\bar{\psi}_i(\bbeta(t),\consttwo) \,|\, \mu_i(t)] &= \max(\consttwo \gamma_i \mu_i(t) (1 - \consttwo \gamma_i \mu_i(t)), 0) \,.
\end{align}
\else 
\fi 

\begin{definition}[Linearized Process]
Given $t_0$ and $\bbeta(t_0)$, for each constant value $\alpha > 0$ define the stochastic process $h^{\alpha}(t)$ for $t\geq t_0$ to have the following \revisetwo{values which determine its joint distribution with random stochastic process $\bbeta(t)$}:
\begin{align}
h^{\alpha}(t_0) &= V(\bbeta(t_0))
\\h^{\alpha}(t+1) 
&= \frac{t}{t+1}  h^{\alpha}(t) + \frac{1}{t+1} \bv^{\ltop}\bar{\bpsi}\left(\bbeta(t), \alpha \frac{h^{\alpha}(t)}{V(\bbeta(t))}\right)\,.
\end{align}
\end{definition}

While in the above $h^{\alpha}(t)$ is defined with respect to random process $\bbeta(t)$, the marginal distribution of $h^{\alpha}(t)$ will have expectation and variance bounds that does not depend on $\bbeta(t)$.
The processes $h^{\alpha}(t)$ will be used as linear upper and lower bounds to $V(\bbeta(t))$. 
Let us define
\begin{align}
R(t + 1, \constone) = \prod_{\tau = 0}^t \frac{\tau+\constone}{\tau+1}
\end{align}
for $\constone \in (0,1)$ and time $t$. We then get the following result:

\begin{lemma} \label{lem::bounds-on-R}
\begin{align}
\frac{1}{\Gamma(\constone) (t+1)^{1 - \constone}} \leq R(t, \constone) \leq \frac{1}{\Gamma(\constone) (t)^{1 - \constone}}
\end{align}
\end{lemma}

\ifarxiv 
\begin{proof}
First,
\begin{align}
R(t, \constone) = \frac{\Gamma(t+\constone)}{\Gamma(\constone) \Gamma(t+1)}\,.
\end{align}
Since $\constone \in (0,1)$, using Gautschi's inequality gives that
\begin{align}
\frac{1}{(t+1)^{1-\constone}}\leq \frac{\Gamma(t+\constone)}{\Gamma(t+1)} \leq \frac{1}{t^{1-\constone}}
\end{align}
Including the $\Gamma(\constone)$ term in the denominator gives the result.
\end{proof}
\else 
\begin{proof}
This follows from Gautschi's inequality.
\end{proof}
\fi 

\begin{lemma}\label{lem::prop_h-t}
Some properties of $h^\alpha(\cdot)$ are:
\begin{enumerate}[(a)]
\item 
\label{item::vb_bounded_h}
For any $\epsilon > 0$, there is a time $t_0$, where for $t>t_0$, $\bbeta(t)$ remains in a disc of radius $r(\epsilon)$ from $\bzero$. Then the trajectory $V(\bbeta(t))$ is such that there exists an $1<\alpha_{+}<1+\epsilon$  and $1 - \epsilon < \alpha_{-} <1$ such that
\begin{align}
h^{\alpha_{-}}(t)\leq V(\bbeta(t)) \leq h^{\alpha_{+}}(t)
\end{align}
for $t>t_0$.

\item  \label{item::h-t-expectation}
For $t \geq t_1$ and any $\alpha > 0$ where $\alpha \lambda < 1$,
\begin{align}
\bbE[h^{\alpha}(t) | h^{\alpha}(t_1)] &= \frac{R(t, \alpha\lambda)}{R(t_1, \alpha\lambda)} h(t_1) 
\\ &= (1 + o(1)) \frac{t^{\alpha\lambda - 1}}{t_1^{\alpha\lambda - 1}} h(t_1)
\end{align}
\item  \label{item::h-t-variance}
For sufficiently large $t_1$ there is a constant $c$ where for $t>t_1$ and any $\alpha > 0$,
\begin{align}
\var[h^{\alpha}(t) | h^{\alpha}(t_1)] \leq c \frac{t^{2 \alpha\lambda - 2}}{t_1^{2\alpha\lambda - 1}} h^{\alpha}(t_1)
\end{align}
\end{enumerate}
\end{lemma}

\ifarxiv 

\begin{proof}
\emph{Proof of Part (\ref{item::vb_bounded_h})}:

From \cite[Theorem 2]{opiniondynamicsCDC}, we know that $\bbeta(t)$ almost surely converges to an equilibrium point. In this case, that equilibrium point is $\bzero$ and thus there must be some time $t_0$ after which which $\bbeta(t)$ lies in a disc around $\bzero$.
When $\bbeta(t)$ is within a disc of radius $r$ of $\bzero$, each $\mu_i(t)$ must also be within $r$ of $0$.

There are constants $\alpha_{-}$ and $\alpha_{+}$ so that for each agent $i$,
\begin{align}
\alpha_{-} \gamma_i \mu_i \leq f(\mu_i, \gamma_i) \leq \alpha_{+} \gamma_i \mu_i
\end{align}
for all $\mu_i \leq r$.
(The values $\alpha_{-}$  and $\alpha_{+}$ get closer to $1$ as $r$ decreases.) 
We focus on showing $h^{\alpha_{-}}(t)\leq V(\bbeta(t))$ (showing the other case is symmetric).

First, by definition $h^{\alpha_{-}}(t_0) = V(\bbeta(t_0))$.
If we assume that $h^{\alpha_{-}} (t)\leq V(\bbeta(t))$, then
\begin{align}
&h^{\alpha_{-}}(t+1) 
\\&= \frac{t}{t+1}  h^{\alpha_{-}}(t) + \frac{1}{t+1} \bv^{\ltop}\bar{\bpsi}\left(\bbeta(t), {\alpha_{-}} \frac{h^{\alpha_{-}}(t)}{V(\bbeta(t))}\right)
\\&\leq \frac{t}{t+1}  V(\bbeta(t)) + \frac{1}{t+1} \bv^{\ltop}\bar{\bpsi}\left(\bbeta(t), {\alpha_{-}} \frac{h^{\alpha_{-}}(t)}{V(\bbeta(t))}\right)
\end{align}
Since $h^{\alpha_{-}} (t)\leq V(\bbeta(t))$,  random variable
\begin{align}\label{eq::tilde_psi_i_alpha+minus}
\bar{\psi}_i\left(\bbeta(t), {\alpha_{-}} \frac{h^{\alpha_{-}}(t)}{V(\bbeta(t))}\right)
\end{align}
 is $1$ with probability 
 \begin{align}
\alpha_{-} \frac{h^{\alpha_{-}}(t)}{V(\bbeta(t))}\gamma_i \mu_i(t) \leq \alpha_{-} \gamma_i \mu_i(t) \leq f(\mu_i, \gamma_i)
\end{align}
where $f(\mu_i, \gamma_i)$ is the probability of $\psi_{i}(\bbeta(t)) = 1$. Because the two variables are maximally coupled, every instance when \eqref{eq::tilde_psi_i_alpha+minus} is $1$, it must also be that $\psi_{i}(\bbeta(t)) = 1$. Then
\begin{align}
\bar{\psi}_i\left(\bbeta(t), {\alpha_{-}} \frac{h^{\alpha_{-}}(t)}{V(\bbeta(t))}\right) \leq \psi_{i}(\bbeta(t))
\end{align}
and thus
\begin{align}
h^{\alpha_{-}}(t+1)
&\leq \frac{t}{t+1}  V(\bbeta(t)) + \frac{1}{t+1} \bv^{\ltop} \bpsi(\bbeta(t))
\\& = V(\bbeta(t+1))\,.
\end{align}

\noindent \emph{Proof of Part (\ref{item::h-t-expectation})}:
We will use proof by induction. For the base case we trivially have
\begin{align}
\bbE&[h^{\alpha}(t_1) | h^{\alpha}(t_1)] =  h^{\alpha}(t_1) = \frac{R(t_1, \alpha\lambda)}{R(t_1, \alpha\lambda)}  h^{\alpha}(t_1) \,.
\end{align}

For the inductive step, we can assume that
\begin{align}
\bbE&[h^{\alpha}(t) | h^{\alpha}(t_1)] = \frac{R(t, \alpha\lambda)}{R(t_1, \alpha\lambda)} h^{\alpha}(t_1) \,.
\end{align}
Then we consider $t+1$:
\ifsixteen
\begin{align}
&\bbE[h^{\alpha}(t+1) | h^{\alpha}(t_1)] 
= \frac{t}{t+1} \bbE[h^{a}(t)| h^{\alpha}(t_1)]
\eqlinebreakshort
+ \frac{1}{t+1}\bbE\left[\bv^{\ltop}\bar{\bpsi}\left(\bbeta(t), {\alpha} \frac{h^{\alpha}(t)}{V(\bbeta(t))}\right)\bigg| h^{\alpha}(t_1)\right]\,.
\end{align}
The expectation in the second term is equivalent to 
\begin{align}
& \bbE\left[\bv^{\ltop}\frac{h^{\alpha}(t)}{V(\bbeta(t))}\alpha \bGamma \bW \bbeta(t) \bigg| h^{\alpha}(t_1)\right]
\\ & = \bbE\left[\frac{h^{\alpha}(t)}{V(\bbeta(t))}{\alpha}\lambda \bv^{\ltop}\bbeta(t) \bigg| h^{\alpha}(t_1)\right]
\\& =  {\alpha}\lambda  \bbE[h^{\alpha}(t)| h^{\alpha}(t_1)]
\end{align}
which gives that
\begin{align}
\bbE&[h^{\alpha}(t+1) | h^{\alpha}(t_1)] 
\\& = \frac{t}{t+1} \bbE[h^{\alpha}(t)| h^{\alpha}(t_1)]
+ \frac{1}{t+1} {\alpha}\lambda  \bbE[h^{\alpha}(t)| h^{\alpha}(t_1)]
\\&= \frac{R(t + 1, \alpha\lambda)}{R(t_1, \alpha\lambda)} h(t_1) \,.
\end{align}
\else
\begin{align}
&\bbE[h^{\alpha}(t+1) | h^{\alpha}(t_1)] 
= \frac{t}{t+1} \bbE[h^{a}(t)| h^{\alpha}(t_1)]
\eqlinebreakshort
+ \frac{1}{t+1}\bbE\left[\bv^{\ltop}\bar{\bpsi}\left(\bbeta(t), {\alpha} \frac{h^{\alpha}(t)}{V(\bbeta(t))}\right)\bigg| h^{\alpha}(t_1)\right]
\\& = \frac{t}{t+1}  \bbE[h^{a}(t)| h^{\alpha}(t_1)]
\eqlinebreakshort
+ \frac{1}{t+1}\bbE\left[\bv^{\ltop}\frac{h^{\alpha}(t)}{V(\bbeta(t))}\alpha \bGamma \bW \bbeta(t) \bigg| h^{\alpha}(t_1)\right]
\\& = \frac{t}{t+1}  \bbE[h^{a}(t)| h^{\alpha}(t_1)]
\eqlinebreakshort
+ \frac{1}{t+1}\bbE\left[\frac{h^{\alpha}(t)}{V(\bbeta(t))}{\alpha}\lambda \bv^{\ltop}\bbeta(t) \bigg| h^{\alpha}(t_1)\right]
\\& = \frac{t}{t+1} \bbE[h^{\alpha}(t)| h^{\alpha}(t_1)]
+ \frac{1}{t+1} {\alpha}\lambda  \bbE[h^{\alpha}(t)| h^{\alpha}(t_1)]
\\&= \frac{R(t + 1, \alpha\lambda)}{R(t_1, \alpha\lambda)} h(t_1) 
\end{align}
\fi

%
%
Then by \Cref{lem::bounds-on-R}
\ifsixteen
we have
\else
, we have
\begin{align}
    \frac{(t+1)^{\alpha\lambda-1}}{t_1^{\alpha\lambda-1}} \leq \frac{R(t, \alpha\lambda)}{R(t_1, \alpha\lambda)} \leq \frac{t^{\alpha\lambda-1}}{(t_1+1)^{\alpha\lambda-1}}
\end{align}
Thus, 
\fi
for any $\epsilon > 0$, there is a sufficiently large $t_1$ such that for all $t \geq t_1$,
\begin{align}
    (1-\epsilon) \frac{t^{\lambda-1}}{t_1^{\lambda-1}} h^{\alpha}(t_1) \leq \bbE[h^{\alpha}(t) \,|\, h^{\alpha}(t_1)] \leq (1+\epsilon) \frac{t^{\lambda-1}}{t_1^{\lambda-1}} h^{\alpha}(t_1)\,.
\end{align}

\noindent \emph{Proof of Part (\ref{item::h-t-variance})}:

When we apply the law of total variance, we get that
\begin{align}
\var&[h^{\alpha}(t+1)| h^{\alpha}(t_1)] 
\\& = \bbE\left[\var[h^{\alpha}(t+1)| \cH_t, h^{\alpha}(t_1)] \big|h^{\alpha}(t_1) \right] 
\eqlinebreakshort
+ \var\left[ \bbE[h^{\alpha}(t+1)| \cH_t, h^{\alpha}(t_1)] \big| h^{\alpha}(t_1) \right] 
\label{eq::var-h-t-totalvar-result}
\,.
\end{align}

We compute the first term in \eqref{eq::var-h-t-totalvar-result} by starting with 
\begin{align}
\var&\left[\bar{\psi}_i\left(\bbeta(t), {\alpha} \frac{h^{\alpha}(t)}{V(\bbeta(t))}\right) \bigg| \cH_t,  h^{\alpha}(t_1)\right] 
\\ &\leq  {\alpha} \frac{h^{\alpha}(t)}{V(\bbeta(t))} \gamma_i \mu_i(t)\,.
\end{align}
where we used \eqref{eq::variance_bound_psi}.

Let $c_0 = \max_i\{v_i\}$. Then we can compute
\begin{align}
&\var[h^{\alpha}(t+1)| \cH_t, h^{\alpha}(t_1)] 
\\&= \var\left[\frac{1}{t+1} \bv^{\ltop}\bar{\bpsi}\left(\bbeta(t), \alpha \frac{h^{\alpha}(t)}{V(\bbeta(t))}\right) \bigg  | \cH_t, h^{\alpha}(t_1)\right]
\\&= \frac{1}{(t+1)^2}\sum_{i = 1}^{n} v_i^2 \var\left[\bar{\psi}_i\left(\bbeta(t), {\alpha} \frac{h^{\alpha}(t)}{V(\bbeta(t))}\right) \bigg| \cH_t,  h^{\alpha}(t_1)\right] 
\\& \leq \frac{1}{(t+1)^2}\sum_{i = 1}^{n} v_i^2 {\alpha} \frac{h^{\alpha}(t)}{V(\bbeta(t))} \gamma_i \mu_i(t)
\\& \leq \frac{1}{(t+1)^2} {\alpha} \frac{h^{\alpha}(t)}{V(\bbeta(t))} c_0 \lambda \bv^{\ltop} \bbeta(t)
 = \frac{c_0{\alpha}   \lambda}{(t+1)^2}  h^{\alpha}(t)\,.
\end{align}

To finish computing the first term in \eqref{eq::var-h-t-totalvar-result}, we have 
\begin{align}
\bbE&\left[\var[h^{\alpha}(t+1)| \cH_t, h^{\alpha}(t_1)] \big|h^{\alpha}(t_1) \right]
\\&\leq \bbE\left[\frac{c_0{\alpha}   \lambda}{(t+1)^2}  h^{\alpha}(t_1) \bigg| h^{\alpha}(t_1)  \right]
\\&\leq c_1 \alpha \lambda \frac{1}{t^2}\frac{t^{\alpha \lambda - 1}}{t_1^{\alpha \lambda - 1}} h^{\alpha}(t_1)
\end{align}
where we used the approximation result from \Cref{lem::prop_h-t}(\ref{item::h-t-expectation}) which is a correct upper bound for some constant $c_1$ and sufficiently large $t_1$. 
For the second term in \eqref{eq::var-h-t-totalvar-result}, the expectation inside the variance is 
\begin{align}
\bbE[h^{\alpha}(t+1)| \cH_t, h^{\alpha}(t_1)]
&= \frac{R(t+1,\alpha \lambda)}{R(t,\alpha \lambda)} h^{\alpha}(t)
\end{align}
and thus
\begin{align}
 \var&\left[ \bbE[h^{\alpha}(t+1)| \cH_t, h^{\alpha}(t_1)] \big| h^{\alpha}(t_1) \right]
 \\&= \left( \frac{R(t+1,\alpha \lambda)}{R(t,\alpha \lambda)} \right)^2 \var[h^{\alpha}(t)|h^{\alpha}(t_1)]\,.
\end{align}

Putting this together gets
\begin{align}
\var&[h^{\alpha}(t+1)|h^{\alpha}(t_1)] \leq
c_1 \alpha \lambda \frac{1}{t^2}\frac{t^{\alpha \lambda - 1}}{t_1^{\alpha \lambda - 1}} h^{\alpha}(t_1) 
\eqlinebreakshort
+ \left( \frac{R(t+1,\alpha \lambda)}{R(t,\alpha \lambda)} \right)^2 \var[h^{\alpha}(t)|h^{\alpha}(t_1)]\,.
\end{align}

Telescoping the variance terms yields:
\begin{align}
\var&[h^{\alpha}(t+1)|h^{\alpha}(t_1)] 
\\&\leq c_2 \alpha \lambda\sum_{\tau = t_1}^ {t}  \left(\frac{(t+1)^{\alpha \lambda - 1}}{\tau^{\alpha \lambda - 1}} \right)^2\frac{1}{\tau^2} \frac{\tau^{\alpha \lambda - 1}}{t_1^{\alpha \lambda - 1}} h^{\alpha}(t_1)
\\& = c_2 \alpha \lambda \frac{(t+1)^{2\alpha \lambda - 2}}{t_1^{\alpha \lambda - 1}} h^{\alpha}(t_1) \sum_{\tau = t_1} ^{t} \frac{1}{\tau^{\alpha \lambda + 1}}\,.
\end{align}

Approximating the sum using an integral then gives:
\begin{align}
\sum_{\tau = t_1}^{t}\frac{1}{\tau^{\alpha\lambda+ 1}} &\leq 
\int_{t_1}^{\infty} \frac{1}{\tau^{\alpha\lambda+ 1}}  d \tau
= 
\frac{1}{\alpha\lambda} t_1^{-\alpha\lambda}
\end{align}
which results in
\begin{align}
\var[h^{\alpha}(t) | h^{\alpha}(t_1)] &\leq c_2 \alpha \lambda \frac{t^{2\alpha(\lambda - 1)}}{t_1^{\alpha\lambda - 1}} h^{\alpha}(t_1) \frac{1}{\alpha\lambda} t_1^{-\alpha\lambda}
\\&= c_2 \frac{t^{2(\alpha\lambda - 1)}}{t_1^{\alpha\lambda - 1}} h^{\alpha}(t_1)  t_1^{-\alpha\lambda}\,.
\end{align}
where $c_2$ is a constant not depending on $\alpha$, $\lambda$, or $t$ (so long as $t_1$ is sufficiently large).
\end{proof}

\else 
\begin{proof}[Proof Sketch] 


\revisetwo{For Part (\ref{item::vb_bounded_h}), we know that there are constants $\alpha_{-}$ and $\alpha_{+}$ so that for each agent $i$,
\begin{align}
\alpha_{-} \gamma_i \mu_i \leq f(\mu_i, \gamma_i) \leq \alpha_{+} \gamma_i \mu_i
\end{align}
for all $\mu_i$ within some radius of $\bzero$. For $\alpha_{-}$, we show that $h^{\alpha_{-}}(t)\leq V(\bbeta(t))$ using induction (showing the other case is symmetric). Because of maximal coupling, we can use   
\begin{align}
\bar{\psi}_i\left(\bbeta(t), {\alpha_{-}} \frac{h^{\alpha_{-}}(t)}{V(\bbeta(t))}\right) \leq \psi_{i}(\bbeta(t))
\end{align}
to show the inductive step.}

\revisetwo{The first equation of part (\ref{item::h-t-expectation}) is shown with induction. The key inductive step is that
\begin{align}
&\bbE[h^{\alpha}(t+1) | h^{\alpha}(t_1)] 
= \frac{t}{t+1} \bbE[h^{a}(t)| h^{\alpha}(t_1)]
\eqlinebreakshort
+ \frac{1}{t+1}\bbE\left[\bv^{\ltop}\bar{\bpsi}\left(\bbeta(t), {\alpha} \frac{h^{\alpha}(t)}{V(\bbeta(t))}\right)\bigg| h^{\alpha}(t_1)\right]
\\ & = \frac{t}{t+1}  \bbE[h^{a}(t)| h^{\alpha}(t_1)]
\eqlinebreakshort
+ \frac{1}{t+1}\bbE\left[\frac{h^{\alpha}(t)}{V(\bbeta(t))}{\alpha}\lambda \bv^{\ltop}\bbeta(t) \bigg| h^{\alpha}(t_1)\right]
\\& = \frac{t}{t+1} \bbE[h^{\alpha}(t)| h^{\alpha}(t_1)]
+ \frac{1}{t+1} {\alpha}\lambda  \bbE[h^{\alpha}(t)| h^{\alpha}(t_1)]
\\&= \frac{R(t + 1, \alpha\lambda)}{R(t_1, \alpha\lambda)} h(t_1) 
\end{align}
The second equation uses \Cref{lem::bounds-on-R}.}

\revise{ Part (\ref{item::h-t-variance}) is shown using the law of total variance followed by telescoping the variance terms and approximating with an integral.}
\end{proof}

\fi 

\subsection{Normalized Convergence Process}

\revise{Using the linearized process, we will define a normalized process whose asymptotic convergence properties will be used to show our desired results.}
We will first use the results from \Cref{lem::prop_h-t} to get a bound on the variance of the linearized process in terms of the square of its expected value. For this bound, we need to use that at any time $t$, we expect that
$V(\bbeta(t)) > c \mindeg / t \eqdef  \mindeg^* /t$ (this is discussed in \eqref{sec::bounds_on_mu}).

\begin{lemma} \label{lem::var-expsq-ratio}
    For any sufficiently large $t_0$ and any $t > t_0$,
    \begin{align}
        \frac{\var[h^{\alpha}(t) \,|\, h^{\alpha}(t_0)]}{\bbE[h^{\alpha}(t) \,|\, h^{\alpha}(t_0)]^2} &\leq \frac{c}{h^{\alpha}(t_0) t_0} \leq c^*
    \end{align}
    for some constant $c^*$ which does not depend on $t$ or $t_0$.
\end{lemma}

\begin{proof}
    \begin{align}
        \frac{\var[h^{\alpha}(t) \,|\, h^{\alpha}(t_0)]}{\bbE[h^{\alpha}(t) \,|\, h^{\alpha}(t_0)]^2} &\leq 
        \frac{c\frac{t^{(\alpha\lambda - 1)2}}{t_0^{2\alpha\lambda - 1}} h^{\alpha}(t_0)}{\left(\frac{t^{\alpha\lambda - 1}}{t_0^{\alpha\lambda - 1}} h^{\alpha}(t_0) \right)^2}
        \\= \frac{c \,  t_0^{\alpha\lambda - 1}}{h^{\alpha}(t_0)} \frac{1}{t_0^{\alpha\lambda}}
        &= \frac{c}{h^{\alpha}(t_0) t_0}
        \leq c^*
    \end{align}
    where the last step uses  $h^{\alpha}(t_0) = V(\bbeta(t_0)) =  \Omega(t_0^{-1})$.
\end{proof}

\ifarxiv 
Note that $h^{\alpha}(t_0) \propto t_0^{-1}$ is a (guaranteed, not probabilistic) worst-case bound, and significantly worse than the expected $h^{\alpha}(t_0) \propto t_0^{\alpha \lambda-1}$. 
Note also
\else 
Note
\fi 
 that replacing $h^{\alpha}(t)$ on the left hand side by a scaled version $\rho \, h^{\alpha}(t)$ (where $\rho$ can depend on $t, t_0$ but not on the value of $h^{\alpha}(t)$ itself) will not change the bound, as it multiplies both the numerator and denominator by $\rho^2$. We therefore define a martingale $\bar{h}^{\alpha}(t)$ as follows:
\begin{definition}
    Given a $t_0$, consider the process $\bar{h}^{\alpha}(\cdot)$ starting from time $t_0$: $\bar{h}^{\alpha}(t_0) = h^{\alpha}(t_0)$; then for any $t > t_0$ we define the normalized convergence process as
    \begin{align}
        \bar{h}^{\alpha}(t) = \frac{R(t_0,\alpha\lambda)}{R(t,\alpha\lambda)} h^{\alpha}(t)\,.
    \end{align}
    We also define the random variable $\bar{h}_{\alpha}$ as follows:
  \begin{align}
        \bar{h}_{\alpha} =  \liminf_{t \to \infty} \bar{h}^{\alpha}(t)\,.
    \end{align}
\end{definition}

Then $\bar{h}^{\alpha}(t)$ is nonnegative, uniformly integrable, and is a martingale.
\begin{lemma} \label{lem::basic-bar-h-properties}
    The sequence  $\{\bar{h}^{\alpha}(t)\}_{t \geq t_0}$ is a uniformly integrable martingale and $\lim_{t \to \infty} \bar{h}^{\alpha}(t) = \bar{h}_{\alpha}$ almost surely. Furthermore, for any $t$, we have $\bar{h}^{\alpha}(t) = \bbE[\bar{h}_{\alpha} \,|\, \cH(t)]$.
\end{lemma}

\begin{proof}
    The process $\{\bar{h}^{\alpha}(t)\}_{t \geq t_0}$ is a martingale
    due to \Cref{lem::prop_h-t}(\ref{item::h-t-expectation}).
    The Martingale Convergence Theorem shows that it converges to a well-defined (random) limit almost surely, since $\bar{h}^{\alpha}(t) \geq 0$ and
    by definition 
    \begin{align}
    \bbE[|\bar{h}^{\alpha}(t)|] = \bbE[\bar{h}^{\alpha}(t)] = \bar{h}^{\alpha}(t_0) < \infty
    \end{align}
    Note that this means that almost surely,
    \begin{align}
        \bar{h}_\alpha = \lim_{t \to \infty} \bar{h}^\alpha(t)
    \end{align}
    as the limit almost surely exists.
    
    Finally, \Cref{lem::var-expsq-ratio} (and the fact that $\{\bar{h}^{\alpha}(t)\}_{t \geq t_0}$ is a martingale) shows that $\bbE[|\bar{h}^{\alpha}(t)|^2]$ is bounded for all $t$, which implies uniform integrability by \cite[Section 13.3]{williams_1991}.
\end{proof}

Then \Cref{lem::var-expsq-ratio} yields the following:
\begin{corollary}
    For any sufficiently large $t^*$, for any $t > t^*$
    \begin{align}
        \bbP\bigg[\bar{h}^{\alpha}(t) \leq \frac{\bar{h}^{\alpha}(t^*)}{2}  \,\Big|\, \cH(t^*)\bigg] \leq \frac{c^*}{c^* + 1/4} 
    \end{align}
    This also implies that for any sufficiently large $t^*$,
    \begin{align} \label{eq::limit-bdd-too}
        \bbP\bigg[\bar{h}_{\alpha} \leq \frac{\bar{h}^{\alpha}(t^*)}{4}  \,\Big|\, \cH(t^*)\bigg] \leq \frac{c^*}{c^* + 1/8} 
    \end{align}
    where $c^*$ is the constant used in \Cref{lem::var-expsq-ratio}.
\end{corollary}

\ifarxiv 

\begin{proof}
    Note that if $\bar{h}^{\alpha}(t) \leq \frac{\bar{h}^{\alpha}(t^*)}{2}$ then $\bar{h}^{\alpha}(t^*)-\bar{h}^{\alpha}(t) \geq \frac{\bar{h}^{\alpha}(t^*)}{2}$; since $\bar{h}^{\alpha}(t^*)-\bar{h}^{\alpha}(t)$ has mean $0$ (conditioned on $\cH(t^*)$) and variance $\leq c^*$ given by \Cref{lem::var-expsq-ratio}, the Chebyshev-Cantelli inequality states that 
    \begin{align}
        \bbP&\bigg[\bar{h}^{\alpha}(t^*)-\bar{h}^{\alpha}(t) \geq \frac{\bar{h}^{\alpha}(t^*)}{2}\bigg|\cH(t^*) \bigg] 
        \\ & \leq \frac{\var[\bar{h}^{\alpha}(t)]}{\var[\bar{h}^{\alpha}(t)] +(\frac{\bar{h}^{\alpha}(t^*)}{2})^2 }
        \\&\leq \frac{c^* \bar{h}^{\alpha}(t^*)^2}{c^* \bar{h}^{\alpha}(t^*)^2 + (\frac{\bar{h}^{\alpha}(t^*)}{2})^2} 
        = \frac{c^*}{c^* + 1/4} 
    \end{align}
    Note that the function $\frac{x}{x+y}$ is increasing in $x$ if $y$ is positive, so we appropriately get an upperbound when applying $\var[\bar{h}^{\alpha}(t)] \leq c^* \bar{h}^{\alpha}(t^*)^2$.
    
    To prove \eqref{eq::limit-bdd-too} (note that the bound is now $\frac{\bar{h}^{\alpha}(t^*)}{4}$ rather than $\frac{\bar{h}^{\alpha}(t^*)}{2}$) we have that by \Cref{lem::basic-bar-h-properties}, $\bar{h}^{\alpha}(t) \to \bar{h}_{\alpha}$ almost surely; this means that $\bar{h}^{\alpha}(t) \to \bar{h}_{\alpha}$ in probability as well, so for any $\delta_1, \delta_2 > 0$, and for sufficiently large $t$,
    \begin{align}\label{eq::hbar_converge_prob}
        \bbP[|\bar{h}_{\alpha}-\bar{h}^{\alpha}(t)| > \delta_1] \leq \delta_2
    \end{align}
    Set $\delta_1 = \frac{\bar{h}^{\alpha}(t^*)}{4}$ and $\delta_2 = \frac{c^*}{c^* + 1/8} - \frac{c^*}{c^* + 1/4} $. Then we assume to the contrary that 
    \begin{align}
        \bbP\bigg[\bar{h}_{\alpha} \leq \frac{\bar{h}^{\alpha}(t^*)}{4}  \,\Big|\, \cH(t^*)\bigg] > \frac{c^*}{c^* + 1/8}
    \end{align}
    This means that, for any sufficiently large $t$,
    \begin{align}
        \bbP&\bigg[\bar{h}_{\alpha} \leq \frac{\bar{h}^{\alpha}(t^*)}{4} \text{ and } \bar{h}^{\alpha}(t) > \frac{\bar{h}^{\alpha}(t^*)}{2}\bigg| \cH(t^*)\bigg] \\&> \frac{c^*}{c^* + 1/8} - \frac{c^*}{c^* + 1/4} 
    \end{align}
    which yields the desired contradiction given by \eqref{eq::hbar_converge_prob}.
\end{proof}

\else 

\begin{proof}[Proof Sketch]
\revisetwo{Since we can bound the variance of $\bar{h}^{\alpha}(t^*)-\bar{h}^{\alpha}(t)$ using \Cref{lem::var_as_exp}, we can apply Chebyshev-Cantelli inequality to show the first statement. To get  \eqref{eq::limit-bdd-too}, we first need to use the fact that \Cref{lem::basic-bar-h-properties} implies that $\bar{h}^{\alpha}(t) \to \bar{h}_{\alpha}$ in probability. Combining this with the first statement proves that $\bar{h}^{\alpha}(t^*)$ and $\bar{h}_{\alpha}$ must be close. } 
\end{proof}

\fi 

\begin{lemma}\label{lem::apply_levy}
For any $\alpha$, almost surely
$
\bar{h}_\alpha > 0
$
\end{lemma}

\begin{proof}
Therefore we have established that for any $\alpha$:
\begin{itemize}
    \item $\bar{h}^{\alpha}(t)$ is a uniformly integrable martingale;
    \item $\bar{h}^{\alpha}(t) \to \bar{h}_{\alpha}$ almost surely;
    \item for any sufficiently large $t^*$
       \eqref{eq::limit-bdd-too} holds
        which implies that for all sufficiently large $t^*$,
        \begin{align}
            \bbP[\bar{h}_{\alpha} > 0 \,|\, \cH(t^*)] \geq 1 - \frac{c^*}{c^* + 1/8} > 0
        \end{align}
\end{itemize}
We define the process $\eta(t)$ as
\begin{align}
    \eta(t) = \bbP[\bar{h}_{\alpha} > 0 \,|\, \cH(t)]
\end{align}
which is a martingale due to the tower property (and uniformly integrable because it is bounded). We likewise define
\begin{align}
    \eta = \bone\{\bar{h}_{\alpha} > 0\}
\end{align}
and (due to uniform integrability) almost surely $\eta(t) \to \eta$ because of Levy's 0-1 Law \cite[Theorem 5.5.8]{Durrett2010}%
\iftrue 
\revisetwo{; specifically, since $\eta$ is an indicator function and is fully determined by the filtration $\bigcup_t \cH(t)$ we know by Levy's Upward Theorem (applied to $\eta(t)$) that
\begin{align}
    \lim_{t \to \infty} \eta(t) &= \lim_{t \to \infty}\bbP[\bar{h}_\alpha > 0 \,|\, \cH(t)]
    = \lim_{t \to \infty}\bbE[\eta \,|\, \cH(t)]
    \\ &\eqt{a.s.} \bbE\bigg[\eta \,\Big|\, \bigcup_t \cH(t)\bigg]
    = \eta
    \in \{0,1\}\,.
\end{align}
}
\fi 
 But $\lim_{t \to \infty} \eta(t) \geq 1 - \frac{c^*}{c^* + 1/8} > 0$, thus showing that $1$ is the only possible limit out of $\{0,1\}$. Thus, $\eta = 1$ almost surely, so $\bar{h}_\alpha > 0$ almost surely.
\end{proof}

\subsection{Rate of Decay for $V$ and $\beta_i$}

\begin{theorem}\label{thm::rate_VBt}

For any $\epsilon > 0$, we have that
\begin{align}
\lim_{t\to \infty} \frac{V(\bbeta(t))}{t^{\lambda - 1 + \epsilon}} &= 0
\text{ and }\lim_{t\to \infty} \frac{V(\bbeta(t))}{t^{\lambda - 1 - \epsilon}} = \infty
\end{align}
almost surely.
\end{theorem}

\begin{proof}
For a given $\epsilon$, using \Cref{lem::prop_h-t}(\ref{item::vb_bounded_h}), we can find a $\delta(\epsilon)$ with corresponding $t_0$ large enough and an $\alpha_{+}$ and $\alpha_{-}$ which are such that
\begin{align}
\alpha_{+} &< 1 + \epsilon / \lambda
~~\text{and}~~ \alpha_{-} > 1 - \epsilon / \lambda
\end{align}

Then for any trajectory $V(\bbeta(t))$, there exists some trajectory $h^{\alpha_{-}}(t)$ and $h^{\alpha_{+}}(t)$ such that 
\begin{align}
h^{\alpha_{-}}(t)\leq V(\bbeta(t)) \leq h^{\alpha_{+}}(t)
\end{align}

Using \Cref{lem::apply_levy}, any trajectory of $h^{\alpha_{-}}(t)$ has a corresponding martingale $\bar{h}^{\alpha_{-}}(t)$ which converges to a constant. Thus  $h^{\alpha_{-}}(t)$ converges to zero at a rate of $\Omega(t^{\alpha_{-} \lambda - 1})$ almost surely. Similarly, $h^{\alpha_{+}}(t)$ converges to zero at a rate of $\Omega(t^{\alpha_{+} \lambda - 1})$ almost surely.

Since the value of $t_0$ does not affect the asymptotic rate after $t_0$, we have
\begin{align}
\lim_{t \to \infty} \frac{V(\bbeta(t))}{t^{\lambda - 1 + \epsilon}} &\leq \lim_{t \to \infty} \frac{h^{\alpha_{+}}(t)}{t^{\lambda - 1 + \epsilon}}
\leq  \lim_{t \to \infty} \frac{c_{+}t^{\alpha_{+} \lambda - 1} }{t^{\lambda - 1 + \epsilon}} 
= 0 
\\
\lim_{t \to \infty} \frac{V(\bbeta(t))}{t^{\lambda - 1 - \epsilon}} &\geq \lim_{t \to \infty} \frac{h^{\alpha_{-}}(t)}{t^{\lambda - 1 - \epsilon}}
\geq  \lim_{t \to \infty} \frac{c_{-}t^{\alpha_{-} \lambda - 1} }{t^{\lambda - 1 - \epsilon}} 
 = \infty \,.
\end{align}

\end{proof}
Next we need the above result on $V(\bbeta(t))$ to imply a result on all $\beta_i(t)$. 

\begin{lemma}\label{lem::exist_agent_rate_tlam}
Given any $\epsilon > 0$, there exists some time $t_0$ and some $c$, such that
for all $t>t_0$, there exists some $i$ such that
\begin{align}
\beta_i(t) \geq c t^{\lambda - 1 - \epsilon} \label{eq::beta_i_larger_rate_ep}
\end{align}
almost surely. 
Additionally, for any $t_2 > t_0$, there exists some $i$ such that $\beta_i(t) \geq c t^{\lambda - 1 - \epsilon}$ holds at least $1/n$ of the times $t$ in $t_2 < t \leq 2 t_2$.
\end{lemma}

\ifarxiv  
Note that at this point in the development of our results, the lemma does not guarantee that the same $i$ will satisfy the property \eqref{eq::beta_i_larger_rate_ep}. Hence why we have the second statement saying that there exists some $i$ that satisfies \eqref{eq::beta_i_larger_rate_ep} some fraction of the time. 
\else 
\fi 

\begin{proof}
First, $V(\bbeta(t)$ converges at a rate at least $\Omega(t^{\lambda - 1 - \epsilon})$.
Thus, there is some $t_0$, where for $t >t_0$,
\begin{align}
V(\bbeta(t)) \geq c t^{\lambda - 1 - \epsilon}\,.
\end{align}

At time $t > t_0$, let $k$ be such that
$
\beta_k(t) \geq \beta_i(t)
$
for all $i$.
Since $\bv$ is a vector so that $\bv^{\ltop} \bone = 1$, we have that
\begin{align}
\beta_k(t) &= (\bv^{\ltop} \bone) \beta_k(t) \geq \bv^{\ltop} \bbeta(t)
 = V(\bbeta(t))\,.
\end{align}
And thus this shows \eqref{eq::beta_i_larger_rate_ep} for each $t$.

For the second statement, we know that there must be some $i$ which satisfies \eqref{eq::beta_i_larger_rate_ep}. Since there are only $n$ candidates for $i$, at least one $i$ must occur the most often, which means this $i$ occurs at least $1/n$ of the time in a certain interval. 
\end{proof}

\revisetwo{We eventually want to show that $\beta_i(t) \geq c t^{\lambda-1-\epsilon}$ for some $c$ for all $t$ larger than some $t^*$ for all $i$. To help show this, we will first look at a weaker case, when this equation holds for some portion of the time. This will help prove what we need.}

\begin{definition}
     For any agent $i$ and constants $c, \rho > 0$, we say that time $t_1 > 0$ is \emph{$(c,\rho)$-good} for agent $i$ if
     \begin{align} \label{eq::mu_i_greater_c_t_rate}
         \beta_i(t) \geq c t^{\lambda-1-\epsilon}
     \end{align}
     for at least $\rho$ fraction of the times $t \in [t_1/2, t_1]$. We denote the set of times where this holds as
     \begin{align}
         \cT_i(c,\rho) = \{t_1 : t_1 \text{ is } (c, \rho)\text{-good for } i\}
     \end{align}
\end{definition}

For shorthand, once $c, \rho$ are fixed, we denote the set of good times for $i$ as $\cT_i$. The key observation is that the set of good times for an agent $i$ eventually becomes good for all her neighbors $j$, which then eventually become good for all their neighbors, and so forth until the set of good times for $i$ must be good for all agents.

\begin{lemma} \label{lem::mu_to_one_beta}
    Fix an agent $i$ and constants $c_i, \rho_i > 0$, and let $\cT_i := \cT_i(c_i, \rho_i)$, and let $j$ be adjacent to $i$. Then there is some $c_j, \rho_j > 0$ such that, almost surely, there is some $t^*$ for which
    \begin{align}
        t_1 > t^* \text{ and } t_1 \in \cT_i \implies t_1 \in \cT_j
    \end{align}
    where $\cT_j := \cT_j(c_j, \rho_j)$.
\end{lemma}


\ifarxiv 
Note that $\rho_j$ can be less than $\rho_i$, meaning that in the argument where good times spread from a source $i$, the $\rho$'s diminish as the process gets further from $i$; however, since the graph is finite, it remains bounded away from $0$ over the whole graph.
\else 
Note that $\rho_j$ can be less than $\rho_i$, meaning that in the argument where good times spread from a source $i$, the $\rho$'s diminish as the process gets further from $i$.
\fi 

\ifarxiv 
\iftrue 
\begin{proof}
    Recall that $t_1 \in \cT_i$ means that $\beta_i \geq c_i t^{\lambda-1-\epsilon}$ for at least a $\rho_i$ fraction of $t \in [t_1/2, t_1]$; we say a time $t$ is \emph{$c_i$-enough} (for agent $i$) if $\beta_i \geq c_i t^{\lambda-1-\epsilon}$ (unlike \emph{good} times $t_1$, this only depends on the value of $\beta_i$ at time $t$, not at any previous time).

    First, note that at least $(\rho_i/4)t_1$ different $t$ in $[t_1/2, (1 - \rho_i/4)t_1]$ are $c_i$-enough (since there are at least $(\rho_i/2)t_1$ $c_i$-enough times in total). For any $t \in [t_1/2, t_1]$,
    \begin{align}
        \beta_i(t) \geq \begin{cases} c_i t_1^{\lambda-1-\epsilon} &\text{if } t \text{ is } c_i\text{-enough} \\ 0 &\text{otherwise} \end{cases}
    \end{align}
    since $t \leq t_1$, and therefore (as $w_{j,i} = \frac{a_{j,i}}{\degree(j)}$)
    \begin{align}
        \mu_j(t) \geq \begin{cases} c_i w_{j,i} t_1^{\lambda-1-\epsilon} &\text{if } t \text{ is } c_i\text{-enough} \\ 0 &\text{otherwise} \end{cases}.
    \end{align}
    So, when $t$ is $c_i$-enough, we get
    \begin{align}
        f(\mu_j(t), \gamma_j) &\geq f(w_{j,i}c_i t_1^{\lambda-1-\epsilon}, \gamma_j)
        \\&= \frac{\gamma_j w_{j,i}c_i t_1^{\lambda-1-\epsilon}}{\gamma_j w_{j,i}c_i t_1^{\lambda-1-\epsilon} + 1 - w_{j,i}c_i t_1^{\lambda-1-\epsilon} }
        \\&\geq c'_j t_1^{\lambda-1-\epsilon}
    \end{align}
    where $c'_j = \min(\gamma_j, 1) w_{j,i} c_i$, which is a lower bound on the probability of agent $j$ declaring $1$ at any $c_i$-enough time. Consider the $\geq \rho_i/4$ such times in $[t_1/2, (1-\rho_i/4)t_1]$; the number of $1$'s declared by $j$ in the range $[t_1/2, (1-\rho_i/4)t_1]$ thus stochastically dominates the sum of $(\rho_i/4)t_1$ independent Bernoulli random variables with probability $c'_j t_1^{\lambda-1-\epsilon}$ each (whose sum has expected value $(\rho_i/4) c'_j t_1^{\lambda-\epsilon}$). By the Chernoff bound, this then yields that
    \begin{align} \label{eq::chernoff-on-psi-j}
        \bbP\Bigg[\sum_{t=t_1/2}^{(1-\rho_i/4)t_1} \psi_{j,t} \leq \frac{\rho_i c'_j}{8} t_1^{\lambda-\epsilon} \Bigg] \leq e^{-\frac{\rho_i}{32}c'_j t_1^{\lambda-\epsilon}}
    \end{align}
    meaning that there is a very high probability of getting at least $\frac{\rho_i c'_j}{8} t_1^{\lambda-\epsilon}$ declarations of $1$ from agent $j$ by time $(1 - \rho_i/4)t_1$. But then for the $(\rho_i/4)t_1$ times $t \in [(1-\rho_i/4)t_1, t_1]$, we have
    \begin{align}
        \beta_j(t) \geq \frac{\rho_i c'_j}{8} t_1^{\lambda-1-\epsilon}
    \end{align}
    so $t_1 \in \cT_j(c_j'', \rho_j)$ where $c_j'' = \frac{\rho_i c'_j}{8} $ and $\rho_j = \rho_i/2$ (since $(\rho_i/4)t_1$ needs to be a $\rho_j$ proportion of $t_1/2$).

    Finally, we need to show that this probabilistic bound then implies that almost surely there are only finitely many $t_1$ which are in $\cT_i(c_i, \rho_i)$ but not $\cT_j(c_j'', \rho_j)$. This case only happens when $\sum_{t=t_1/2}^{(1-\rho_i/4)t_1} \psi_{j,t} \leq \frac{\rho_i c'_j}{8} t_1^{\lambda-\epsilon}$, and the probability of this (by \eqref{eq::chernoff-on-psi-j}) decreases faster than any inverse polynomial of $t_1$, and hence has a finite sum over all $t_1 \in \cT_i$ (notably, $\sum_{t_1=1}^\infty e^{-\frac{\rho_i}{32}c'_j t_1^{\lambda-\epsilon}} < \infty$, so summing only over $t_1 \in \cT_i$ must also be finite). Thus, by the Borel-Cantelli Lemma, almost surely it happens only finitely many times, and we are done.
\end{proof}
\else 
\begin{proof}
    Given that \eqref{eq::mu_i_greater_c_t_rate} (with $c = c_i$) holds at least $\rho_i$ fraction of the time, we want to determine a lower bound on the number of $1$'s that agent $i$ declares in the time period from $t_1/2$ to $(1 - \rho_i/4)t_1$ (rounded up). Note that this leaves only $(\rho_i/4)t_1 = (\rho_i/2)(t_1/2)$ time steps (rounded down) afterwards, whereas we have \eqref{eq::mu_i_greater_c_t_rate} for at least $\rho_i t_1/2$ time steps within $[t_1/2, t_1]$. Thus, we have at least $(\rho_i/4)t_1$ time steps in $[t_1/2, (1-\rho_i/4)t_1]$ for which \eqref{eq::mu_i_greater_c_t_rate} holds.

    Without loss of generality, assume that $0$'s are declared whenever \eqref{eq::mu_i_greater_c_t_rate} does not hold, which is trivially a lower bound on the number of $1$'s that are actually declared during those steps. This means we need only look at the $\geq (\rho_i/4) t_1$ time steps for which $\beta_i(t) \geq c_i t^{\lambda-1-\epsilon}$. 

    Let $t$ be a time such that $\beta_i(t) \geq c_i t^{\lambda-1-\epsilon}$. Recall that probability that agent $j$ declares $1$ at time $t$ is given by
    \begin{align}
        f(\mu_j, \gamma_j) &= \frac{\gamma_j \mu_j(t)}{\gamma_j \mu_j(t) + 1 - \mu_j(t)}
    \end{align}
    where, since $j$ is a neighbor of $i$,
    \begin{align}
        \mu_j(t) &= \sum_{j' = 1}^n w_{j',j} \beta_{j'}(t) 
        \geq w_{i,j} \beta_i(t)
        \geq w_{i,j} c_i t^{\lambda-1-\epsilon}
    \end{align}
    where $w_{i,j} = a_{i,j}/\degree(j) \in (0, 1]$
    \ifsixteen
    .
    \else
    is the positive constant denoting the normalized edge weight between $i$ and $j$. 
    \fi
    Note that since we are assuming $\beta_i(t) \geq c_i t^{\lambda-1-\epsilon}$, and $\beta_i(t) \in [0,1]$, this implies that $t$ is sufficiently large so that $c_i t^{\lambda-1-\epsilon} \leq 1$.
    
    Since $f$ is monotonic in $\mu_j$, we have that
\ifsixteen 
\begin{align}
        f(\mu_j(t), \gamma_j) &\geq f(w_{i,j}c_i t^{\lambda-1-\epsilon}, \gamma_j)
        \\&= \frac{\gamma_j w_{i,j}c_i t^{\lambda-1-\epsilon}}{\gamma_j w_{i,j}c_i t^{\lambda-1-\epsilon} + 1 - w_{i,j}c_i t^{\lambda-1-\epsilon} }
        \\&\geq c'_j t_1^{\lambda-1-\epsilon}
    \end{align}
\else 
    \begin{align}
        f(\mu_j(t), \gamma_j) &\geq f(w_{i,j}c_i t^{\lambda-1-\epsilon}, \gamma_j)
        \\&= \frac{\gamma_j w_{i,j}c_i t^{\lambda-1-\epsilon}}{\gamma_j w_{i,j}c_i t^{\lambda-1-\epsilon} + 1 - w_{i,j}c_i t^{\lambda-1-\epsilon} }
        \\&\geq \min(1, \gamma_j) w_{i,j}c_i t^{\lambda-1-\epsilon} \label{eq::bounding-probability-for-each-good-step}
    \end{align}
where the last inequality follows because $w_{i,j}c_i t^{\lambda-1-\epsilon} \leq 1$, the denominator is at most $\max(1, \gamma_j)$, and because $t \in [t_1/2, t_1]$ which implies that $t_1^{\lambda-1-\epsilon} \leq  t^{\lambda-1-\epsilon}$.

For convenience, we denote the constant $c'_j$ as
\begin{align}
    c'_j := \min(1, \gamma_j) w_{i,j}c_i 
\end{align}
so that \eqref{eq::bounding-probability-for-each-good-step} is written in simplified form as
\begin{align} \label{eq::simplified-probability-for-each-good-step}
    f(\mu_j(t), \gamma_j) &\geq c'_j t_1^{\lambda-1-\epsilon}
\end{align}
for all $t$ such that $\beta_i(t) \geq c_i t^{\lambda-1-\epsilon}$. 
\fi 
As mentioned, this happens at least $(\rho_i/4)t_1$ times between $t_1/2$ and $(1-\rho_i/4)t_1$. Thus, the number of $1$'s declared by agent $j$ in the interval $[t_1/2, (1-\rho_i/4)t_1]$ stochastically dominates the sum of $(\rho_i/4)t_1$ independent Bernoulli random variables with probability $c'_j t_1^{\lambda-1-\epsilon}$.

We then let $\xi_1, \dots, \xi_{(\rho_i/4)t_1}$ (rounded up) be independent Bernoulli random variables, each with a probability of returning $1$ of $c'_j t_1^{\lambda-1-\epsilon}$. Since these are i.i.d. Bernoulli random variables, we can apply a concentration result on their sum. The expected value of their sum (ignoring rounding of $(\rho_i/4)t_1$, which is negligible) is
\begin{align}
    \bbE\Bigg[ \sum_{k=1}^{(\rho_i/4)t_1} \xi_k \Bigg] = (\rho_i/4)t_1 c'_j t_1^{\lambda-1-\epsilon} = (\rho_i/4) c'_j t_1^{\lambda-\epsilon}
\end{align}
Thus, using a Chernoff bound,
\begin{align}
    \bbP\Bigg[\sum_{k=1}^{(\rho_i/4)t} \xi_k \leq \frac{1}{2}(\rho_i/4) c'_j t_1^{\lambda-\epsilon} \Bigg] \leq e^{-\frac{\rho_i}{32}c'_j t_1^{\lambda-\epsilon}}
\end{align}
For any $t_1 \in \cT_i$, this means that
\begin{align}
    \bbP\Bigg[\sum_{t=t_1/2}^{(1-\rho_i/4)t_1} \psi_{j,t} \leq \frac{\rho_i c'_j}{8} t_1^{\lambda-\epsilon} \Bigg] \leq e^{-\frac{\rho_i}{32}c'_j t_1^{\lambda-\epsilon}}
\end{align}
We now claim that if 
\begin{align}
    \sum_{t=t_1/2}^{(1-\rho_i/4)t_1} \psi_{j,t} \geq \frac{\rho_i c'_j}{8} t_1^{\lambda-\epsilon}
\label{eq::many_psi_happen_in_interval}
\end{align}
(i.e. the sum of the $\psi_{j,t}$ are not significantly lower than expected) then $t_1 \in \cT_j$. This is because for any $t \in [(1 - \rho_i/4)t_1, t_1]$ (the part at the end held out) we have
\begin{align}
    \beta_j(t) = \frac{\sum_{\tau=1}^t \psi_j(\tau)}{t} \geq \frac{\frac{\rho_i c'_j}{8} t_1^{\lambda-\epsilon}}{t_1} = \frac{\rho_i c'_j}{8} t_1^{\lambda- 1 - \epsilon}
\end{align}
Note that this therefore applies for $(\rho_i/4)t_1$ times in the range $[t_1/2, t_1]$, which is a $\rho_i/2$ fraction of that interval. Thus, setting 
\begin{align}
    c_j = \frac{\rho_i c'_j}{8} \text{ and } \rho_j = \rho_i/2
\end{align}
we can conclude that $t_1 \in \cT_j$.

\ifsixteen
It remains to be shown that (almost surely) after a certain time $t^*$ \eqref{eq::many_psi_happen_in_interval}
\else
One last step is needed: we showed a high probability that $\sum_{t=t_1/2}^{(1-\rho_i/4)t_1} \psi_{j,t} \geq c_j t_1^{\lambda-\epsilon}$ and that if this happened then $t_1 \in \cT_j$; however, it 
remains to be shown that (almost surely) after a certain time $t^*$ this 
\fi
always happens when $t_1 \in \cT_i$ for $t_1 > t^*$. This follows since $e^{-\frac{\rho_i}{32}c'_j t_1^{\lambda-\epsilon}}$ decays to $0$ faster than any inverse polynomial, and hence
\begin{align}
    \sum_{t_1 \in \cT_i} e^{-\frac{\rho_i}{32}c'_j t_1^{\lambda-\epsilon}} \leq \sum_{t_1 = 1}^\infty e^{-\frac{\rho_i}{32}c'_j t_1^{\lambda-\epsilon}} < \infty \,.
\end{align}
Thus by the Borel-Cantelli Lemma, almost surely
\begin{align}
    \bbP\Bigg[\sum_{t=t_1/2}^{(1-\rho_i/4)t_1} \psi_{j,t} \leq \frac{\rho_i c'_j}{8} t_1^{\lambda-\epsilon} \Bigg] \leq e^{-\frac{\rho_i}{32}c'_j t_1^{\lambda-\epsilon}}
\end{align}
happens only finitely many times, so there is some finite $t^*$ for which the above holds, thus finishing the proof.
\end{proof}
\fi 

\else 
\begin{proof}[Proof Sketch]
\revise{
First, we determine a lower bound on the number of times that \eqref{eq::mu_i_greater_c_t_rate} holds. Then for neighbor $j$, the evolution of the dynamics gives that $\mu_j(t)$ at those times must be greater than some constant times $t^{\lambda -1 - \epsilon}$ at some set number of time steps. We then use a Chernoff bound on the sums of independent variables generated at these time steps by agent $j$, which shows that with high probability $\beta_j(t)$ is greater than some constant multiplied by $t^{\lambda -1 - \epsilon}$. Then using Borel-Cantelli, we convert the high probability result into a result that holds for all but finitely many time steps. 
}
\end{proof}
\fi

\begin{proposition} \label{prop::one_beta_leads_all_beta}
    For any agent $i$ and constants $c_i, \rho_i > 0$, let $\cT_i := \cT_i(c_i, \rho_i)$. Then there is a set of constants $c_j, \rho_j > 0$ for all $j \neq i$ such that, almost surely, there is some $t^* > 0$ such that
    \begin{align}
        t_1 > t^* \text{ and } t_1 \in \cT_i \implies t_1 \in \cT_j \text{ for all } j
    \end{align}
    where $\cT_j := \cT_j(c_j, \rho_j)$.
\end{proposition}

\ifarxiv 
\begin{proof}
\else
\begin{proof}[Proof Sketch]
\fi
    
    This follows from \Cref{lem::mu_to_one_beta} by induction over \revisetwo{agents sorted by distance} to agent $i$. 
    \ifarxiv 
    Let $\dist(i,j)$ denote the distance of vertex $j$ from $i$, and $N_i(k) := \{j : \dist(i,j) \leq k\}$. We show that if the condition holds for all $j \in N_i(k)$, it holds for all $j \in N_i(k+1)$ as well, and therefore it holds for all $j \in N_i(n)$ (i.e. for all $j$, since no vertex can be more than $n$ distance from $i$).

    Base case: $k = 0$. Then $N_i(0) = \{i\}$, and the condition is trivially true.
    Inductive step: we know that there is a $t^*_{(k)}$ such that for all $t_1 > t^*_{(k)}$ and all $j \in N_i(k)$,
    \begin{align}
        t_1 \in \cT_i \implies t_1 \in \cT_j
    \end{align}
    (where $\cT_j := \cT_j(c_j, \rho_j)$ for $c_j, \rho_j > 0$). Now consider some $j' \in N_i(k+1)$; by definition it has some neighbor $j \in N_i(k)$. Thus by \Cref{lem::mu_to_one_beta}, there is some $t^*(j')$ and constants $c_{j'}, \rho_{j'} > 0$ such that for all $t_1 > t^*(j')$,
    \begin{align}
        t_1 \in \cT_j \implies t_1 \in \cT_{j'}
    \end{align}
    But we know that for all $t_1 > t^*_{(k)}$,
    $
        t_1 \in \cT_i \implies t_1 \in \cT_j
    $
    and therefore for all $t_1 > \max(t^*_{(k)}, t^*(j'))$ we have
    \begin{align}
        t_1 \in \cT_i \implies t_1 \in \cT_{j'} \,.
    \end{align}
    We can then choose
    $
        t^*_{(k+1)} = \max(\max_{j'}(t^*(j')), t^*_{(k)})
    $
    to complete the induction step.

    Finally, we take the $c_j,\rho_j > 0$ generated for all $j$, and, letting $k_{\max} := \max_j \dist(i,j)$, set $t^* = t^*_{(k_{\max})}$ to complete the result.
     \else 
    \fi 
\end{proof}

It is important that the graph is finite, 
\ifarxiv
since having only a finite number of induction steps means that the final value of $t^* = t^*_{(k_{\max})}$ will be finite.
\else
since then we only have a finite number of induction steps.
\fi

We also need that $t_1$ being a good time for agent $i$ means that $\beta_i(t_1)$ also obeys a constant factor lower bound of order $t_1^{\lambda-1-\epsilon}$ (for all sufficiently large $t_1$):

\begin{lemma} \label{lem::good-time-large-beta}
    For any $t_1 \geq 2/\rho_i$, if $t_1 \in \cT_i(c_i, \rho_i)$, then
    \begin{align}
        \beta_i(t_1) \geq (c_i/2) t_1^{\lambda-1-\epsilon}\,.
    \end{align}
\end{lemma}

\begin{proof}
    This follows since $t_1 \in \cT_i(c_i, \rho_i)$ and $t_1 \geq 2/\rho_i$ means that there is at least one $t \in [t_1/2, t_1]$ such that
    \ifarxiv
    \begin{align}
        \beta_i(t) \geq c_i t^{\lambda-1-\epsilon} \geq c_i t_1^{\lambda-1-\epsilon}\,.
    \end{align}
    \else
    $\beta_i(t) \geq c_i t^{\lambda-1-\epsilon} \geq c_i t_1^{\lambda-1-\epsilon}$.
    \fi 
    But this means that
\begin{align}
        \beta_i(t_1) &= \frac{\sum_{\tau=1}^{t_1} \psi_{i,\tau}}{t_1}
        \geq \frac{1}{2} \frac{\sum_{\tau=1}^{t_1} \psi_{i,\tau}}{t_1/2}
        \\ &\geq \frac{1}{2} \frac{\sum_{\tau=1}^{t} \psi_{i,\tau}}{t}
        = \frac{1}{2}\beta_i(t)
        \geq (c_i/2) t_1^{\lambda-1-\epsilon}\,.
\end{align}
\end{proof}

This yields the result that, almost surely,
$
    \beta_i(t) = \tilde{O}(t^{\lambda-1})
$:

\begin{proposition}
For any $i \in [n]$ and any $\epsilon > 0$, almost surely
\ifarxiv
\begin{align}
\label{eq::ind_agent_rate_zero} \lim_{t\to \infty} \frac{\beta_i(t)}{t^{\lambda - 1 + \epsilon}} &= 0
\\ \label{eq::ind_agent_rate_infinity}
\lim_{t\to \infty} \frac{\beta_i(t)}{t^{\lambda - 1 - \epsilon}} &= \infty
\end{align}
\else
\begin{align}
\label{eq::ind_agent_rate_zero} \lim_{t\to \infty} \frac{\beta_i(t)}{t^{\lambda - 1 + \epsilon}} &= 0
 \label{eq::ind_agent_rate_infinity} \text{ and }
\lim_{t\to \infty} \frac{\beta_i(t)}{t^{\lambda - 1 - \epsilon}} = \infty
\end{align}
\fi
\end{proposition}

\iffalse 
\begin{proposition}
For any $i$ and $\epsilon > 0$, almost surely as $t \to \infty$,
\begin{align}
    t^{\lambda-1-\epsilon} \leq \beta_i(t) \leq t^{\lambda-1+\epsilon}
\end{align}
\end{proposition}
\fi  

\begin{proof}
\ifarxiv
First, we can show \eqref{eq::ind_agent_rate_zero} as a corollary of \Cref{thm::rate_VBt}.
\else
The first equation follows from \Cref{thm::rate_VBt}.
\fi 
Since $V(\bbeta(t)) = \bv^{\ltop} \bbeta(t) = \sum_{i = 1}^n v_i \beta_i(t)$ and each $v_i > 0$ is constant, for each $i$ there is a $c_i$ so that
$
\beta_i(t)  \leq c_i V(\bbeta(t))
$,
so
\begin{align}
\lim_{t\to \infty} \frac{\beta_i(t)}{t^{\lambda - 1 + \epsilon}} &\leq \lim_{t \to \infty} \frac{c_i V(\bbeta(t))}{t^{\lambda - 1 + \epsilon}} = 0\,.
\end{align}

\ifarxiv  
Next, \Cref{lem::exist_agent_rate_tlam} \revisetwo{(with $\epsilon / 2$)} shows that there is some $c > 0$ and time $t_0$ such that for all $t > t_0$, 
\begin{align}
    t \in \cT_i(c, 1/n) \text{ for some } i
\end{align}
\Cref{prop::one_beta_leads_all_beta} then shows that for each $i$, there is some collection of $c^{(i)}_j, \rho^{(i)}_j > 0$ such that almost surely there is a time $t^{(i)}$ such that when $t > t^{(i)}$,
\begin{align}
    t \in \cT_i(c,1/n) \implies t \in \cT_j(c^{(i)}_j, \rho^{(i)}_j)
\end{align}
which, by \Cref{lem::good-time-large-beta}, implies that (when $t$ is sufficiently large) for any $j$
\begin{align}
    t \in \cT_i(c,1/n) \implies \beta_j(t) \geq (c_j^{(i)}/2) t^{\lambda-1-\epsilon/2}
\end{align}
Finally, let $c_j = \min_i c^{(i)}_j > 0$. Then, for sufficiently large $t$,
\begin{align}
    \exists i \text{ s.t. } t \in \cT_i(c,1/n) \implies \beta_j(t) \geq (c_j/2) t^{\lambda-1-\epsilon/2}
\end{align}
But by \Cref{lem::exist_agent_rate_tlam} we know that almost surely there is some $t_0$ such that for any $t > 0$, there is some $i$ such that $t \in \cT_i(c,1/n)$, which then implies that there are $c_1, \dots, c_n$ such that for any sufficiently large $t$,
$
    \beta_j(t) \geq (c_j/2) t^{\lambda-1-\epsilon/2}
$
and we are done.
\else 
Next, \Cref{lem::exist_agent_rate_tlam} \revisetwo{(with $\epsilon / 2$)} shows that there is some $c > 0$ and time $t_0$ such that for all $t > t_0$, there is some $i$ such that $t \in \cT_i(c, 1/n)$. By \Cref{prop::one_beta_leads_all_beta}, there is (almost surely) some time $t^*$ and constants $c'_j, \rho_j > 0$ such that for any $i$, any $t > t^*$ and $j \in [n]$, we have $t \in \cT_i(c, 1/n) \implies t \in \cT_j(c'_j, \rho_j)$. Thus, for all $t > \max(t_0, t^*)$, we have $t \in \cT_j(c'_j, \rho_j)$ for all $j$. Finally, we apply \Cref{lem::good-time-large-beta} to show there are constants $c_j > 0$ such that \revisetwo{$\beta_j(t) \geq c_j t^{\lambda-1-\epsilon/2}$} for all $t > \max(t^*, t_0)$.
\fi 
\end{proof}

\subsection{Computing Consensus Convergence Rate}

Finally, we can use the consensus rate to bound the 
convergence rate of the inherent belief estimator.
If for each $i$, if $\beta_i(t) \geq c t^{\lambda - 1 - \epsilon}$, then $\mu_i(t) \geq c t^{\lambda - 1 - \epsilon}$.
\ifarxiv
Using \eqref{eq::compute_xt_for_estimator},
we have that
\begin{align}
x(t) &\geq (\gamma_i - 1) \mu_i(t)
\geq (\gamma_i - 1)  c t^{\lambda - 1 - \epsilon}\,.
\end{align}
\else 
We compute that $x(t) \geq (\gamma_i - 1) \mu_i(t) \geq (\gamma_i - 1)  c t^{\lambda - 1 - \epsilon}$.
\fi
Then 
\ifarxiv 
\begin{align}
X(t) &\geq (\gamma_i - 1) c \sum_{\tau = t_0}^t \tau^{\lambda - 1 -\epsilon} 
\\ &\approx (\gamma_i - 1) c \int_{t_0}^{t} \tau ^{\lambda - 1 - \epsilon} d\tau
\\ &= (\gamma_i-1) c   \frac{1}{\lambda - \epsilon}\left(t^{\lambda - \epsilon} - t_0^{\lambda - \epsilon} \right)
\approx c' t^{\lambda - \epsilon}\,.
\end{align}
\else 
$X(t) \geq (\gamma_i - 1) c \sum_{\tau = t_0}^t \tau^{\lambda - 1 -\epsilon} 
\approx c' t^{\lambda - \epsilon}$.
\fi 

\Cref{prop::estimator_convergence_function} yields that
$
\bbP\left[\exists t \geq t^* : \widehat{\phi}_i(t) \neq \phi_i\right] \leq \delta
$
if 
\begin{align}
t^* \geq \left(\frac{1}{c'} \frac{1}{\xi_i} \log \frac{1}{\delta (e^{\xi_i} - 1)} \right)^{1/(\lambda - \epsilon)}\label{eq::rate_estimator_consensus}
\end{align}
(assuming that $t^*$ is such that $c' (t^*)^{\lambda - \epsilon} >2 $). 

Compared to \eqref{eq::rate_worst_case_mu}, we see that instead of a rate which is $1/\delta$ to some power, we get $\tilde{\Theta}(\log(1/\delta)^{1/\lambda})$ (since \eqref{eq::rate_estimator_consensus} holds for all $\epsilon > 0$), which is a big improvement. This, along with \eqref{eq::order_rate_for_interior}, yields the following theorem, suggesting that estimating inherent beliefs in consensus is more difficult:

\begin{theorem}
    For the inherent belief estimator $\widehat{\phi}_i(t)$ given in \Cref{def::estimator_inherent}, let $t^*$ be defined as the time of convergence:
    \begin{align}
        t^* := \max(t : \widehat{\phi}_i(t) \neq \phi_i)
    \end{align}
    i.e. the first time such that $\widehat{\phi}_i(t) = \phi_i$ for all $t > t^*$. Then:
    \begin{align}
        t^* = 
        \begin{cases}
            O(\log(1/\delta)) &\text{if no consensus}
            \\ O(\log(1/\delta)^{\frac{1}{\lambda} + \epsilon}) \text{ for any } \epsilon > 0 &\text{if consensus}
        \end{cases}
    \end{align}
    with probability $\geq 1-\delta$, where $\lambda$ is the largest eigenvalue of $\bGamma \bW$ if the consensus is to $\bzero$, and of $\bGamma^{-1} \bW$ if to $\bone$.
\end{theorem}

\begin{proof}
    This follows directly from \eqref{eq::order_rate_for_interior} (for the non-consensus case) and from \eqref{eq::rate_estimator_consensus} (for the consensus case).   
    \ifarxiv
    Note that if something holds for an exponent of $\frac{1}{\lambda - \epsilon}$ for all $\epsilon > 0$, this is equivalent to holding for an exponent of $\frac{1}{\lambda} + \epsilon$ for all $\epsilon > 0$, so we can make the substitution.
    \fi
\end{proof}


\fi



\ifacc
\else
\section{Conclusion}

In this work, we study the Interacting P\'{o}lya Urn model of opinion dynamics model under social pressure. 
We expanded upon \cite{socialPressure2021} by showing there exists an estimator for bias parameters and inherent beliefs. 
Specifically, we showed that the history of any agent and their neighbors' declarations is sufficient in the limit to determine an agent's inherent belief and bias parameter for any network structure, using estimators based on maximum likelihoods. We also analyzed the rate at which the inherent belief estimator converges.

{ \subsection{Limitations and Future Work}

An important open question is how accurately the Interacting P\'{o}lya Urn models real social pressure in various situations, in particular how well the proposed interaction between inherent belief (or bias) with social pressure corresponds to reality.

Also, although the model can accommodate general weighted graphs and biases, there is still a significant simplification due to the assumption of a fixed social network (where agents do not change how much they are influenced by others) and synchronous updates. The estimators presented also assume full knowledge of the network and of agents' expressed opinions, which may not be the case when studying real social networks.
Furthermore, while this work shows rigorous upper bounds on the (probabilistic) convergence time of the bias parameter and inherent belief estimators, it remains open whether this bound represents the actual convergence time or whether it may converge faster. Finally, the Interacting P\'{o}lya Urn model assumes that agents do not change their true belief. A more realistic model may assume that an agent might change her inherent belief (or her bias parameter). New analysis would be required for this, since a constant bias parameter is necessary in this work. 

These limitations provide interesting directions for future work. From a social science and psychology perspective, it would be interesting to compare the Interacting P\'{o}lya Urn model to the behavior of existing social networks; this may also reveal whether changes to the model may make it a more accurate reflection of real behavior (such as time-discounting, so that opinions expressed further in the past affect the social pressure less than recently-expressed opinions).
Furthermore, it remains open if and how the methods presented here can be generalized to variations of the model, for instance if the agents alter the social network over time in response to their social environment. Extending the estimators (and the analysis) to cases of partial information would also increase their practical relevance.
Showing lower bounds on the convergence time, and in particular closing the gap to the shown upper bounds, also remains open.

Another crucial area for future work is to consider how to efficiently \emph{intervene} in a social network in order influence \revisetwo{declared} opinions \revisetwo{(for a different interpretation of the model where declared opinions are actions and inherent beliefs are only biases)}; this stems from one of the key motivations behind the study of opinion dynamics, which is to help guide marketing and persuasion campaigns. 

Finally, the model can be extended to $k > 2$ different opinions, with each agent declaring $\psi(t) \in [k]$ at each time, and raises the question of how this would affect the problem of estimating inherent beliefs and biases.
}

%
\fi

\ifarxiv 
\else 
\section*{References}
\fi 

\bibliographystyle{resources/IEEEbib}
\bibliography{ref_social}




\ifarxiv 
\else 

\vspace{-3pc}

\begin{IEEEbiography}
[{\includegraphics[width=1in,height=1.25in,clip,keepaspectratio]{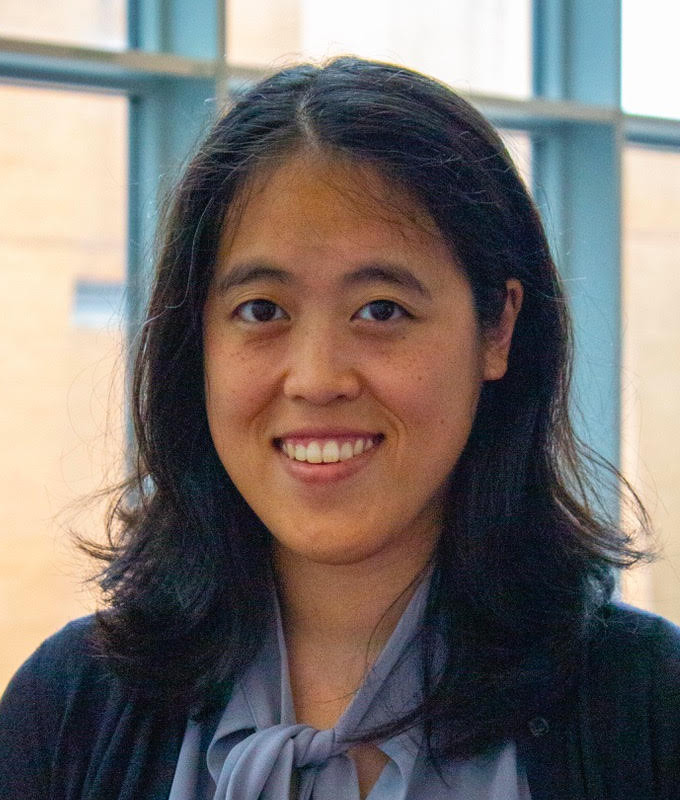}}]%
{Jennifer Tang} (Member, IEEE)
is a Postdoctoral Associate at the Institute of Data, Systems and Society (IDSS) and the Laboratory for Information and Decision Systems (LIDS) at MIT. She received her Ph.D and S.M in Electrical Engineering and Computer Science at MIT, and a B.S.E in Electrical Engineering from Princeton University. Her research interests include information theory, quantization and data compression, high-dimensional statistics and models for social dynamics and inference. 
\end{IEEEbiography}

\vspace{-3pc}

\begin{IEEEbiography}[{\includegraphics[width=1in,height=1.25in,clip,keepaspectratio]{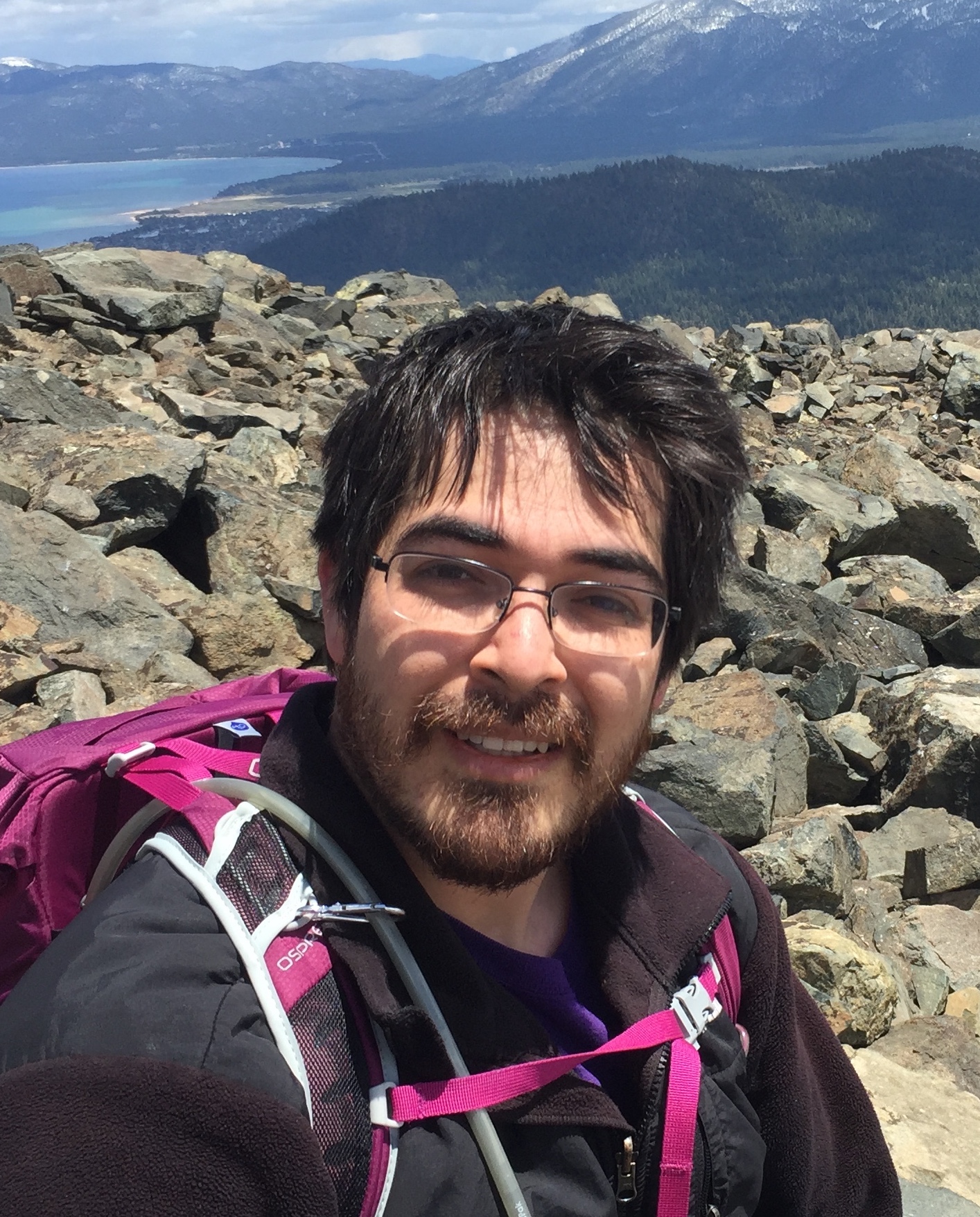}}]%
{Aviv Adler} (Member, IEEE)
is a Senior Engineer at Analog Devices Inc. Previously, he was a Postdoctoral Scholar at AUTOLab at UC Berkeley.
He received his Ph.D. and S.M. in the Department of Electrical Engineering and Computer Science at MIT as a member of LIDS, and his A.B. in Mathematics at Princeton University. Aviv's research interests include robotics, computational geometry, motion planning, optimization, complexity theory, and information theory.
\end{IEEEbiography}

\vspace{-3pc}

\begin{IEEEbiography}[{\includegraphics[width=1in,height=1.25in,clip,keepaspectratio]{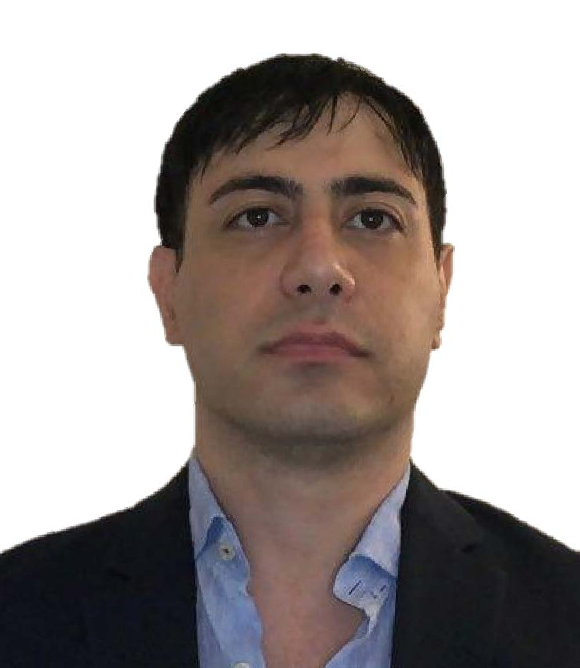}}]%
{Amir Ajorlou} (Member, IEEE)
 is a research scientist at the Laboratory for Information and Decision Systems (LIDS) at MIT. Previously, he was a postdoctoral research fellow at MIT and Penn from 2013 to 2016. He received his BS from Sharif University of Technology in Tehran, Iran and his PhD in electrical and computer engineering from Concordia University in Montreal, Canada, in 2013. He has been the recipient of several awards, including two gold medals in the International Mathematical Olympiad (IMO), Concordia University Doctoral Prize in Engineering and Computer Science, Governor General of Canada Academic Gold Medal, and the NSERC Postdoctoral Fellowship. He has served as an Associate Editor for the Conference Editorial Board of IEEE Control Systems Society from 2017 to 2022, and as Chair of the Control Systems Chapter of the IEEE Montreal Section in 2011. His current research uses a collection of tools and techniques from game theory, networks, microeconomics, optimization, applied probability theory, and control theory to model and analyze strategic decision making under uncertainty in complex social and economic networks. 
\end{IEEEbiography}

\vspace{-3pc}

\begin{IEEEbiography}[{\includegraphics[width=1in,height=1.25in,clip,keepaspectratio]{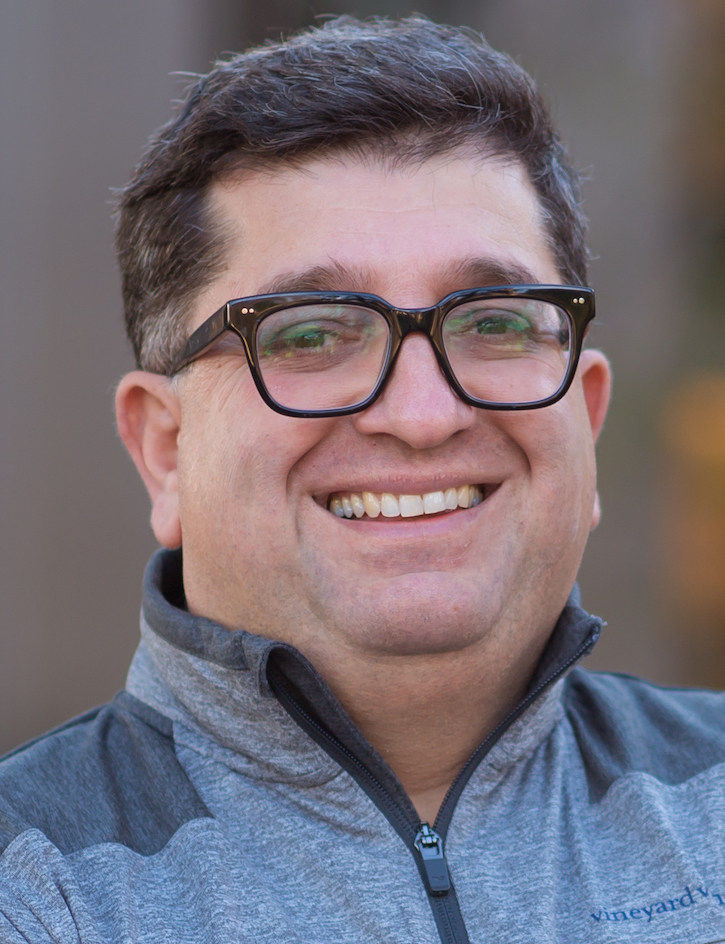}}]%
{Ali Jadbabaie} (Fellow, IEEE) is the JR East Professor and Head of the Department of Civil and Environmental Engineering at Massachusetts Institute of Technology (MIT), where he is also a core faculty in the Institute for Data, Systems, and Society (IDSS) and a Principal Investigator with the Laboratory for Information and Decision Systems. Previously, he served as the Director of the Sociotechnical Systems Research Center and as the Associate Director of IDSS which he helped found in 2015. 

He received a B.S. degree with High Honors in electrical engineering with a focus on control systems from the Sharif University of Technology, an M.S. degree in electrical and computer engineering from the University of New Mexico, and a Ph.D. degree in control and dynamical systems from the California Institute of Technology. He was a Postdoctoral Scholar at Yale University before joining the faculty at the University of Pennsylvania, where he was subsequently promoted through the ranks and held the Alfred Fitler Moore Professorship in network science in the Department of Electrical and Systems Engineering.   He is a recipient of a National Science Foundation Career Development Award, an US Office of Naval Research Young Investigator Award, the O. Hugo Schuck Best Paper Award from the American Automatic Control Council, and the George S. Axelby Best Paper Award from the IEEE Control Systems Society. He has been a senior author of several student best paper awards, in several conferences including ACC, IEEE CDC and IEEE ICASSP. He is an IEEE fellow, and the recipient of a Vannevar Bush Fellowship from the Office of Secretary of Defense. His research interests are broadly in systems theory, decision theory and control,  optimization theory, as well as multiagent systems, collective decision making in social and economic networks, and computational social science. 
\end{IEEEbiography}


\fi 

\end{document}
